\documentclass[opre,nonblindrev]{informs4Preprint}

\OneAndAHalfSpacedXI

\TheoremsNumberedThrough 
\EquationsNumberedThrough
\usepackage{xfrac}
\usepackage{xcolor}
\usepackage{bbm}
\usepackage{booktabs}
\usepackage{tikz}
\usepackage{subcaption}
\usepackage{float}
\usepackage[algo2e,ruled,vlined,linesnumbered]{algorithm2e}

\usepackage{hhline}

\usepackage{changepage}

\newenvironment{myprocedure}[1][htb]
{\begin{algorithm2e}[#1]}{\end{algorithm2e}}

\DeclareRobustCommand{\mybox}[2][gray!10]{\begin{tcolorbox}[   left=0pt,
        right=0pt,
        top=0pt,
        bottom=0pt,
        colback=#1,
        colframe=#1,
        enlarge left by=0mm,
        boxsep=15pt,
        arc=5pt,outer arc=5pt,
        ]
        #2
\end{tcolorbox}
}

\DeclareRobustCommand{\myboxtwo}[2][gray!15]{
\begin{tcolorbox}[ 
        colback=white,      colframe=gray,  
        boxrule=0.2pt,      arc=2pt,outer arc=2pt,
        left=12pt,
        right=12pt,
        top=5pt,
        bottom=5pt,
        width=1.07\linewidth,
        enlarge left by=-0.55cm,
        before upper=\renewcommand{\baselinestretch}{1.3}\selectfont,
        after upper=\normalfont
        ]
 #2
 \end{tcolorbox}
}

\usepackage{thm-restate}

\usepackage[breakable, skins]{tcolorbox}
\DeclareRobustCommand{\mybox}[2][gray!10]{\begin{tcolorbox}[
        left=0.5pt,
        right=0.5pt,
        top=0.5pt,
        bottom=0.5pt,
        colback=#1,
        colframe=#1,
        width=\dimexpr\textwidth\relax, 
        enlarge left by=1mm,
        boxsep=3pt,
        arc=2pt,outer arc=2pt,
        code={\OneAndAHalfSpacedXII}
        ]
        #2
        
\end{tcolorbox}
}
\usepackage{csquotes}
\makeatletter
\renewenvironment*{displayquote}
  {\begingroup\setlength{\leftmargini}{0.2cm}\csq@getcargs{\csq@bdquote{}{}}}
  {\csq@edquote\endgroup}
\makeatother

\usepackage{mathtools, bm, amsmath}

\newcommand{\openbox}{\hbox{$\square$}}

\definecolor{cornellred}{rgb}{0.7, 0.11, 0.11}
\definecolor{maroon}{rgb}{0.52, 0, 0}
\definecolor{dgreen}{rgb}{0.0, 0.5, 0.0}
\definecolor{ballblue}{rgb}{0.13, 0.67, 0.8}
\definecolor{royalbamsmathlue(web)}{rgb}{0.25, 0.41, 0.88}
\definecolor{bleudefrance}{rgb}{0.19, 0.55, 0.91}
\definecolor{royalazure}{rgb}{0.0, 0.22, 0.66}
\usepackage{hyperref}
\hypersetup{
	colorlinks = true,
	linkcolor=royalazure,
	urlcolor=royalazure,
	citecolor=royalazure,
	linkbordercolor = {white}
}
\usepackage{cleveref}
\usepackage{enumitem}
\usepackage{multirow}

\usepackage{pgf,tikz,pgfplots}
\pgfplotsset{compat=1.15}
\usetikzlibrary{arrows}
\usepackage{pgfplots}
\usetikzlibrary{patterns}
\usepgfplotslibrary{fillbetween}
\usetikzlibrary{intersections}

\usepackage{fix-cm}

\usetikzlibrary{arrows, decorations.markings}

\tikzstyle{vecArrow} = [thick, decoration={markings,mark=at position
	1 with {\arrow[semithick]{open triangle 60}}},
	double distance=1.4pt, shorten >= 5.5pt,
	preaction = {decorate},
postaction = {draw,line width=1.4pt, white,shorten >= 4.5pt}]
\tikzstyle{innerWhite} = [semithick, white,line width=1.4pt, shorten >= 4.5pt]

	\usepackage[authoryear,round]{natbib}

\makeatletter
\newcommand{\subalign}[1]{\vcenter{\Let@ \restore@math@cr \default@tag
    \baselineskip\fontdimen10 \scriptfont\tw@
    \advance\baselineskip\fontdimen12 \scriptfont\tw@
    \lineskip\thr@@\fontdimen8 \scriptfont\thr@@
    \lineskiplimit\lineskip
    \ialign{\hfil$\m@th\scriptstyle##$&$\m@th\scriptstyle{}##$\hfil\crcr
      #1\crcr
    }}}
\makeatother

\usepackage{makecell}
\usepackage{threeparttable}

	\providecommand{\given}{}
	\DeclarePairedDelimiterX{\set}[1]\{\}{\renewcommand\given{\nonscript\:\delimsize\vert\nonscript\:\mathopen{}}#1}
	\let\Pr\relax
	\DeclarePairedDelimiterXPP{\Pr}[1]{\mathbb{P}}[]{}{\renewcommand\given{\nonscript\:\delimsize\vert\nonscript\:\mathopen{}}#1}
	\DeclarePairedDelimiterXPP{\Ex}[1]{\mathbb{E}}[]{}{\renewcommand\given{\nonscript\:\delimsize\vert\nonscript\:\mathopen{}}#1}

\newcolumntype{P}[1]{>{\centering\arraybackslash}c{#1}}

\newcommand{\xhdr}[1]{\vspace{2mm} \noindent{\bf #1}~}

\newcommand*{\rom}[1]{\expandafter\romannumeral #1}
\newcommand{\Rom}[1]{\uppercase\expandafter{\romannumeral #1\relax}}

\newcommand{\naturals}{\mathbb{N}}

\newcommand{\OPT}{\textup{\textsc{Opt}}}
\newcommand{\ALG}{\textup{\textsc{Alg}}}

\newcommand{\WithinClass}{\textup{\textsc{WithinClass}}}
\newcommand{\BetweenClasses}{\textup{\textsc{BetweenClasses}}}

\newcommand{\SPipelineSchedule}{\textsc{Staggered-Pipeline-Schedule}}
\newcommand{\SimPipelineSchedule}{\textsc{Simultaneous-Schedule}}

\newcommand{\minparal}{{\paral}_{\min}}

\newcommand{\initialLen}{s}
\newcommand{\responseLen}{o}
\newcommand{\responseLenEst}{\tilde{\responseLen}}

\newcommand{\job}{j}                     \newcommand{\Jobs}{\mathcal{J}}          \newcommand{\jobIdx}{i}                  \newcommand{\RankIdx}{r}                 

\newcommand{\paral}{k}                   \newcommand{\sliceLen}{\tau}             \newcommand{\sliceLenHat}{\widehat{\sliceLen}}
\newcommand{\sliceLenP}[1]{\sliceLen_{#1}}
\newcommand{\sliceLenHatP}[1]{\sliceLenHat_{#1}}
\newcommand{\paralSliceLen}{\paral^*(\sliceLen,\initialLen)}
\newcommand{\paralSliceLenP}{\paral^*(\sliceLenP{\PhaseIdx},\initialLen)}
\newcommand{\kvmem}{M}

\newcommand{\startTime}{t^{\mathrm{(start)}}}               \newcommand{\compTime}{t^{\mathrm{(end)}}}                

\newcommand{\finishTime}{c}

\newcommand{\startTimeOf}[1]{\startTime_{#1}}
\newcommand{\compTimeOf}[1]{\compTime_{#1}}

\newcommand{\MemPeak}{\textup{\textsf{Peak}}}

\newcommand{\FTime}[2][]{{\textsc{FlowTime}}\ifthenelse{\not\equal{}{#1}}{_{#1}}{}\!\left[{\def\givenn{\middle|}#2}\right]}
\newcommand{\C}{C}

\newcommand{\GeoSALG}{\textsc{Geometric-Slicing-Algorithm}}

\newcommand{\GeoBALG}{\textsc{Geometric-Batching-Algorithm}}
\newcommand{\GBA}{\textsc{GBA}}
\newcommand{\UpperGBA}{\overline{\GBA}}

\newcommand{\GSA}{{\textsc{GSA}}}
\newcommand{\SPS}{\textsc{SPS}}

\newcommand{\SimPS}{\textsc{SimS}}
\newcommand{\GBAD}{\textsc{GBA-D}}

\newcommand{\GSAH}{\textsc{GSA-Spec}}
\newcommand{\MCSF}{\textsc{MC-SF}}
\newcommand{\AMin}{\textsc{A-Min}}
\newcommand{\vLLM}{\textsc{vLLM}}

\newcommand{\remainJobs}{\mathcal{T}}
\newcommand{\responseScalar}{\alpha}
\newcommand{\sliceLowerBound}{\beta}

\newcommand{\approxratio}{\Gamma}
\newcommand{\localratio}{\gamma}

\newcommand{\ActiveBatch}{\mathcal{B}}
\newcommand{\curProgress}{u}

\newcommand{\Area}{A}
\newcommand{\AreaLB}{A}
\newcommand{\NumberJobs}{n}
\newcommand{\PhaseIdx}{p}
\newcommand{\PhaseIdxSecond}{q}
\newcommand{\PhaseIdxThird}{v}
\newcommand{\PhaseIdxFourth}{t}
\newcommand{\MaxPhaseIdx}{\ell}

\newcommand{\Ladder}{Q}
\newcommand{\ClassPrefix}{S}
\newcommand{\PhasePrefix}{T}
\newcommand{\Spillover}{\Delta}
\newcommand{\NumJobsGe}[1]{N_{\ge #1}}
\newcommand{\NumJobsGt}[1]{N_{> #1}}
\newcommand{\NumJobsClass}{n}
\newcommand{\JobsClass}{J}

\newcommand{\TmpFunc}{f}

\newcommand{\promptRatio}{\rho}

\newcommand{\condition}{\,\mid\,}

\newcommand{\prob}[2][]{\text{Pr}\ifthenelse{\not\equal{}{#1}}{_{#1}}{}\!\left[{\def\givenn{\middle|}#2}\right]}
\newcommand{\expect}[2][]{\mathbb{E}\ifthenelse{\not\equal{}{#1}}{_{#1}}{}\!\left[{\def\givenn{\middle|}#2}\right]}

\newcommand{\tparen}{\big}
\newcommand{\tprob}[2][]{\text{Pr}\ifthenelse{\not\equal{}{#1}}{_{#1}}{}\tparen[{\def\given{\tparen|}#2}\tparen]}
\newcommand{\texpect}[2][]{\mathbb{E}\ifthenelse{\not\equal{}{#1}}{_{#1}}{}\tparen[{\def\given{\tparen|}#2}\tparen]}

\newcommand{\sprob}[2][]{\text{Pr}\ifthenelse{\not\equal{}{#1}}{_{#1}}{}[#2]}
\newcommand{\sexpect}[2][]{\mathbb{E}\ifthenelse{\not\equal{}{#1}}{_{#1}}{}[#2]}

\newcommand{\indicator}[1]{{\mathbbm{1}\left\{ #1 \right\}}}

\makeatletter
\DeclareRobustCommand{\qed}{\ifmmode \else \leavevmode\unskip\penalty9999 \hbox{}\nobreak\hfill
  \fi
  \quad\hbox{\qedsymbol}}
\newcommand{\qedsymbol}{\openbox}
\renewenvironment{proof}[1][\proofname]{\par
  \normalfont
  \topsep6\p@\@plus6\p@ \trivlist
  \item[\hskip\labelsep\itshape
    #1.]\ignorespaces
}{\qed\endtrivlist
}
\newcommand{\proofname}{Proof}
\makeatother

\MANUSCRIPTNO{}

\begin{document}

\RUNAUTHOR{Feng, Yang, Zhang}

\RUNTITLE{Competitive Non-Clairvoyant KV-Cache Scheduling for LLM Inference}
\TITLE{Competitive Non-Clairvoyant KV-Cache Scheduling for LLM Inference}

\ARTICLEAUTHORS{\AUTHOR{Yiding Feng}
\AFF{Hong Kong University of Science and Technology, China, \EMAIL{ydfeng@ust.hk}}
\AUTHOR{Zonghan Yang}
\AFF{Shanghai Jiao Tong University, China, \EMAIL{fstqwq@sjtu.edu.cn}}
\AUTHOR{Yuhao Zhang}
\AFF{Shanghai Jiao Tong University, China, \EMAIL{zhang\_yuhao@sjtu.edu.cn}}
} 

\ABSTRACT{Large Language Model (LLM) inference presents a unique scheduling challenge due to the Key-Value (KV) cache, where a job's memory footprint grows linearly with the number of decoded tokens. This growth couples scheduling decisions with feasibility: a scheduler must minimize latency under a hard memory budget, yet the response lengths of requests are inherently unknown. While recent works have explored this problem either assuming clairvoyance---exact knowledge of response lengths---or relying on machine-learned predictions, obtaining robust performance guarantees without any prior knowledge of job sizes remains a theoretically fundamental and practically important open problem.

In this work, we propose the Geometric Slicing Algorithm (GSA), the non-clairvoyant policy to achieve \emph{the first constant competitive ratio for this problem in the offline batch setting}. GSA manages uncertainty through a geometric phase structure that periodically restarts jobs to bound memory exposure, combined with a staggered pipeline mechanism that enables high concurrency by smoothing aggregate memory consumption. We prove that GSA achieves a competitive ratio of at most $61.92$ for general instances, improving to $32$ in the large-memory regime. Our algorithmic framework also yields a clairvoyant counterpart, the Geometric Batching Algorithm (GBA), which achieves an approximation ratio of $10.67$ for general instances and $6.75$ in the large-memory regime---significantly improving upon the best previously known bound of over $9000$. Numerical experiments on real request traces demonstrate that our algorithms perform robustly while preserving these worst-case guarantees. }


\maketitle

\newpage

\section{Introduction}
\label{sec:intro}

\begin{figure}[b]
    \centering
    \includegraphics[width=0.8\linewidth]{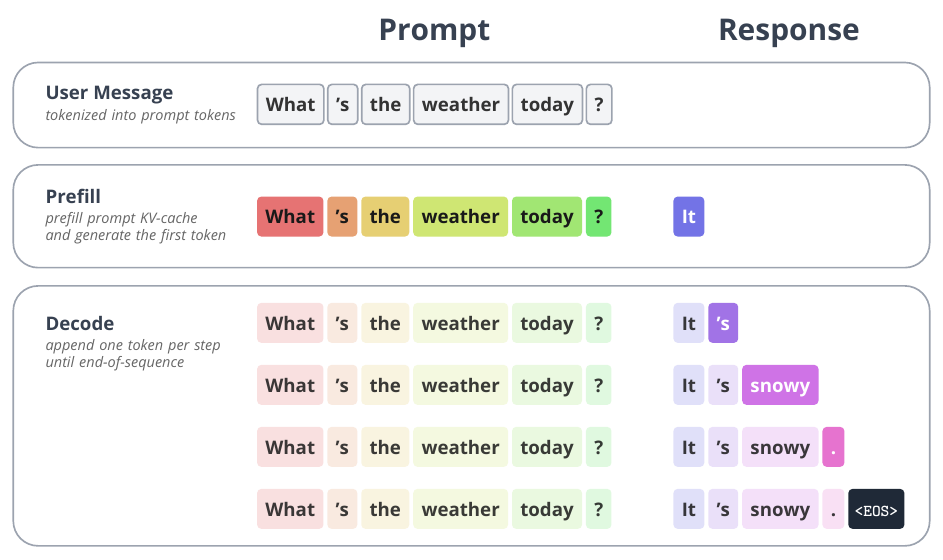}
    \caption{The diagram illustrates the LLM inference process flowing from top to bottom. First, the user message is tokenized. Next, the \textit{prefill} phase processes the prompt to generate the first decode token. Then, the \textit{decode} phase generates output tokens auto-regressively until the End-Of-Sequence token. Note that the KV-cache (colored context) grows linearly with each step.}
    \label{fig:llm_inference_diagram}
\end{figure}

Large Language Models (LLMs) have emerged as the cornerstone of modern AI, powering chat assistants, search engines, and developer tools~\citep{BMRSKDNSSAA20,OpenAI-GPT4-23,Microsoft-Copilot-Web,Google-Gemini-Web}.  
While training remains a massive undertaking, the inference phase increasingly dominates operational costs---a trend amplified by the rise of reasoning-enhanced models that trade extensive Chain-of-Thought computation for improved capabilities~\citep{W0SBIXCLZ22,OpenAI-o1-24,DeepSeek-25}.  
These escalating resource demands strain deployment infrastructure significantly, impacting not only capital expenditures but also environmental sustainability
~\citep{SGM19,LJS24,LYIR25}.  
Consequently, optimizing system efficiency without compromising quality-of-service requirements---such as latency---has become an operational imperative.

Virtually all foundation LLMs are built upon the Transformer architecture~\citep{VSPUJGKP-17}, whose self-attention mechanism computes dependencies across input tokens. To avoid quadratic recomputation at each decoding step, modern inference systems universally employ a \emph{Key-Value (KV) cache} to store intermediate attention states~\citep{KLZSZYGZS-23,ZYXS0YCKSGB24,HKMMSKG24}.  
While this optimization reduces computational complexity to linear time, it introduces a critical resource trade-off: execution progress becomes tightly coupled with GPU memory consumption. 
This coupling has shifted attention toward systems-level optimization, especially KV-cache management.

As illustrated in \Cref{fig:llm_inference_diagram}, when a user request arrives, it is tokenized and processed by the LLM in two phases.  
The \emph{decode} phase then produces subsequent tokens auto-regressively, appending a new KV pair to the cache for each generated token.  
Consequently, the memory footprint of a request is \emph{not static} but grows linearly with its response length.

This dynamic memory growth---where a request's memory footprint increases linearly with its accumulated processing time---introduces a first fundamental challenge for LLM inference scheduling: \emph{temporal coupling between concurrency and feasibility}. Unlike classical scheduling models that assume static resource requirements, the KV cache in LLMs expands as decoding progresses. Consequently, a batch of jobs that fits within memory at the start of execution may inevitably violate memory constraints later. This forces the scheduler to balance high concurrency against the risk of memory overflows, which trigger costly preemptions under kill-and-restart semantics---the core tension in efficient LLM serving.

Compounding this issue is a second critical challenge: \emph{non-clairvoyance}. 
{Because termination is governed by model-generated stop conditions rather than an externally specified job size, the final response length is not determined a priori and only gradually reveals itself during execution.\footnote{Some practical systems introduce response-length predictors to inform batching, admission control, or memory reservation. However, these predictors provide only noisy, workload-dependent estimates rather than oracle knowledge, and they lack guarantees that would eliminate non-clairvoyance~\citep{ZRXL0Y23,JW0W23,FZSQS024}. This direction has been explored in recent prior work~\citep{CYZ-25}. See more discussion in related work (\Cref{apx:related work}).} This uncertainty renders the processing time of each job unpredictable until completion, placing the problem squarely in the non-clairvoyant setting of classical scheduling theory. The lack of foresight harms average latency, as traditional clairvoyant strategies become brittle. 
Moreover, because the KV cache grows dynamically with generated tokens, uncertainty in response length directly translates into uncertainty in future memory footprints---a difficulty absent in classical models. Thus, the scheduler must simultaneously manage both temporal resource coupling and incomplete information, navigating a doubly constrained trade-off between  latency and memory safety.

Motivated by these challenges and the critical importance of efficient LLM inference scheduling, we ask:
\myboxtwo{
\begin{displayquote}
\emph{Can we design effective batching policies for LLM inference requests that remain robust to the risk of infeasibility caused by dynamic memory growth, even in non-clairvoyant environments?}
\end{displayquote}
}

To address this question, we adopt the offline batch scheduling model studied in recent literature~\citep[e.g.,][]{JJMMPZ-25,WYZ-25,CYZ-25}. 
In this model, all inference requests arrive simultaneously at time zero and are processed on a single GPU under a hard KV-cache memory budget. Each request consists of a fixed-length prompt followed by an unknown number of output tokens (i.e., response length) generated autoregressively, creating a non-clairvoyant environment. Since our focus is exclusively on scheduling the decode phase, the prefill phase---which initializes the KV cache from the prompt and generates the first token---is assumed to be completed in a single time step (i.e., the duration of decoding one token)\footnote{This assumption is justified as practical systems often handle the prefill phase separately via prefill/decode disaggregation~\citep{ZLCHZL0024,CPZSGMFB25}.}.

During decoding, a request's memory footprint grows linearly with its progress: each generated token appends to the KV cache, increasing total memory consumption. At every round, the scheduler must decide which subset of unfinished requests to process as a batch on the GPU. If executing the chosen batch would violate the memory budget, the scheduler is forced to preempt (i.e., kill) some currently active requests, discarding all their partial progress, which reflects the prohibitive cost of state swapping in real systems. Critically, the scheduler operates in a non-clairvoyant setting: it does not know how long each request will run until it completes.

The objective is to minimize \emph{total flow time}, i.e., the sum of job completion times, which directly captures end-to-end latency in the offline batch setting. We evaluate non-clairvoyant algorithms via their \emph{competitive ratio}: the worst-case ratio of their total flow time to that of an optimal clairvoyant scheduler with full knowledge of all response lengths. For completeness, we also analyze the clairvoyant counterpart of our algorithm, which serves as a structural proxy in our competitive analysis.

\subsection{Our Contributions}

In this paper, we provide compelling answers to the questions posed above. In a nutshell, we present the \emph{first $O(1)$-competitive algorithm} for the non-clairvoyant offline batch scheduling model under dynamic KV-cache memory constraints, without any assumptions on the input instance. Moreover, under the practically relevant \emph{large-memory regime}, where the KV-cache budget is sufficiently large relative to prompt and response lengths,\footnote{The large-memory assumption is well-motivated in LLM inference and cloud computing settings, and is commonly adopted in prior work on LLM scheduling and online resource allocation (see, e.g., \citealp{KP-00,MSVV-07,GNR-14,MS-20,CYZ-25}).} we obtain an improved competitive ratio.
More formally, our first main result is as follows:

\mybox{\emph{\textbf{Main Result (I)} We propose a non-clairvoyant, polynomial-time scheduling algorithm
(\Cref{alg:gsa}). For general instances, it achieves a competitive ratio of at most $61.92$; for large-memory instances, this improves to $32$.}}

\begin{table}[ht]
\centering
\footnotesize
\renewcommand\arraystretch{1.2}

\begin{threeparttable}
\caption{Summary of results and comparison with prior work}
\label{tab:results}

\begin{tabular}{cccc}
  \toprule
  \multicolumn{2}{c}{\textbf{Setting}} &
  \qquad\qquad\qquad \textbf{\color{blue}This Work} \qquad\qquad\qquad \quad&
  \textbf{\color{maroon}Prior Work} \\
  \midrule

\multirow{4}{*}{\textbf{Non-Clairvoyant}}
    & \multirow{2}{*}{General Instances}
    & {\color{blue}\textbf{61.92}}
    & \multirow{2}{*}{\color{maroon}(unknown)} \\
    & 
    & {\color{blue}(Theorem.~\ref{thm:GSA})}
    &  \\
    \cmidrule(l{0.5em}r{0.5em}){2-4}
    & Large-Memory
    & {\color{blue}\textbf{32}}
    & {\color{maroon}$O(\log M)$} \\
    & Instances
    & {\color{blue}(Theorem.~\ref{thm:GSA})}
    & {\color{maroon}\citep{CYZ-25}} \\
  \midrule

\multirow{4}{*}{\textbf{Clairvoyant}}
    & \multirow{2}{*}{General Instances}
    & {\color{blue}\textbf{10.67}}
    & {\color{maroon}9216} \\
    &
    & {\color{blue}(Theorem.~\ref{thm:GBA})}
    & {\color{maroon}\citep{JJMMPZ-25}} \\
    \cmidrule(l{0.5em}r{0.5em}){2-4}
    & Large-Memory
    & {\color{blue}\textbf{6.75}}
    & {\color{maroon}$48^{\ddagger}$} \\
    & Instances
    & {\color{blue}(Theorem.~\ref{thm:GBA})}
    & {\color{maroon}\citep{WYZ-25}} \\
  \midrule

\multirow{4}{*}{\makecell[c]{\textbf{Clairvoyant}\\\textbf{Identical Jobs}}}
    & \multirow{2}{*}{General Instances}
    & {\color{blue}\textbf{2}}
    & {\color{maroon}4} \\
    &
    & {\color{blue}(Theorem.~\ref{thm:staggered:identical job})}
    & {\color{maroon}\citep{JJMMPZ-25}} \\
    \cmidrule(l{0.5em}r{0.5em}){2-4}
    & {Large-Memory}
    & {\color{blue}\textbf{1$^{\dagger}$}}
    & {\color{maroon}4} \\
    & {Instances}
    & {\color{blue}(Theorem.~\ref{thm:staggered:identical job})}
    & {\color{maroon}\citep{JJMMPZ-25}} \\
  \bottomrule
\end{tabular}

\begin{tablenotes}[para, flushleft]
  \item \emph{\underline{Note}:  
  Numbers represent competitive ratios (for non-clairvoyant algorithms) and approximation ratios (for clairvoyant algorithms).
  All large-memory instances assume that the KV-cache memory budget $\kvmem \to \infty$ while the prompt lengths and response lengths remain fixed.  
    $\dagger$ additionally requires the number of jobs $n \to \infty$.  
    $\ddagger$ continues to hold even under heterogeneous prompt lengths \citep{WYZ-25}.}
\end{tablenotes}
\end{threeparttable}
\end{table}
 
Our algorithmic framework also yields better guarantees in the clairvoyant setting.
In particular, we show that a natural clairvoyant counterpart of our algorithm---designed with full knowledge of request response lengths---yields even stronger guarantees, significantly improve upon previous results. Formally, our second main result is:

\mybox{\emph{\textbf{Main Result (II)} We propose a clairvoyant, polynomial-time scheduling algorithm
(\Cref{alg:gba}), which serves as the structural analog of \Cref{alg:gsa}. For general instances, it achieves an approximation ratio of at most $10.67$; for large-memory instances, this improves to $6.75$. When all requests have identical response lengths, the ratio further improves to $2$ in general and asymptotically approaches $1$ in the large-memory limit.}}

We summarize all our results and compare them with prior work in \Cref{tab:results}.  
Although LLM inference and its optimization under KV-cache constraints are relatively recent, there has been a rapidly growing algorithmic literature on scheduling for this setting. Among the various models studied, the {offline batch scheduling model}---which we adopt---is one of the most prominent due to its practical relevance and mathematical tractability.

In the {non-clairvoyant} setting, \citet{CYZ-25} proposed a scheduling policy whose analysis is tailored to the large-memory instances (where the memory capacity $\kvmem \to \infty$). They establish the first bounded competitive ratio of $O(\log \kvmem)$.
In contrast, as our first main result, we present the first constant-competitive algorithm for general instances, achieving a competitive ratio of at most $61.92$ without any asymptotic assumptions.\footnote{In \Cref{sec:intro:technique} and \Cref{subsec:main alg}, we explain why constant competitiveness cannot be achieved by certain classic heuristics---including the algorithm of \citet{CYZ-25} and how it motivates one key algorithmic ingredient in our algorithm.}

For the {clairvoyant} setting, constant-approximation algorithms have been previously designed. \citet{JJMMPZ-25} propose a clairvoyant scheduler and prove an approximation ratio of $9216$.\footnote{While \citet{JJMMPZ-25} state their result as an $O(1)$ approximation (Theorem 4.3), the explicit constant $9216$ arises from combining multiplicative factors in their structural lemmas ($1536$ from Lemma 4.4 and $6$ from Lemma 4.7).}  
For large-memory instances, this was later improved to $48$ by \citet{WYZ-25}.\footnote{\label{footnote:WYZ 48 approx}\citet{WYZ-25} establishes the approximation ratio of 48 even under heterogeneous prompt lengths.} 
Our second main result---a clairvoyant structural analog of our non-clairvoyant algorithm---significantly improves these guarantees: we reduce the approximation ratio from $9216$ to $10.67$ for general instances and from $48$ to $6.75$ for large-memory instances.\footnote{In \Cref{sec:intro:technique} and \Cref{subsec:staggered}, we detail how our staggered pipeline design enables smoother memory usage and higher concurrency, which eventually translates into tighter approximation bounds.}

We further validate our theoretical insights through numerical experiments (\Cref{apx:numerical}). Using both synthetic workloads and real-world traces from the LMSYS-Chat-1M dataset, we observe that the structural properties of our geometric scheduling framework---such as staggered execution and disciplined phase-based preemption---can lead to meaningful practical improvements. In our evaluations, our algorithms perform favorably against state-of-the-art baselines across a range of settings, and heuristic variants ({\GBAD} and {\GSAH}) demonstrate robust behavior while preserving worst-case guarantees.

\subsection{Overview of the Techniques}
\label{sec:intro:technique}

As discussed above, we develop a unified algorithmic framework that achieves strong performance guarantees in both non-clairvoyant and clairvoyant settings. Realizing this goal requires overcoming four key challenges:  
(i) maintaining batching feasibility effectively under dynamic KV-cache memory constraints;  
(ii) handling heterogeneity in request response lengths;  
(iii) coping with uncertainty about response lengths (i.e., non-clairvoyance); and  
(iv) analyzing performance against an optimal clairvoyant scheduler that is difficult to characterize explicitly.

Challenges (i) and (ii) arise in both clairvoyant and non-clairvoyant settings, but are exacerbated by challenge (iii)---the lack of knowledge about request sizes---which makes effective batching significantly harder. Challenge (iv) is primarily analytical: because the optimal clairvoyant schedule has no simple closed form, standard competitive or approximation analysis techniques do not apply directly.

Our framework combines two key ingredients: \emph{staggered pipelining} and \emph{geometric slicing}.  
Staggered pipelining tackles challenge (i) by enabling high concurrency under dynamic memory constraints.  
Geometric slicing addresses challenges (ii) and (iii) by grouping jobs into geometrically spaced classes based on estimated response lengths, balancing load and remaining robust to estimation errors.

To handle the analytical difficulty of challenge (iv), we introduce a novel \emph{memory-time area} perspective. This geometric view yields a tractable bound on the optimal clairvoyant scheduler and provides a common metric to compare our algorithms against the optimum---effectively bridging scheduling decisions and resource consumption for competitive and approximation analyses.

Below, we discuss each of these components and their integration into our algorithms in greater detail.

\xhdr{Smoothing Memory Footprints via Staggered Pipeline.}  
A common approach used in both standard heuristics (e.g., First-Come-First-Served or Shortest-Job-First) and policies in prior  work~\citep[e.g.,][]{KLZSZYGZS-23,ZCSLTLZWXGZSZHZ24,JJMMPZ-25,CYZ-25} is \emph{simultaneous batching}: packing as many requests as possible into a single batch and executing them in lockstep. This strategy is well-motivated in the classic flow-time minimization setting (where memory consumption is static) and is indeed optimal for identical requests under that model.
However, simultaneous batching is inherently memory-inefficient in the KV-cache setting. Because all requests in a batch progress through decoding at the same rate, they reach their peak memory usage simultaneously. To avoid overflow, the scheduler must restrict the batch size based on this collective peak, resulting in a ``sawtooth'' memory profile where a significant portion of the memory budget remains underutilized on average.

Motivated by this inefficiency, we propose {\SPipelineSchedule} ({\SPS}), a core scheduling subroutine that forms the foundation of our main algorithms. Despite its simplicity, \emph{{\SPS} provides a principled approach to smoothing memory footprint over time, thereby enabling significantly higher concurrency under a fixed memory budget.} As such, it may be of independent interest for broader memory-constrained scheduling problems.

More precisely, {\SPS} is designed to process a set of requests with identical response lengths---or, more generally, a fixed time slice, where any request not completed within its allotted slice is killed. Rather than starting all requests simultaneously, {\SPS} intentionally offsets their start times to ``smooth out'' aggregate memory consumption. By staggering execution so that different requests reach their peak memory usage at distinct rounds, {\SPS} sustains a much higher degree of concurrency within the same memory budget.

Visual comparison between {\SPS} and simultaneous batching is provided in \Cref{example:schedule comparison}, \Cref{fig:schedule comparison}, and \Cref{fig:exp-uniform}. Thanks to its higher concurrency, {\SPS} achieves a better approximation ratio even for identical requests, and is asymptotically optimal in the large-memory regime. When combined with our other algorithmic ingredients, this improvement translates directly into stronger guarantees for more general settings.

\xhdr{Kill Two Birds with One Stone: Geometric Slicing for Effective Grouping.}  
While {\SPS} enables high concurrency for requests with identical response lengths, it does not address heterogeneity in response lengths or the uncertainty inherent in the non-clairvoyant setting. We resolve both challenges with a single idea: \emph{geometric slicing}, which groups requests into geometrically spaced classes based on their true response lengths (in the clairvoyant setting) or estimated lengths (in the non-clairvoyant setting).

We first describe the clairvoyant case, which is simpler. Our algorithm {\GeoBALG} ({\GBA}, see \Cref{alg:gba}) is parameterized by a scaling factor $\responseScalar \in (1, \infty)$. It partitions all requests into classes $\PhaseIdx = 0, 1, \dots$ according to geometrically increasing time scales:
\[
\sliceLenHatP{\PhaseIdx} = \responseScalar^{\PhaseIdx} \cdot \sliceLenHatP{0},
\]
where $1 \leq \sliceLenHatP{0} \leq \responseScalar$ is an initial time slice. For each class $\PhaseIdx$, {\GeoBALG} invokes {\SPS} with time slice $\lceil \sliceLenHatP{\PhaseIdx} \rceil$, which---by construction---is sufficient to complete all class-$\PhaseIdx$ requests. As in classical flow-time minimization, processing shorter requests first improves total flow time; invoking {\SPS} per class ensures high concurrency while respecting the KV-cache memory constraint.

Our non-clairvoyant algorithm {\GeoSALG} ({\GSA}, see \Cref{alg:gsa}) follows the same geometric structure. It uses the same parameters $\responseScalar$ and $\sliceLenHatP{0}$, and executes {\SPS} with time slice $\lceil \sliceLenHatP{\PhaseIdx} \rceil$ in phase $\PhaseIdx$. The key difference is that, without knowledge of true response lengths, it cannot isolate class-$\PhaseIdx$ requests. Instead, it runs {\SPS} on all remaining (unfinished) requests in each phase. Although this introduces overhead compared to the clairvoyant counterpart, our analysis shows this overhead is tightly controlled.

In fact, the choice of $\responseScalar$ entails a fundamental tradeoff. A small $\responseScalar$ yields many fine-grained classes, which may contain too few requests to fully utilize {\SPS}, thereby reducing concurrency. It also increases overhead in the non-clairvoyant setting, as {\GSA} performs more frequent kill-and-restart operations across a larger number of phases. Conversely, a large $\responseScalar$ produces coarse classes where the allocated time slice significantly exceeds the actual response lengths of many requests, leading to wasted memory and compute resources. We optimize $\responseScalar$ for different regimes to achieve the competitive and approximation ratios reported in \Cref{tab:results}. In practice, both $\responseScalar$ and $\sliceLenHatP{0}$ are tunable parameters (see \Cref{apx:numerical} for further discussion).

It is also worth highlighting that {\GSA} employs a counter-intuitive \emph{aggressive preemption} strategy: it systematically terminates all unfinished requests at the end of each phase. This contrasts sharply with many heuristics and policies in prior work, which avoid preemption unless the memory budget is exceeded---aiming to minimize restarts and associated wasted computation. Moreover, when preemption is unavoidable, those approaches typically prefer killing requests with shorter processing times. While intuitive, this strategy can severely degrade performance under the flow-time objective: long-running requests block memory for extended periods, limiting parallelism and inflating total flow time. As illustrated in \Cref{example:long-job-trap}, this distinction is crucial---{\GSA}'s disciplined phase-based preemption enables it to outperform such heuristics by maintaining higher concurrency and smoother memory utilization.

\xhdr{Competitive/Approximation Analysis via the Memory-Time Area Perspective.}  
A central challenge in our analysis is the intractability of the optimal scheduler. Due to the dynamic memory footprint and the hard KV-cache constraint, the optimal schedule can be highly complex. While it can be formulated as an integer program---and some prior work \citep[e.g.,][]{JJMMPZ-25} uses its linear programming relaxation and dual as a performance upper bound---this approach remains technically cumbersome. 

In contrast, we take a fundamentally different approach based on \emph{memory-time area}, which is defined as the total KV-cache resources consumed by a request over its execution, combining fixed prompt memory and linearly growing decode cache (see \Cref{def:area} and the colored triangles/trapezoids in \Cref{fig:schedule comparison}).

This perspective yields a clean accounting principle: any feasible schedule's total memory-time area cannot exceed the area supplied by the memory budget over time. Crucially, this enables two insights. First, under the relaxed constraint that only total area must be respected (a necessary condition for feasibility), the optimal schedule becomes simple: process jobs in non-decreasing order of response length. This provides a tractable upper bound on the true optimum.

Second, the area view allows us to precisely quantify the packing efficiency of {\SPS}. We show that both the clairvoyant {\GBA} (which applies {\SPS} per class) and {\SPS} itself achieve approximately optimal area utilization, directly yielding strong approximation guarantees. Finally, leveraging the structural similarity between {\GSA} and {\GBA}, we bound the non-clairvoyant overhead by a constant factor, establishing {\GSA}'s constant competitive ratio.

\subsection{Further Related Work}
\label{apx:related work}
We further review related work in the theoretical literature that provides foundational techniques for our algorithmic framework.

\xhdr{Packing Problems.}
The offline batch scheduling problem under KV-cache constraints shares structural similarities with two-dimensional packing. In our memory-time area framework, jobs occupy triangular or trapezoidal footprints, and scheduling them within the memory budget resembles packing shapes into a strip of fixed width. \citet{LMM02} provide a comprehensive survey of 2D orthogonal packing variants (e.g., bin packing and strip packing for rectangles), summarizing the complexity landscape and standard algorithmic paradigms including exact methods, approximation algorithms, and heuristics. For strip packing specifically, \citet{AKPP17} prove APX-hardness with polynomially bounded item dimensions---in particular, it is NP-hard to approximate within a factor better than $12/11 - \varepsilon$.

\xhdr{Flow Time Minimization.}
Minimizing total flow time on parallel machines has been extensively studied. \citet{LR07} develop approximation algorithms for total flow time on parallel machines with release dates, providing the first nontrivial approximation guarantees for this setting. \citet{AALR02} study flow-time minimization under the restriction that jobs cannot migrate between machines, providing competitive guarantees via techniques that compare to migratory benchmarks. More recently, \citet{GSYZ25} establish tight randomized bounds for online non-preemptive total flow time on identical machines, matching $\Theta(\sqrt{n/m})$ upper and lower bounds, and extend the framework to kill-and-restart preemption with similar bounds.

\xhdr{Non-Clairvoyant Scheduling.}
The non-clairvoyant scheduling model, where processing times are revealed only upon job completion, was formalized by \citet{MPT94}, who prove foundational upper and lower bounds and analyze baseline policies including round-robin for flow-time objectives. \citet{BL01} analyze non-clairvoyant flow-time minimization on single and parallel machines, showing competitiveness of multilevel-feedback-style randomized strategies under adversarial models. These classical results assume costless preemption: a job can be paused and resumed without penalty. In contrast, \citet{J-25} study non-clairvoyant single-machine scheduling with kill-and-restart preemption, where preempted jobs lose all progress, proving lower bounds for deterministic strategies and giving tight analyses for scaling restart strategies.

\xhdr{Scheduling with Partial Information.}
Between the extremes of clairvoyance and non-clairvoyance, several intermediate information models have been studied. \citet{BLMP04} formalize semi-clairvoyant scheduling where only coarse size classes are known, deriving constant-competitive algorithms for average flow time. \citet{ALT-22} develop online single-machine flow-time algorithms that remain competitive for every distortion level between predicted and true processing times. \citet{IKQP23} introduce prediction-augmented non-clairvoyant scheduling for minimizing total completion time, designing algorithms with robustness--consistency guarantees as prediction error varies. \citet{BP24} study non-clairvoyant scheduling where predictions are available for only a subset of jobs, establishing near-optimal bounds and a learning-augmented algorithm with robustness-consistency-smoothness tradeoffs. \citet{BCLS25} study non-clairvoyant scheduling with continuous feedback in the form of progress estimates, providing competitive algorithms under adversarial and stochastic progress-bar models.

\xhdr{Scheduling for LLM Inference.}
\citet{JJMMPZ-25} formalize online batching under a hard KV-cache memory constraint, establish strong impossibility results for adversarial arrivals, and propose a greedy shortest-first policy with bounded competitive guarantees under restricted assumptions.  
\citet{WYZ-25} study settings with heterogeneous prefill and decode lengths, prove that the resulting batching problem is NP-hard, and design constant-competitive schedulers.  
\citet{CYZ-25} consider the same offline batch inference model as ours and propose a scheduling algorithm with an approximation ratio of $O(\log(\responseLen_{\max}/\responseLen_{\min}))$ in the large-memory regime, where $\responseLen_{\max}$ and $\responseLen_{\min}$ denote the maximum and minimum response lengths, respectively. Since $\responseLen_{\max} \leq \kvmem$ and $\responseLen_{\min} \geq 1$, their result implies a non-clairvoyant competitive ratio of $O(\log \kvmem)$ for large-memory instances. They also analyze algorithm performance under additional distributional assumptions on response lengths.

Beyond worst-case analysis, \citet{ALSW25} introduce a fluid approximation as an analytical benchmark and derive threshold-based online policies with heavy-traffic optimality guarantees.  
\citet{LDP-25} adopt a queueing-theoretic perspective, proving throughput-optimality for work-conserving policies in LLM serving.  
\citet{ZLMP-25} address prompt caching in multi-turn conversations, focusing on tail latency; they propose a modified LRU policy optimized for P90/P95 objectives. 
\section{Preliminaries}
\label{sec:prelim}
\subsection{Model and Objective}
\label{sec:model}

We now establish the formal framework for the non-clairvoyant scheduling problem. While abstract, this model captures the operational realities of modern LLM inference; we refer the reader to \Cref{sec:model justificaiton} for detailed justifications of our core assumptions.

\xhdr{Jobs and Scheduling Actions.}
The system consists of a single computational worker (aka., a GPU) managed by a scheduler. There are $\NumberJobs$ jobs (aka., inference requests), each job $i$ has a \emph{prompt} of length $\initialLen \in \naturals$ and a \emph{response length} (aka., the number of decode tokens) $\responseLen_i \in \naturals$, which is the number of time units needed to complete the job once it starts. In this paper, we focus on the offline batch setting with identical prompt sizes, i.e., all jobs are available at time $0$ and share a common prompt length $\initialLen$.

Time proceeds in discrete, unit-length rounds $t=0,1,2,\cdots$. Let $\curProgress_{i,t}\in\{0,\cdots,\responseLen_i\}$ be the amount of processing the job has received up to the \emph{start} of round $t$. If job $i$ reaches $\curProgress_{i,t} = \responseLen_i$ at round $t$, then this job is completed (aka., finished) and its completion time is defined as $t$.
At the beginning of each round $t$, the scheduler chooses an \emph{active batch} $\ActiveBatch_t \subseteq [\NumberJobs]$ of unfinished jobs to process. During round $t$, every active job $i\in\ActiveBatch_t$ receives one unit of processing (i.e., decode the next token), so
\begin{align*}
\curProgress_{i,t+1}=\curProgress_{i,t}+1.
\end{align*}
The active batch is \emph{dynamic}: its composition can change in each round. Essentially, the scheduler may perform two actions at the beginning of each round by choosing $\ActiveBatch_t$:
\begin{enumerate}
    \item \underline{\emph{(Start)}}: activate an inactive job $i \notin \ActiveBatch_{t - 1}$ where $\curProgress_{i,t}=0$.
    \item \underline{\emph{(Kill)}}: deactivate an active, unfinished job $i \in \ActiveBatch_{t - 1} \setminus \ActiveBatch_{t}$ where $\curProgress_{i,t} < \responseLen_i$, discarding all its accumulated progress and resetting $\curProgress_{i,t + 1}$ to $0$.
\end{enumerate}

\xhdr{KV-Cache Memory Feasibility.}
Unlike classical parallel scheduling with only a hard limit on the number of concurrent jobs, our model features a \emph{KV-cache memory budget} $\kvmem \in \naturals$. This budget constrains the total memory used by active jobs, which grows as the tokens are decoded one by one.
When job $i$ has received $\curProgress_{i,t}$ units of processing by round $t$, it occupies $\initialLen + \curProgress_{i,t} + 1$ units of KV-cache memory: $\initialLen$ units for the prompt and one additional unit for each decoded token (including the token being decoded in round $t$).
Thus, the memory feasibility constraint requires 
\begin{align}
    \label{eq:memory-feasibility}
    \tag{\textsc{Memory-Feasibility}}
    \forall t\in\naturals:~~~~~~~~~
    \sum\nolimits_{i \in \ActiveBatch_t} \bigl( \initialLen + \left(\curProgress_{i,t}+1\right) \bigr) \;\leq\; \kvmem.
\end{align}
To avoid trivial infeasibility, we assume \emph{per-job feasibility}: $\initialLen + \responseLen_i \leq \kvmem$ for all jobs $i \in [\NumberJobs]$. This ensures that any job can be completed when processed in isolation.

\xhdr{Non-Clairvoyant Environment.}
Our base model assumes a \emph{non-clairvoyant scheduler}, who does not know the response length $\responseLen_i$ of any job $i$ in advance. This value is revealed only upon the job's completion. In addition to the base model, we also consider a clairvoyant variant (in Section~\ref{sec:clairvoyant greedy}) in which the scheduler has full knowledge of all response lengths $\{\responseLen_i\}_{i \in [\NumberJobs]}$.

\xhdr{Objective and Benchmark.}
The scheduler's objective is to minimize the \emph{total flow time} (aka., end-to-end latency).\footnote{In our setting, all jobs are available at time zero, so flow time equals completion time.}  Let $\ALG$ be a scheduling algorithm, and let $\finishTime_i$ denote the round in which job $i$ completes under $\ALG$.  The total flow time of $\ALG$ is defined as
\begin{align*}
    \FTime{\ALG} \;\triangleq\; \sum\nolimits_{i \in [\NumberJobs]} \finishTime_i.
\end{align*}
Our primary focus is on designing non-clairvoyant algorithms that operate without knowledge of job sizes. To quantify the performance loss incurred by this lack of information, we benchmark our algorithms against $\OPT$, the minimum total flow time achievable by any clairvoyant scheduler satisfying the memory feasibility constraint~\eqref{eq:memory-feasibility}. We use the standard terminology of \textit{competitive ratio} for this comparison.
Specifically, $\ALG$ is \emph{$\approxratio$-competitive} if, for every problem instance, $\FTime{\ALG} \;\leq\; \approxratio \cdot \OPT$.\footnote{We allow the optimal clairvoyant scheduler to preempt or pause jobs (i.e., retain a job in memory without processing it). However, it can be shown that neither preemption nor pausing is necessary in an optimal solution, as postpoing the start of a job will always be a better choice.} For clairvoyant algorithms where the comparison is against the optimal solution of the same information structure, this measure is referred to as the \textit{approximation ratio}.

\subsection{Discussion on the Model Primitives}
\label{sec:model justificaiton}
We next explain several key modeling choices made in our problem formulation.
Our theoretical model abstracts the complex dynamics of LLM inference systems. Here, we elaborate on the practical considerations that inform our theoretical abstractions and place our work in the context of related literature.

\xhdr{Offline Batch Inference and Identical Prompt Lengths.}
We assume all jobs are available at time $0$ and share identical prompt lengths. This models the \emph{offline batch inference} scenario, a common industry practice for high-throughput tasks where aggregating requests of similar prompt lengths is practical and cost-effective. For instance, both OpenAI and Google offer Batch APIs for processing large workloads asynchronously, often with significant cost savings~\citep{OpenAI-BatchAPI-24, Google-BatchAPI-24}. Our model can be viewed as optimizing the scheduling in every single batch request with similar prompt lengths. 
Furthermore, the identical prompt length assumption is also motivated by fixed prompt scenarios such as \emph{agentic workflows}, where each agent is defined by a fixed system prompt and inference providers use prefix caching to reuse KV tensors corresponding to these prompts across iterations~\citep{ZXYTCCZCHWZ25,HAS025,PPSHGLQWD25}. Since the system prompts remain constant, the dominant part of prompt lengths becomes more uniform across requests.

Theoretically speaking, we introduce these two assumptions to isolate the algorithmic challenges arising from the growing memory scenario and the non-clairvoyant perspective. The rationale is that relaxing either assumption renders the problem significantly more difficult, even for clairvoyant schedulers with static memory:

\begin{itemize}
    \item \emph{(Online arrivals)} Given identical prompt lengths, setting $\initialLen = \kvmem / 2$ simplifies the problem to online single-machine scheduling to minimize total flow time under kill-and-restart preemption. This problem admits a competitive ratio lower bound of $\Omega(\sqrt{n})$ for $n$ jobs, even for randomized algorithms~\citep{ES03}.
    \item \emph{(Heterogeneous prompt lengths)} In the offline batch setting where $\responseLen_i = 1$ for all jobs, allowing heterogeneous prompt lengths reduces the problem to Min-Sum Bin Packing, which is strongly NP-hard~\citep{EJL18}.
\end{itemize}

Consequently, without these two assumptions, the analysis would deviate from our primary focus and become bogged down by the aforementioned fundamental hardness results.

\xhdr{Memory Model.} The constraint~\eqref{eq:memory-feasibility} assumes that memory is consumed strictly linearly with the number of decoded tokens and that fragmentation is negligible. This assumption is grounded in modern inference engine architectures, specifically the \emph{PagedAttention} mechanism introduced in vLLM~\citep{KLZSZYGZS-23}, and also in other LLM inference systems~\citep{NVIDIA-23,ZYXS0YCKSGB24}. PagedAttention partitions the KV-cache into fixed-size blocks that can be stored in non-contiguous memory, effectively eliminating external fragmentation. Consequently, the memory budget $\kvmem$ represents the total number of available token slots (or blocks) on the GPU, and feasibility is determined solely by the total count of tokens, as modeled in~\eqref{eq:memory-feasibility}.

\xhdr{Modeling Prefill Cost via Kill-and-Restart.}
Although our model does not explicitly represent the prefill phase, its cost is implicitly captured by our kill-and-restart preemption mechanism.
In modern LLM inference systems, the primary output of the prefill phase is the construction of the initial KV-cache for the prompt.
When a request is preempted, this KV-cache state must either be swapped to slower memory or discarded entirely. In the latter case, resuming the request requires rerunning the prefill computation from scratch to reconstruct this state.
Our model focuses on the kill-and-restart dynamic, as it cleanly captures the essential trade-offs in high-throughput systems. This modeling choice is grounded in the operational reality of inference engines, where the I/O overhead of swapping the KV-cache is often prohibitively expensive, making kill-based preemption the preferred strategy~\citep{KLZSZYGZS-23}.

Furthermore, the prefill phase can be of separate interest and can be addressed orthogonally.
In practice, prefill can be handled by \emph{prefill/decode disaggregation} (PD separation), which assigns prefill and decode to distinct resource pools. This architectural choice reduces phase interference and enables phase-specific optimizations~\citep{ZLCHZL0024,CPZSGMFB25}.
Accordingly, our model measures time in unit rounds corresponding to a single decode step (one output token) and does not explicitly include the prefill stage in the scheduling horizon.

\section{Geometric Slicing Algorithm}
\label{sec:algorithm description}

Our main contribution is a non-clairvoyant KV-cache scheduling algorithm, {\GeoSALG} ({\GSA}). In \Cref{subsec:staggered}, we first introduce the {\SPipelineSchedule} policy, which is near-optimal for instances with identical jobs. This policy serves as the core subroutine of {\GSA}, which we present in \Cref{subsec:main alg} with all its algorithmic ingredients. The competitive ratio analysis of {\GSA} is deferred to \Cref{sec:clairvoyant greedy,sec:algorithms}.

\subsection{Subroutine: {\SPipelineSchedule}}
\label{subsec:staggered}

In this section, we introduce {\SPipelineSchedule} ({\SPS}), a core scheduling subroutine that serves as the foundation for our main algorithms. While simple in its construction, {\SPS} addresses a fundamental inefficiency in how current systems handle KV-cache memory, and its design may be of independent interest for broader memory-constrained scheduling problems.

The design of {\SPS} is motivated by a canonical setting: \emph{given a sequence of jobs $\Jobs$ with \emph{identical} prompt length $\initialLen$ and response length $\responseLen$, how can we schedule them to minimize total completion time (or equivalently, maximize throughput) under a fixed memory budget?}

\xhdr{The Simultaneous Batching Trap.} A conventional approach to this problem is \emph{simultaneous batching}: packing as many jobs as possible into a single batch and executing them in lockstep. This strategy is a natural consequence of standard policies like First-Come-First-Served (FCFS) or Shortest-First under certain conditions, and is widely adopted in LLM serving systems~\citep{KLZSZYGZS-23,ZYXS0YCKSGB24} as well as the theoretical literature~\citep{JJMMPZ-25,CYZ-25}.\footnote{Notably, in the classic flow-time minimization scheduling problem (where memory occupation remains static over time) with identical jobs, simultaneous batching is indeed optimal.}

However, simultaneous execution is inherently memory-inefficient. Since all jobs in a batch progress through their generation phases at the same rate, they reach their peak memory usage simultaneously. To prevent memory overflow, the scheduler must limit the batch size based on this collective peak. This leads to a ``sawtooth'' memory profile where, on average, a significant portion of the memory budget remains underutilized, as illustrated in \Cref{example:schedule comparison} and \Cref{subfig:schedule comparison:simultaneous}.

\xhdr{{\SPS} Subroutine: Smoothing via Staggering.} To overcome this bottleneck, {\SPipelineSchedule} ({\SPS}) adopts a \emph{staggered} execution strategy. Instead of starting all jobs at once, {\SPS} intentionally offsets their start times to ``smooth out'' the aggregate memory consumption over time. By ensuring that different jobs reach their peak memory usage at different rounds, {\SPS} can safely sustain a much higher degree of parallelism within the same memory budget.

Formally, {\SPS} is parameterized by a \emph{degree of parallelism} $\paral$ and a \emph{slice length} $\sliceLen$. For a sequence of jobs indexed $\jobIdx=0, \dots, |\Jobs|-1$, it schedules job $\jobIdx$ to start at round $\startTimeOf{\jobIdx}\triangleq\left\lfloor {\jobIdx \cdot \sliceLen}/{\paral} \right\rfloor$. When the slice length $\sliceLen$ is set to the identical response length $\responseLen$, each job completes at round $\compTimeOf{\jobIdx}\triangleq \startTimeOf{\jobIdx} + \responseLen$. A detailed description is provided in Procedure~\ref{alg:staggered}, and a visual comparison is shown in \Cref{example:schedule comparison} and \Cref{subfig:schedule comparison:staggered}.

\begin{myprocedure}[ht]
\caption{{\SPipelineSchedule} ({\SPS})}
\label{alg:staggered}
\SetKwInOut{Input}{input}
\SetKwInOut{Output}{output}
\Input{
    Sequence of jobs $\Jobs = \{\job_0, \job_1, \dots\}$,  Degree of parallelism $\paral \in \naturals$, 
    Slice length $\sliceLen \in \naturals$
}
\Output{
    Start time $\startTimeOf{\jobIdx}$ and end time $\compTimeOf{\jobIdx}$ for each job $\job_{\jobIdx} \in \Jobs$
}
\BlankLine
\For{each job index $\jobIdx \in \{0, 1, \cdots, |\Jobs| - 1\}$}{
    set start time $\startTimeOf{\jobIdx} \gets \left\lfloor {\jobIdx \cdot \sliceLen}/{\paral} \right\rfloor$\;
    set end time $\compTimeOf{\jobIdx} \gets \startTimeOf{\jobIdx} + \sliceLen$\;
    schedule job $\job_{\jobIdx}$ to run in the interval $[\startTimeOf{\jobIdx}, \compTimeOf{\jobIdx})$\;
    \tcc{Jobs not completed by the start of round $\compTimeOf{\jobIdx}$ are killed.}
}
\Return $\{\startTimeOf{\jobIdx}, \compTimeOf{\jobIdx}\}_{\job_{\jobIdx} \in \Jobs}$\;
\end{myprocedure} 
\begin{example}[A Toy Example with Identical Jobs]
    \label{example:schedule comparison}
    Consider an instance with $15$ identical jobs (indexed $\jobIdx=0, \dots, 14$), each having prompt length $\initialLen = 0$ and response length $\responseLen = 5$, under a KV-cache memory budget of $\kvmem = 15$.

    Under {\SPipelineSchedule} ({\SPS}) with degree of parallelism $\paral = 5$ and slice length $\sliceLen = \responseLen = 5$, the $\jobIdx$-th job starts at round $\lfloor 5\jobIdx/5 \rfloor = \jobIdx$. The last job ($\jobIdx=14$) starts at round $14$ and completes at the end of round $18$. This yields a total flow time of $\FTime{\SPS} = 5 + 6 + \cdots + 19 = 180$.
    See \Cref{subfig:schedule comparison:staggered} for a graphical illustration.

    In contrast, {\SimPipelineSchedule} ({\SimPS}) processes $\paral = 3$ jobs at a time, starting new jobs in rounds $0$, $5$, $10$, $15$, and $20$. This yields a lower degree of parallelism due to peak memory constraints, and the last job finishes at round~$25$, resulting in a total flow time of $\FTime{\SimPS} = 225$.  
    See \Cref{subfig:schedule comparison:simultaneous} for a graphical illustration. 
\end{example}

\begin{figure}
  \centering
\subfloat[Execution of {\SPS} with $\paral = 5$]{
\begin{tikzpicture}[x=0.25cm,y=0.25cm]
\definecolor{cellbg}{RGB}{245,245,245}
  \definecolor{job1}{RGB}{230,115,115}
  \definecolor{job2}{RGB}{230,161,115}
  \definecolor{job3}{RGB}{230,207,115}
  \definecolor{job4}{RGB}{207,230,115}
  \definecolor{job5}{RGB}{161,230,115}
  \definecolor{job6}{RGB}{115,230,115}
  \definecolor{job7}{RGB}{115,230,161}
  \definecolor{job8}{RGB}{115,230,207}
  \definecolor{job9}{RGB}{115,207,230}
  \definecolor{job10}{RGB}{115,161,230}
  \definecolor{job11}{RGB}{115,115,230}
  \definecolor{job12}{RGB}{161,115,230}
  \definecolor{job13}{RGB}{207,115,230}
  \definecolor{job14}{RGB}{230,115,207}
  \definecolor{job15}{RGB}{230,115,161}
  \fill[cellbg] (0,0) rectangle ++(19,15);
\fill[job1] (0,0) rectangle ++(1,1);
  \fill[job2] (1,0) rectangle ++(1,1);
  \fill[job1] (1,1) rectangle ++(1,1);
  \fill[job1] (1,2) rectangle ++(1,1);
  \fill[job3] (2,0) rectangle ++(1,1);
  \fill[job2] (2,1) rectangle ++(1,1);
  \fill[job2] (2,2) rectangle ++(1,1);
  \fill[job1] (2,3) rectangle ++(1,1);
  \fill[job1] (2,4) rectangle ++(1,1);
  \fill[job1] (2,5) rectangle ++(1,1);
  \fill[job4] (3,0) rectangle ++(1,1);
  \fill[job3] (3,1) rectangle ++(1,1);
  \fill[job3] (3,2) rectangle ++(1,1);
  \fill[job2] (3,3) rectangle ++(1,1);
  \fill[job2] (3,4) rectangle ++(1,1);
  \fill[job2] (3,5) rectangle ++(1,1);
  \fill[job1] (3,6) rectangle ++(1,1);
  \fill[job1] (3,7) rectangle ++(1,1);
  \fill[job1] (3,8) rectangle ++(1,1);
  \fill[job1] (3,9) rectangle ++(1,1);
  \fill[job5] (4,0) rectangle ++(1,1);
  \fill[job4] (4,1) rectangle ++(1,1);
  \fill[job4] (4,2) rectangle ++(1,1);
  \fill[job3] (4,3) rectangle ++(1,1);
  \fill[job3] (4,4) rectangle ++(1,1);
  \fill[job3] (4,5) rectangle ++(1,1);
  \fill[job2] (4,6) rectangle ++(1,1);
  \fill[job2] (4,7) rectangle ++(1,1);
  \fill[job2] (4,8) rectangle ++(1,1);
  \fill[job2] (4,9) rectangle ++(1,1);
  \fill[job1] (4,10) rectangle ++(1,1);
  \fill[job1] (4,11) rectangle ++(1,1);
  \fill[job1] (4,12) rectangle ++(1,1);
  \fill[job1] (4,13) rectangle ++(1,1);
  \fill[job1] (4,14) rectangle ++(1,1);
  \fill[job6] (5,0) rectangle ++(1,1);
  \fill[job5] (5,1) rectangle ++(1,1);
  \fill[job5] (5,2) rectangle ++(1,1);
  \fill[job4] (5,3) rectangle ++(1,1);
  \fill[job4] (5,4) rectangle ++(1,1);
  \fill[job4] (5,5) rectangle ++(1,1);
  \fill[job3] (5,6) rectangle ++(1,1);
  \fill[job3] (5,7) rectangle ++(1,1);
  \fill[job3] (5,8) rectangle ++(1,1);
  \fill[job3] (5,9) rectangle ++(1,1);
  \fill[job2] (5,10) rectangle ++(1,1);
  \fill[job2] (5,11) rectangle ++(1,1);
  \fill[job2] (5,12) rectangle ++(1,1);
  \fill[job2] (5,13) rectangle ++(1,1);
  \fill[job2] (5,14) rectangle ++(1,1);
  \fill[job7] (6,0) rectangle ++(1,1);
  \fill[job6] (6,1) rectangle ++(1,1);
  \fill[job6] (6,2) rectangle ++(1,1);
  \fill[job5] (6,3) rectangle ++(1,1);
  \fill[job5] (6,4) rectangle ++(1,1);
  \fill[job5] (6,5) rectangle ++(1,1);
  \fill[job4] (6,6) rectangle ++(1,1);
  \fill[job4] (6,7) rectangle ++(1,1);
  \fill[job4] (6,8) rectangle ++(1,1);
  \fill[job4] (6,9) rectangle ++(1,1);
  \fill[job3] (6,10) rectangle ++(1,1);
  \fill[job3] (6,11) rectangle ++(1,1);
  \fill[job3] (6,12) rectangle ++(1,1);
  \fill[job3] (6,13) rectangle ++(1,1);
  \fill[job3] (6,14) rectangle ++(1,1);
  \fill[job8] (7,0) rectangle ++(1,1);
  \fill[job7] (7,1) rectangle ++(1,1);
  \fill[job7] (7,2) rectangle ++(1,1);
  \fill[job6] (7,3) rectangle ++(1,1);
  \fill[job6] (7,4) rectangle ++(1,1);
  \fill[job6] (7,5) rectangle ++(1,1);
  \fill[job5] (7,6) rectangle ++(1,1);
  \fill[job5] (7,7) rectangle ++(1,1);
  \fill[job5] (7,8) rectangle ++(1,1);
  \fill[job5] (7,9) rectangle ++(1,1);
  \fill[job4] (7,10) rectangle ++(1,1);
  \fill[job4] (7,11) rectangle ++(1,1);
  \fill[job4] (7,12) rectangle ++(1,1);
  \fill[job4] (7,13) rectangle ++(1,1);
  \fill[job4] (7,14) rectangle ++(1,1);
  \fill[job9] (8,0) rectangle ++(1,1);
  \fill[job8] (8,1) rectangle ++(1,1);
  \fill[job8] (8,2) rectangle ++(1,1);
  \fill[job7] (8,3) rectangle ++(1,1);
  \fill[job7] (8,4) rectangle ++(1,1);
  \fill[job7] (8,5) rectangle ++(1,1);
  \fill[job6] (8,6) rectangle ++(1,1);
  \fill[job6] (8,7) rectangle ++(1,1);
  \fill[job6] (8,8) rectangle ++(1,1);
  \fill[job6] (8,9) rectangle ++(1,1);
  \fill[job5] (8,10) rectangle ++(1,1);
  \fill[job5] (8,11) rectangle ++(1,1);
  \fill[job5] (8,12) rectangle ++(1,1);
  \fill[job5] (8,13) rectangle ++(1,1);
  \fill[job5] (8,14) rectangle ++(1,1);
  \fill[job10] (9,0) rectangle ++(1,1);
  \fill[job9] (9,1) rectangle ++(1,1);
  \fill[job9] (9,2) rectangle ++(1,1);
  \fill[job8] (9,3) rectangle ++(1,1);
  \fill[job8] (9,4) rectangle ++(1,1);
  \fill[job8] (9,5) rectangle ++(1,1);
  \fill[job7] (9,6) rectangle ++(1,1);
  \fill[job7] (9,7) rectangle ++(1,1);
  \fill[job7] (9,8) rectangle ++(1,1);
  \fill[job7] (9,9) rectangle ++(1,1);
  \fill[job6] (9,10) rectangle ++(1,1);
  \fill[job6] (9,11) rectangle ++(1,1);
  \fill[job6] (9,12) rectangle ++(1,1);
  \fill[job6] (9,13) rectangle ++(1,1);
  \fill[job6] (9,14) rectangle ++(1,1);
  \fill[job11] (10,0) rectangle ++(1,1);
  \fill[job10] (10,1) rectangle ++(1,1);
  \fill[job10] (10,2) rectangle ++(1,1);
  \fill[job9] (10,3) rectangle ++(1,1);
  \fill[job9] (10,4) rectangle ++(1,1);
  \fill[job9] (10,5) rectangle ++(1,1);
  \fill[job8] (10,6) rectangle ++(1,1);
  \fill[job8] (10,7) rectangle ++(1,1);
  \fill[job8] (10,8) rectangle ++(1,1);
  \fill[job8] (10,9) rectangle ++(1,1);
  \fill[job7] (10,10) rectangle ++(1,1);
  \fill[job7] (10,11) rectangle ++(1,1);
  \fill[job7] (10,12) rectangle ++(1,1);
  \fill[job7] (10,13) rectangle ++(1,1);
  \fill[job7] (10,14) rectangle ++(1,1);
  \fill[job12] (11,0) rectangle ++(1,1);
  \fill[job11] (11,1) rectangle ++(1,1);
  \fill[job11] (11,2) rectangle ++(1,1);
  \fill[job10] (11,3) rectangle ++(1,1);
  \fill[job10] (11,4) rectangle ++(1,1);
  \fill[job10] (11,5) rectangle ++(1,1);
  \fill[job9] (11,6) rectangle ++(1,1);
  \fill[job9] (11,7) rectangle ++(1,1);
  \fill[job9] (11,8) rectangle ++(1,1);
  \fill[job9] (11,9) rectangle ++(1,1);
  \fill[job8] (11,10) rectangle ++(1,1);
  \fill[job8] (11,11) rectangle ++(1,1);
  \fill[job8] (11,12) rectangle ++(1,1);
  \fill[job8] (11,13) rectangle ++(1,1);
  \fill[job8] (11,14) rectangle ++(1,1);
  \fill[job13] (12,0) rectangle ++(1,1);
  \fill[job12] (12,1) rectangle ++(1,1);
  \fill[job12] (12,2) rectangle ++(1,1);
  \fill[job11] (12,3) rectangle ++(1,1);
  \fill[job11] (12,4) rectangle ++(1,1);
  \fill[job11] (12,5) rectangle ++(1,1);
  \fill[job10] (12,6) rectangle ++(1,1);
  \fill[job10] (12,7) rectangle ++(1,1);
  \fill[job10] (12,8) rectangle ++(1,1);
  \fill[job10] (12,9) rectangle ++(1,1);
  \fill[job9] (12,10) rectangle ++(1,1);
  \fill[job9] (12,11) rectangle ++(1,1);
  \fill[job9] (12,12) rectangle ++(1,1);
  \fill[job9] (12,13) rectangle ++(1,1);
  \fill[job9] (12,14) rectangle ++(1,1);
  \fill[job14] (13,0) rectangle ++(1,1);
  \fill[job13] (13,1) rectangle ++(1,1);
  \fill[job13] (13,2) rectangle ++(1,1);
  \fill[job12] (13,3) rectangle ++(1,1);
  \fill[job12] (13,4) rectangle ++(1,1);
  \fill[job12] (13,5) rectangle ++(1,1);
  \fill[job11] (13,6) rectangle ++(1,1);
  \fill[job11] (13,7) rectangle ++(1,1);
  \fill[job11] (13,8) rectangle ++(1,1);
  \fill[job11] (13,9) rectangle ++(1,1);
  \fill[job10] (13,10) rectangle ++(1,1);
  \fill[job10] (13,11) rectangle ++(1,1);
  \fill[job10] (13,12) rectangle ++(1,1);
  \fill[job10] (13,13) rectangle ++(1,1);
  \fill[job10] (13,14) rectangle ++(1,1);
  \fill[job15] (14,0) rectangle ++(1,1);
  \fill[job14] (14,1) rectangle ++(1,1);
  \fill[job14] (14,2) rectangle ++(1,1);
  \fill[job13] (14,3) rectangle ++(1,1);
  \fill[job13] (14,4) rectangle ++(1,1);
  \fill[job13] (14,5) rectangle ++(1,1);
  \fill[job12] (14,6) rectangle ++(1,1);
  \fill[job12] (14,7) rectangle ++(1,1);
  \fill[job12] (14,8) rectangle ++(1,1);
  \fill[job12] (14,9) rectangle ++(1,1);
  \fill[job11] (14,10) rectangle ++(1,1);
  \fill[job11] (14,11) rectangle ++(1,1);
  \fill[job11] (14,12) rectangle ++(1,1);
  \fill[job11] (14,13) rectangle ++(1,1);
  \fill[job11] (14,14) rectangle ++(1,1);
  \fill[job15] (15,1) rectangle ++(1,1);
  \fill[job15] (15,2) rectangle ++(1,1);
  \fill[job14] (15,3) rectangle ++(1,1);
  \fill[job14] (15,4) rectangle ++(1,1);
  \fill[job14] (15,5) rectangle ++(1,1);
  \fill[job13] (15,6) rectangle ++(1,1);
  \fill[job13] (15,7) rectangle ++(1,1);
  \fill[job13] (15,8) rectangle ++(1,1);
  \fill[job13] (15,9) rectangle ++(1,1);
  \fill[job12] (15,10) rectangle ++(1,1);
  \fill[job12] (15,11) rectangle ++(1,1);
  \fill[job12] (15,12) rectangle ++(1,1);
  \fill[job12] (15,13) rectangle ++(1,1);
  \fill[job12] (15,14) rectangle ++(1,1);
  \fill[job15] (16,3) rectangle ++(1,1);
  \fill[job15] (16,4) rectangle ++(1,1);
  \fill[job15] (16,5) rectangle ++(1,1);
  \fill[job14] (16,6) rectangle ++(1,1);
  \fill[job14] (16,7) rectangle ++(1,1);
  \fill[job14] (16,8) rectangle ++(1,1);
  \fill[job14] (16,9) rectangle ++(1,1);
  \fill[job13] (16,10) rectangle ++(1,1);
  \fill[job13] (16,11) rectangle ++(1,1);
  \fill[job13] (16,12) rectangle ++(1,1);
  \fill[job13] (16,13) rectangle ++(1,1);
  \fill[job13] (16,14) rectangle ++(1,1);
  \fill[job15] (17,6) rectangle ++(1,1);
  \fill[job15] (17,7) rectangle ++(1,1);
  \fill[job15] (17,8) rectangle ++(1,1);
  \fill[job15] (17,9) rectangle ++(1,1);
  \fill[job14] (17,10) rectangle ++(1,1);
  \fill[job14] (17,11) rectangle ++(1,1);
  \fill[job14] (17,12) rectangle ++(1,1);
  \fill[job14] (17,13) rectangle ++(1,1);
  \fill[job14] (17,14) rectangle ++(1,1);
  \fill[job15] (18,10) rectangle ++(1,1);
  \fill[job15] (18,11) rectangle ++(1,1);
  \fill[job15] (18,12) rectangle ++(1,1);
  \fill[job15] (18,13) rectangle ++(1,1);
  \fill[job15] (18,14) rectangle ++(1,1);
\draw[black, line width=0.2pt] (0,0) rectangle (19,15);
  \draw[black!60, line width=0.2pt] (1,0)--(1,15);
  \draw[black!60, line width=0.2pt] (2,0)--(2,15);
  \draw[black!60, line width=0.2pt] (3,0)--(3,15);
  \draw[black!60, line width=0.2pt] (4,0)--(4,15);
  \draw[black!60, line width=0.2pt] (5,0)--(5,15);
  \draw[black!60, line width=0.2pt] (6,0)--(6,15);
  \draw[black!60, line width=0.2pt] (7,0)--(7,15);
  \draw[black!60, line width=0.2pt] (8,0)--(8,15);
  \draw[black!60, line width=0.2pt] (9,0)--(9,15);
  \draw[black!60, line width=0.2pt] (10,0)--(10,15);
  \draw[black!60, line width=0.2pt] (11,0)--(11,15);
  \draw[black!60, line width=0.2pt] (12,0)--(12,15);
  \draw[black!60, line width=0.2pt] (13,0)--(13,15);
  \draw[black!60, line width=0.2pt] (14,0)--(14,15);
  \draw[black!60, line width=0.2pt] (15,0)--(15,15);
  \draw[black!60, line width=0.2pt] (16,0)--(16,15);
  \draw[black!60, line width=0.2pt] (17,0)--(17,15);
  \draw[black!60, line width=0.2pt] (18,0)--(18,15);
  \draw[black!60, line width=0.2pt] (0,1)--(19,1);
  \draw[black!60, line width=0.2pt] (0,2)--(19,2);
  \draw[black!60, line width=0.2pt] (0,3)--(19,3);
  \draw[black!60, line width=0.2pt] (0,4)--(19,4);
  \draw[black!60, line width=0.2pt] (0,5)--(19,5);
  \draw[black!60, line width=0.2pt] (0,6)--(19,6);
  \draw[black!60, line width=0.2pt] (0,7)--(19,7);
  \draw[black!60, line width=0.2pt] (0,8)--(19,8);
  \draw[black!60, line width=0.2pt] (0,9)--(19,9);
  \draw[black!60, line width=0.2pt] (0,10)--(19,10);
  \draw[black!60, line width=0.2pt] (0,11)--(19,11);
  \draw[black!60, line width=0.2pt] (0,12)--(19,12);
  \draw[black!60, line width=0.2pt] (0,13)--(19,13);
  \draw[black!60, line width=0.2pt] (0,14)--(19,14);
\node[anchor=north] at (9.50,-0.15) {\small round $t = {0, \dots, 18}$};
  \node[rotate=90,anchor=south] at (-0.35,7.50) {\small memory ($\kvmem = 15$)};
\end{tikzpicture} \label{subfig:schedule comparison:staggered}
}
~~~~
  \subfloat[Execution of {\SimPS} with $\paral = 3$]{
\begin{tikzpicture}[x=0.25cm,y=0.25cm]
\definecolor{cellbg}{RGB}{245,245,245}
  \definecolor{job1}{RGB}{230,115,115}
  \definecolor{job2}{RGB}{230,161,115}
  \definecolor{job3}{RGB}{230,207,115}
  \definecolor{job4}{RGB}{207,230,115}
  \definecolor{job5}{RGB}{161,230,115}
  \definecolor{job6}{RGB}{115,230,115}
  \definecolor{job7}{RGB}{115,230,161}
  \definecolor{job8}{RGB}{115,230,207}
  \definecolor{job9}{RGB}{115,207,230}
  \definecolor{job10}{RGB}{115,161,230}
  \definecolor{job11}{RGB}{115,115,230}
  \definecolor{job12}{RGB}{161,115,230}
  \definecolor{job13}{RGB}{207,115,230}
  \definecolor{job14}{RGB}{230,115,207}
  \definecolor{job15}{RGB}{230,115,161}
  \fill[cellbg] (0,0) rectangle ++(25,15);
\fill[job1] (0,0) rectangle ++(1,1);
  \fill[job2] (0,5) rectangle ++(1,1);
  \fill[job3] (0,10) rectangle ++(1,1);
  \fill[job1] (1,0) rectangle ++(1,1);
  \fill[job1] (1,1) rectangle ++(1,1);
  \fill[job2] (1,5) rectangle ++(1,1);
  \fill[job2] (1,6) rectangle ++(1,1);
  \fill[job3] (1,10) rectangle ++(1,1);
  \fill[job3] (1,11) rectangle ++(1,1);
  \fill[job1] (2,0) rectangle ++(1,1);
  \fill[job1] (2,1) rectangle ++(1,1);
  \fill[job1] (2,2) rectangle ++(1,1);
  \fill[job2] (2,5) rectangle ++(1,1);
  \fill[job2] (2,6) rectangle ++(1,1);
  \fill[job2] (2,7) rectangle ++(1,1);
  \fill[job3] (2,10) rectangle ++(1,1);
  \fill[job3] (2,11) rectangle ++(1,1);
  \fill[job3] (2,12) rectangle ++(1,1);
  \fill[job1] (3,0) rectangle ++(1,1);
  \fill[job1] (3,1) rectangle ++(1,1);
  \fill[job1] (3,2) rectangle ++(1,1);
  \fill[job1] (3,3) rectangle ++(1,1);
  \fill[job2] (3,5) rectangle ++(1,1);
  \fill[job2] (3,6) rectangle ++(1,1);
  \fill[job2] (3,7) rectangle ++(1,1);
  \fill[job2] (3,8) rectangle ++(1,1);
  \fill[job3] (3,10) rectangle ++(1,1);
  \fill[job3] (3,11) rectangle ++(1,1);
  \fill[job3] (3,12) rectangle ++(1,1);
  \fill[job3] (3,13) rectangle ++(1,1);
  \fill[job1] (4,0) rectangle ++(1,1);
  \fill[job1] (4,1) rectangle ++(1,1);
  \fill[job1] (4,2) rectangle ++(1,1);
  \fill[job1] (4,3) rectangle ++(1,1);
  \fill[job1] (4,4) rectangle ++(1,1);
  \fill[job2] (4,5) rectangle ++(1,1);
  \fill[job2] (4,6) rectangle ++(1,1);
  \fill[job2] (4,7) rectangle ++(1,1);
  \fill[job2] (4,8) rectangle ++(1,1);
  \fill[job2] (4,9) rectangle ++(1,1);
  \fill[job3] (4,10) rectangle ++(1,1);
  \fill[job3] (4,11) rectangle ++(1,1);
  \fill[job3] (4,12) rectangle ++(1,1);
  \fill[job3] (4,13) rectangle ++(1,1);
  \fill[job3] (4,14) rectangle ++(1,1);
  \fill[job4] (5,0) rectangle ++(1,1);
  \fill[job5] (5,5) rectangle ++(1,1);
  \fill[job6] (5,10) rectangle ++(1,1);
  \fill[job4] (6,0) rectangle ++(1,1);
  \fill[job4] (6,1) rectangle ++(1,1);
  \fill[job5] (6,5) rectangle ++(1,1);
  \fill[job5] (6,6) rectangle ++(1,1);
  \fill[job6] (6,10) rectangle ++(1,1);
  \fill[job6] (6,11) rectangle ++(1,1);
  \fill[job4] (7,0) rectangle ++(1,1);
  \fill[job4] (7,1) rectangle ++(1,1);
  \fill[job4] (7,2) rectangle ++(1,1);
  \fill[job5] (7,5) rectangle ++(1,1);
  \fill[job5] (7,6) rectangle ++(1,1);
  \fill[job5] (7,7) rectangle ++(1,1);
  \fill[job6] (7,10) rectangle ++(1,1);
  \fill[job6] (7,11) rectangle ++(1,1);
  \fill[job6] (7,12) rectangle ++(1,1);
  \fill[job4] (8,0) rectangle ++(1,1);
  \fill[job4] (8,1) rectangle ++(1,1);
  \fill[job4] (8,2) rectangle ++(1,1);
  \fill[job4] (8,3) rectangle ++(1,1);
  \fill[job5] (8,5) rectangle ++(1,1);
  \fill[job5] (8,6) rectangle ++(1,1);
  \fill[job5] (8,7) rectangle ++(1,1);
  \fill[job5] (8,8) rectangle ++(1,1);
  \fill[job6] (8,10) rectangle ++(1,1);
  \fill[job6] (8,11) rectangle ++(1,1);
  \fill[job6] (8,12) rectangle ++(1,1);
  \fill[job6] (8,13) rectangle ++(1,1);
  \fill[job4] (9,0) rectangle ++(1,1);
  \fill[job4] (9,1) rectangle ++(1,1);
  \fill[job4] (9,2) rectangle ++(1,1);
  \fill[job4] (9,3) rectangle ++(1,1);
  \fill[job4] (9,4) rectangle ++(1,1);
  \fill[job5] (9,5) rectangle ++(1,1);
  \fill[job5] (9,6) rectangle ++(1,1);
  \fill[job5] (9,7) rectangle ++(1,1);
  \fill[job5] (9,8) rectangle ++(1,1);
  \fill[job5] (9,9) rectangle ++(1,1);
  \fill[job6] (9,10) rectangle ++(1,1);
  \fill[job6] (9,11) rectangle ++(1,1);
  \fill[job6] (9,12) rectangle ++(1,1);
  \fill[job6] (9,13) rectangle ++(1,1);
  \fill[job6] (9,14) rectangle ++(1,1);
  \fill[job7] (10,0) rectangle ++(1,1);
  \fill[job8] (10,5) rectangle ++(1,1);
  \fill[job9] (10,10) rectangle ++(1,1);
  \fill[job7] (11,0) rectangle ++(1,1);
  \fill[job7] (11,1) rectangle ++(1,1);
  \fill[job8] (11,5) rectangle ++(1,1);
  \fill[job8] (11,6) rectangle ++(1,1);
  \fill[job9] (11,10) rectangle ++(1,1);
  \fill[job9] (11,11) rectangle ++(1,1);
  \fill[job7] (12,0) rectangle ++(1,1);
  \fill[job7] (12,1) rectangle ++(1,1);
  \fill[job7] (12,2) rectangle ++(1,1);
  \fill[job8] (12,5) rectangle ++(1,1);
  \fill[job8] (12,6) rectangle ++(1,1);
  \fill[job8] (12,7) rectangle ++(1,1);
  \fill[job9] (12,10) rectangle ++(1,1);
  \fill[job9] (12,11) rectangle ++(1,1);
  \fill[job9] (12,12) rectangle ++(1,1);
  \fill[job7] (13,0) rectangle ++(1,1);
  \fill[job7] (13,1) rectangle ++(1,1);
  \fill[job7] (13,2) rectangle ++(1,1);
  \fill[job7] (13,3) rectangle ++(1,1);
  \fill[job8] (13,5) rectangle ++(1,1);
  \fill[job8] (13,6) rectangle ++(1,1);
  \fill[job8] (13,7) rectangle ++(1,1);
  \fill[job8] (13,8) rectangle ++(1,1);
  \fill[job9] (13,10) rectangle ++(1,1);
  \fill[job9] (13,11) rectangle ++(1,1);
  \fill[job9] (13,12) rectangle ++(1,1);
  \fill[job9] (13,13) rectangle ++(1,1);
  \fill[job7] (14,0) rectangle ++(1,1);
  \fill[job7] (14,1) rectangle ++(1,1);
  \fill[job7] (14,2) rectangle ++(1,1);
  \fill[job7] (14,3) rectangle ++(1,1);
  \fill[job7] (14,4) rectangle ++(1,1);
  \fill[job8] (14,5) rectangle ++(1,1);
  \fill[job8] (14,6) rectangle ++(1,1);
  \fill[job8] (14,7) rectangle ++(1,1);
  \fill[job8] (14,8) rectangle ++(1,1);
  \fill[job8] (14,9) rectangle ++(1,1);
  \fill[job9] (14,10) rectangle ++(1,1);
  \fill[job9] (14,11) rectangle ++(1,1);
  \fill[job9] (14,12) rectangle ++(1,1);
  \fill[job9] (14,13) rectangle ++(1,1);
  \fill[job9] (14,14) rectangle ++(1,1);
  \fill[job10] (15,0) rectangle ++(1,1);
  \fill[job11] (15,5) rectangle ++(1,1);
  \fill[job12] (15,10) rectangle ++(1,1);
  \fill[job10] (16,0) rectangle ++(1,1);
  \fill[job10] (16,1) rectangle ++(1,1);
  \fill[job11] (16,5) rectangle ++(1,1);
  \fill[job11] (16,6) rectangle ++(1,1);
  \fill[job12] (16,10) rectangle ++(1,1);
  \fill[job12] (16,11) rectangle ++(1,1);
  \fill[job10] (17,0) rectangle ++(1,1);
  \fill[job10] (17,1) rectangle ++(1,1);
  \fill[job10] (17,2) rectangle ++(1,1);
  \fill[job11] (17,5) rectangle ++(1,1);
  \fill[job11] (17,6) rectangle ++(1,1);
  \fill[job11] (17,7) rectangle ++(1,1);
  \fill[job12] (17,10) rectangle ++(1,1);
  \fill[job12] (17,11) rectangle ++(1,1);
  \fill[job12] (17,12) rectangle ++(1,1);
  \fill[job10] (18,0) rectangle ++(1,1);
  \fill[job10] (18,1) rectangle ++(1,1);
  \fill[job10] (18,2) rectangle ++(1,1);
  \fill[job10] (18,3) rectangle ++(1,1);
  \fill[job11] (18,5) rectangle ++(1,1);
  \fill[job11] (18,6) rectangle ++(1,1);
  \fill[job11] (18,7) rectangle ++(1,1);
  \fill[job11] (18,8) rectangle ++(1,1);
  \fill[job12] (18,10) rectangle ++(1,1);
  \fill[job12] (18,11) rectangle ++(1,1);
  \fill[job12] (18,12) rectangle ++(1,1);
  \fill[job12] (18,13) rectangle ++(1,1);
  \fill[job10] (19,0) rectangle ++(1,1);
  \fill[job10] (19,1) rectangle ++(1,1);
  \fill[job10] (19,2) rectangle ++(1,1);
  \fill[job10] (19,3) rectangle ++(1,1);
  \fill[job10] (19,4) rectangle ++(1,1);
  \fill[job11] (19,5) rectangle ++(1,1);
  \fill[job11] (19,6) rectangle ++(1,1);
  \fill[job11] (19,7) rectangle ++(1,1);
  \fill[job11] (19,8) rectangle ++(1,1);
  \fill[job11] (19,9) rectangle ++(1,1);
  \fill[job12] (19,10) rectangle ++(1,1);
  \fill[job12] (19,11) rectangle ++(1,1);
  \fill[job12] (19,12) rectangle ++(1,1);
  \fill[job12] (19,13) rectangle ++(1,1);
  \fill[job12] (19,14) rectangle ++(1,1);
  \fill[job13] (20,0) rectangle ++(1,1);
  \fill[job14] (20,5) rectangle ++(1,1);
  \fill[job15] (20,10) rectangle ++(1,1);
  \fill[job13] (21,0) rectangle ++(1,1);
  \fill[job13] (21,1) rectangle ++(1,1);
  \fill[job14] (21,5) rectangle ++(1,1);
  \fill[job14] (21,6) rectangle ++(1,1);
  \fill[job15] (21,10) rectangle ++(1,1);
  \fill[job15] (21,11) rectangle ++(1,1);
  \fill[job13] (22,0) rectangle ++(1,1);
  \fill[job13] (22,1) rectangle ++(1,1);
  \fill[job13] (22,2) rectangle ++(1,1);
  \fill[job14] (22,5) rectangle ++(1,1);
  \fill[job14] (22,6) rectangle ++(1,1);
  \fill[job14] (22,7) rectangle ++(1,1);
  \fill[job15] (22,10) rectangle ++(1,1);
  \fill[job15] (22,11) rectangle ++(1,1);
  \fill[job15] (22,12) rectangle ++(1,1);
  \fill[job13] (23,0) rectangle ++(1,1);
  \fill[job13] (23,1) rectangle ++(1,1);
  \fill[job13] (23,2) rectangle ++(1,1);
  \fill[job13] (23,3) rectangle ++(1,1);
  \fill[job14] (23,5) rectangle ++(1,1);
  \fill[job14] (23,6) rectangle ++(1,1);
  \fill[job14] (23,7) rectangle ++(1,1);
  \fill[job14] (23,8) rectangle ++(1,1);
  \fill[job15] (23,10) rectangle ++(1,1);
  \fill[job15] (23,11) rectangle ++(1,1);
  \fill[job15] (23,12) rectangle ++(1,1);
  \fill[job15] (23,13) rectangle ++(1,1);
  \fill[job13] (24,0) rectangle ++(1,1);
  \fill[job13] (24,1) rectangle ++(1,1);
  \fill[job13] (24,2) rectangle ++(1,1);
  \fill[job13] (24,3) rectangle ++(1,1);
  \fill[job13] (24,4) rectangle ++(1,1);
  \fill[job14] (24,5) rectangle ++(1,1);
  \fill[job14] (24,6) rectangle ++(1,1);
  \fill[job14] (24,7) rectangle ++(1,1);
  \fill[job14] (24,8) rectangle ++(1,1);
  \fill[job14] (24,9) rectangle ++(1,1);
  \fill[job15] (24,10) rectangle ++(1,1);
  \fill[job15] (24,11) rectangle ++(1,1);
  \fill[job15] (24,12) rectangle ++(1,1);
  \fill[job15] (24,13) rectangle ++(1,1);
  \fill[job15] (24,14) rectangle ++(1,1);
\draw[black, line width=0.2pt] (0,0) rectangle (25,15);
  \draw[black!60, line width=0.2pt] (1,0)--(1,15);
  \draw[black!60, line width=0.2pt] (2,0)--(2,15);
  \draw[black!60, line width=0.2pt] (3,0)--(3,15);
  \draw[black!60, line width=0.2pt] (4,0)--(4,15);
  \draw[black!60, line width=0.2pt] (5,0)--(5,15);
  \draw[black!60, line width=0.2pt] (6,0)--(6,15);
  \draw[black!60, line width=0.2pt] (7,0)--(7,15);
  \draw[black!60, line width=0.2pt] (8,0)--(8,15);
  \draw[black!60, line width=0.2pt] (9,0)--(9,15);
  \draw[black!60, line width=0.2pt] (10,0)--(10,15);
  \draw[black!60, line width=0.2pt] (11,0)--(11,15);
  \draw[black!60, line width=0.2pt] (12,0)--(12,15);
  \draw[black!60, line width=0.2pt] (13,0)--(13,15);
  \draw[black!60, line width=0.2pt] (14,0)--(14,15);
  \draw[black!60, line width=0.2pt] (15,0)--(15,15);
  \draw[black!60, line width=0.2pt] (16,0)--(16,15);
  \draw[black!60, line width=0.2pt] (17,0)--(17,15);
  \draw[black!60, line width=0.2pt] (18,0)--(18,15);
  \draw[black!60, line width=0.2pt] (19,0)--(19,15);
  \draw[black!60, line width=0.2pt] (20,0)--(20,15);
  \draw[black!60, line width=0.2pt] (21,0)--(21,15);
  \draw[black!60, line width=0.2pt] (22,0)--(22,15);
  \draw[black!60, line width=0.2pt] (23,0)--(23,15);
  \draw[black!60, line width=0.2pt] (24,0)--(24,15);
  \draw[black!60, line width=0.2pt] (0,1)--(25,1);
  \draw[black!60, line width=0.2pt] (0,2)--(25,2);
  \draw[black!60, line width=0.2pt] (0,3)--(25,3);
  \draw[black!60, line width=0.2pt] (0,4)--(25,4);
  \draw[black!60, line width=0.2pt] (0,5)--(25,5);
  \draw[black!60, line width=0.2pt] (0,6)--(25,6);
  \draw[black!60, line width=0.2pt] (0,7)--(25,7);
  \draw[black!60, line width=0.2pt] (0,8)--(25,8);
  \draw[black!60, line width=0.2pt] (0,9)--(25,9);
  \draw[black!60, line width=0.2pt] (0,10)--(25,10);
  \draw[black!60, line width=0.2pt] (0,11)--(25,11);
  \draw[black!60, line width=0.2pt] (0,12)--(25,12);
  \draw[black!60, line width=0.2pt] (0,13)--(25,13);
  \draw[black!60, line width=0.2pt] (0,14)--(25,14);
\node[anchor=north] at (12.50,-0.15) {\small round $t = {0, \dots, 24}$};
  \node[rotate=90,anchor=south] at (-0.35,7.50) {\small memory ($\kvmem = 15$)};
\end{tikzpicture} \label{subfig:schedule comparison:simultaneous}
}
  \caption{Illustration of \Cref{example:schedule comparison}. Each color represents the memory usage of a single job across rounds.
  }
\label{fig:schedule comparison}
\end{figure}
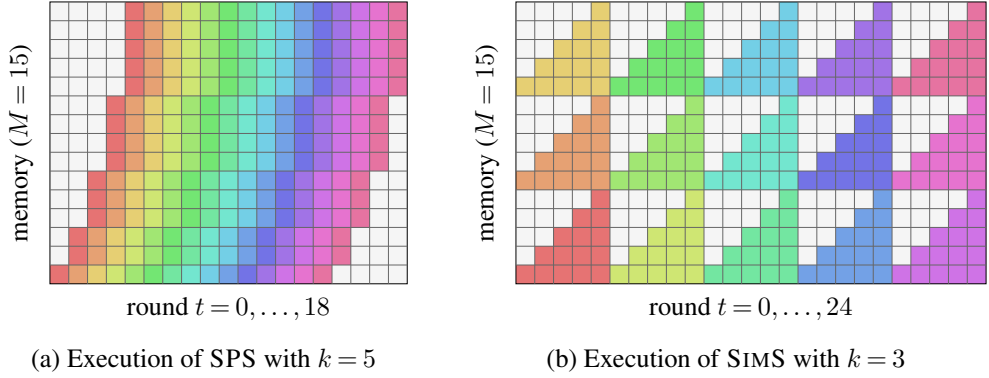

\xhdr{Theoretical Guarantees of {\SPS}.}
Through \Cref{example:schedule comparison} and \Cref{fig:schedule comparison}, we observe that {\SPS} achieves strong performance in this simple instance.  
Compared to {\SimPipelineSchedule}, {\SPS} exhibits significantly smoother memory usage over time, which in turn enables a higher degree of parallelism under the same memory budget. We now formalize these observations with the following theoretical guarantees for {\SPS}.

\begin{theorem}[Approximation of {\SPS} for Identical Jobs]
\label{thm:staggered:identical job}
Let $\Jobs$ be a sequence of jobs with identical prompt length $\initialLen$ and identical response length $\responseLen$. 
Then the schedule produced by \textsc{StaggeredPipelineSchedule} with a proper choice of parameters 
has an approximation ratio at most $2$. Moreover, when prompt length $\initialLen$ and response length $\responseLen$ are fixed and both KV-cache memory budget $\kvmem \to \infty$ and the number of jobs $|\Jobs| \to \infty$ with $\kvmem = o(|\Jobs|)$, the approximation ratio converges to $1 + o(1)$.
\end{theorem}

We remark that the above theorem guarantees both the asymptotic optimality of {\SPS} and a worst-case approximation ratio of $2$. In contrast, as shown in \Cref{apx:simps-lower-bounds}, {\SimPipelineSchedule} incurs an approximation ratio of at least $2$ even in this asymptotic regime.

In the following, we prove \Cref{thm:staggered:identical job}. We begin by establishing two key lemmas concerning the memory usage and the maximum feasible degree of parallelism in {\SPS}.
\begin{restatable}[Peak Memory Usage of {\SPS}]{lemma}{lemPeakMem}
\label{lem:staggered:peak-memory}
For any sequence of jobs with identical prompt length $\initialLen \in \naturals$, any degree of parallelism $\paral \in \naturals$ and slice length $\sliceLen \in \naturals$, running {\SPS}$(\Jobs, \paral, \sliceLen)$ indefinitely (i.e., $|\Jobs| \rightarrow \infty$) requires memory
\begin{align*}
\MemPeak(\paral, \sliceLen, \initialLen) \;\triangleq\;
\initialLen \cdot \paral \;+\; \frac{1}{2} \left( 
\sliceLen \cdot \paral + \sliceLen + \paral - \gcd(\sliceLen, \paral) 
\right).
\end{align*}
Moreover, $\MemPeak(\paral, \sliceLen, \initialLen)$ is increasing in $\paral$ when fixing $\sliceLen$ and $\initialLen$.
\end{restatable}

\begin{restatable}[Feasible Parallelism of {\SPS}]{lemma}{lemFeasibleParal}
\label{lem:staggered:max-parallelism}
Fix a job sequence $\Jobs$ with identical prompt length $\initialLen \in \naturals$ and a slice length $\sliceLen \in \naturals$. 
Define the \emph{maximum memory-feasible degree of parallelism} of {\SPipelineSchedule} ({\SPS}) as
\begin{align}
\label{eq:max-parallelism}
\paralSliceLen \;\triangleq\; \max\left\{ \paral \in \naturals \condition \MemPeak(\paral, \sliceLen, \initialLen) \leq \kvmem \right\}.
\end{align}
where $\MemPeak(\paral, \sliceLen, \initialLen)$ is the peak memory usage of the schedule produced by {\SPS} defined in \Cref{lem:staggered:peak-memory}.
Then, the maximum memory-feasible degree of parallelism $\paralSliceLen$ can be lower bounded as 
\begin{align*}
\paralSliceLen \;\geq\; 
\left\lfloor 
\frac{2\kvmem - \sliceLen + 1}{\,2\initialLen + \sliceLen + 1\,} 
\right\rfloor.
\end{align*}
\end{restatable}

\begin{proof}[Proof of \Cref{lem:staggered:peak-memory}]
Let the jobs be indexed $\jobIdx = 0, 1, \dots$ according to the start time order assigned by {\SPS}. Recall that the start time of job $\jobIdx$ is $\startTimeOf{\jobIdx} = \lfloor \jobIdx \cdot \sliceLen / \paral \rfloor$, and each job runs for $\sliceLen$ rounds.

Since $\startTimeOf{\jobIdx+\paral} = \lfloor (\jobIdx+\paral)\sliceLen/\paral \rfloor = \startTimeOf{\jobIdx} + \sliceLen$, job $\jobIdx$ completes immediately before job $\jobIdx+\paral$ starts. Therefore, the set of active jobs $\ActiveBatch_t$ at any time $t = q\sliceLen + \rho$ (with any integers $q \geq 1$ and $0\leq \rho \leq \sliceLen-1$) always consists of exactly $\paral$ jobs. 
Each such job can be uniquely associated with a relative index $x \in \{0, \dots, \paral-1\}$ such that its start time is either $q\sliceLen + \delta_x$ or $(q-1)\sliceLen + \delta_x$, where $\delta_x \triangleq \lfloor x\sliceLen/\paral \rfloor$ represents the start time offset within a period.
Therefore, we can count the total memory usage, $S(t)$, at time $t$:
\begin{align*}
S(t) &= \sum\nolimits_{\jobIdx \in \ActiveBatch_t} (\initialLen + t - \startTimeOf{\jobIdx} + 1) 
= \paral \cdot \initialLen + \sum\nolimits_{x : \delta_x \le \rho} (\rho - \delta_x + 1) + \sum\nolimits_{x : \delta_x > \rho} (\sliceLen + \rho - \delta_x + 1) \\
&= \paral \cdot \initialLen + \paral(\rho + 1) - \sum\nolimits_{x=0}^{\paral-1} \delta_x + \sum\nolimits_{x : \delta_x > \rho} \sliceLen 
= \paral \cdot \initialLen + \paral(\rho + 1) - \sum\nolimits_{x=0}^{\paral-1} \delta_x + (\paral - c(\rho))\sliceLen, \\
&= \paral \cdot \initialLen + \paral\sliceLen - \sum\nolimits_{x=0}^{\paral-1} \delta_x + \paral(\rho+1) - c(\rho)\sliceLen
\end{align*}
where auxiliary notation $c(\rho) \triangleq |\{x \in \{0, \dots, \paral-1\} : \delta_x \le \rho\}|$ is the number of jobs that start at or before offset $\rho$ within a period.
Note that $c(\rho) = \lceil \paral(\rho+1)/\sliceLen \rceil$, which implies $c(\rho)\cdot \sliceLen \ge \paral(\rho+1)$. Thus, the term $\paral(\rho + 1) - c(\rho)\sliceLen$ is non-positive and the peak memory usage is bounded by:
\begin{align*}
\MemPeak(\paral, \sliceLen, \initialLen) = \max_{0\leq \rho\leq \sliceLen - 1}~ S(t) \le \paral \cdot \initialLen + \paral\sliceLen - \sum\nolimits_{x=0}^{\paral-1} \lfloor x\sliceLen/\paral \rfloor.
\end{align*}
The sum $\sum\nolimits_{x=0}^{\paral-1} \lfloor x\sliceLen/\paral \rfloor$ counts the number of positive integer points $(i, j) \in \naturals^2$ such that $0 < i < \paral$ and $0 < j \le i \cdot \sliceLen/\paral$. This corresponds to the number of integer points strictly inside the rectangle defined by $(0,0)$ and $(\paral, \sliceLen)$ that lie on or below the diagonal line connecting these two points. The number of integer points on the diagonal itself, excluding the endpoints, is $\gcd(\paral, \sliceLen)-1$ (where $\gcd(\cdot,\cdot)$ is the greatest common divisor). By symmetry, the number of points below the diagonal is half the number of off-diagonal points. Therefore, the sum is:
\begin{align*}
\sum\nolimits_{x=0}^{\paral-1} \lfloor x\sliceLen/\paral \rfloor = \frac{\overbrace{(\paral-1)(\sliceLen-1)}^{\text{points in rectangle}} + \overbrace{(\gcd(\paral, \sliceLen)-1)}^{\text{points on diagonal}}}{2}.
\end{align*}
Substituting this into the bound for $\MemPeak(\paral, \sliceLen, \initialLen)$:
\begin{align*}
\MemPeak(\paral, \sliceLen, \initialLen) &\leq \paral \cdot \initialLen + \paral\sliceLen - \frac{(\paral-1)(\sliceLen-1) + \gcd(\paral, \sliceLen)-1}{2} 
= \initialLen \cdot \paral + \frac{\paral\sliceLen + \paral + \sliceLen - \gcd(\paral, \sliceLen)}{2}.
\end{align*}
and the equality is attained at every $t = q\sliceLen - 1$ for $q \geq 1$.

To prove that $\MemPeak(\paral, \sliceLen, \initialLen)$ is monotone in $\paral$ when fixing $\sliceLen$ and $\initialLen$, we have
$$
\MemPeak(\paral + 1, \sliceLen, \initialLen) - \MemPeak(\paral, \sliceLen, \initialLen) = s + \frac{\sliceLen + 1 - \gcd(\paral + 1, \sliceLen) + \gcd(\paral, \sliceLen)}{2} \geq s + \frac{\sliceLen + 1 - \sliceLen + 1}{2} \geq s + 1,
$$
which completes the proof of \Cref{lem:staggered:peak-memory} as desired.
\end{proof}

\begin{proof}[Proof of \Cref{lem:staggered:max-parallelism}]
By \Cref{lem:staggered:peak-memory},
the condition $\MemPeak(\paral, \sliceLen, \initialLen) \leq \kvmem$ is equivalent to
\begin{align*}
\paral (2\initialLen + \sliceLen + 1) \leq 2\kvmem - \sliceLen + \gcd(\sliceLen, \paral).
\end{align*}
By $\gcd(\sliceLen, \paral) \geq 1$, we can choose
\begin{align*}
\left\lfloor \frac{2\kvmem - \sliceLen + 1}{2\initialLen + \sliceLen + 1} \right\rfloor \leq \frac{2\kvmem - \sliceLen + \gcd(\sliceLen, \paral)}{2\initialLen + \sliceLen + 1}
\end{align*}
as a valid lower bound for the maximum memory-feasible degree of parallelism $\paralSliceLen$.
\end{proof}

We now proceed to prove \Cref{thm:staggered:identical job}.  
At the heart of our analysis is the notion of \emph{memory-time area} (\Cref{def:area}), which quantifies the total KV-cache resources consumed by a job over its execution: it accounts for both the fixed prompt memory ($\initialLen$) and the linearly growing cache from $\responseLen$ generated tokens. 
See the colored triangles (or more generally, trapezoids) in \Cref{fig:schedule comparison}, which illustrates this memory-time resource consumption.  
This geometric perspective yields a clean resource accounting principle: the total memory-time area consumed by any feasible schedule cannot exceed the area supplied by the memory budget $\kvmem$ over time.

\begin{definition}[Memory-Time Area]
\label{def:area}
For a job with prompt length $\initialLen$ and response length $\responseLen$, its \emph{memory-time area} is defined as
\begin{align*}
\Area(\responseLen,\initialLen) \;\triangleq\; \initialLen \cdot \responseLen + \frac{\responseLen(\responseLen + 1)}{2}.
\end{align*}
This quantity represents the total memory-time resources consumed by the job over its execution.
\end{definition}

Our goal is to compare the packing efficiency of the optimal scheduler with that of {\SPS}.
Leveraging the geometric perspective above, \Cref{lem:area-lower-bound} establishes a fundamental lower bound on the optimal flow time by observing that even the best schedule must allocate at least $\Area(\sliceLen,\initialLen)$ area per job (for jobs with response length $\geq \sliceLen$), and this area must be packed into a time horizon constrained by $\kvmem$.  
Conversely, \Cref{lem:aux:ratio} captures the \emph{packing efficiency} of {\SPS}: it shows that the average time per job in {\SPS} is nearly proportional to the area-per-unit-memory, up to a small constant factor depending on the degree of parallelism.  
Together, these lemmas bridge memory usage and flow time, allowing us to compare {\SPS} directly to the information-theoretic optimum via area-based reasoning.

\begin{restatable}[Area-Based Lower Bound on Optimal Flow Time]{lemma}{lemAreaOptLB}
\label{lem:area-lower-bound}
Consider any instance with $\NumberJobs$ jobs, each having prompt length $\initialLen$ and response length at least $\sliceLen$.  
Then the optimal total flow time $\OPT$ satisfies
\begin{align*}
\OPT \;\geq\; \frac{\NumberJobs(\NumberJobs + 1)}{2} \cdot \frac{\Area(\sliceLen,\initialLen)}{\kvmem}.
\end{align*}
\end{restatable}

\begin{restatable}[Area-Based Packing Efficiency of {\SPS}]{lemma}{lemAreaSpsUB}
\label{lem:aux:ratio}
For any integers $\initialLen \geq 0$ and $\sliceLen \geq 1$, let $\paralSliceLen$ denote the maximum memory-feasible degree of parallelism defined in \Cref{lem:staggered:max-parallelism} (see Eqn.~\eqref{eq:max-parallelism}). Then
\begin{align*}
\frac{\sliceLen}{\paralSliceLen} \;<\; \left(1 + \frac{2}{\paralSliceLen}\right) \cdot \frac{\Area(\sliceLen,\initialLen)}{\kvmem}.
\end{align*}
\end{restatable}

To gain more intuition behind these bounds, observe that the ratio $\kvmem / \Area(\sliceLen,\initialLen)$ can be interpreted as an upper bound on the \emph{throughput} (jobs per round) achievable by any feasible schedule---particularly the optimal clairvoyant scheduler. Meanwhile, $\paralSliceLen / \sliceLen$ represents the effective throughput of {\SPS}. In this sense, \Cref{lem:area-lower-bound} lower bound the optimal total flow time via its throughput limit, while \Cref{lem:aux:ratio} shows that {\SPS} achieves nearly optimal throughput, up to a small multiplicative factor in the parallelism dimension.

Below, we prove \Cref{lem:area-lower-bound,lem:aux:ratio}. Using these results, we establish the asymptotic approximation ratio of $1 + o(1)$ stated in \Cref{thm:staggered:identical job} for the large-memory regime. The proof of the (non-asymptotic) approximation ratio of $2$ for general identical-job instances---also stated in \Cref{thm:staggered:identical job}---is deferred to \Cref{apx:sps-paral-1}.

\begin{proof}[Proof of \Cref{lem:area-lower-bound}]
Let $\C_1 \leq \C_2 \leq \dots \leq \C_{\NumberJobs}$ denote the completion times of the jobs in an optimal schedule.  
For each index $i \in [\NumberJobs]$, consider the first $i$ jobs to complete (i.e., those with completion times $\C_1, \dots, \C_i$).  
Each of these $i$ jobs has response length at least $\sliceLen$, and thus consumes memory-time area at least $\Area(\sliceLen,\initialLen)$ (by \Cref{def:area}).  
Therefore, the total memory-time area consumed by these $i$ jobs is at least $i \cdot \Area(\sliceLen,\initialLen)$.
On the other hand, during the first $\C_i$ rounds, the total memory available per round is at most $\kvmem$, so the total memory-time area supplied by the system up to time $\C_i$ is at most $\kvmem \cdot \C_i$.
Since the consumed area cannot exceed the supplied area, we have
$i \cdot \Area(\sliceLen,\initialLen) \;\leq\; \kvmem \cdot \C_i$,
which implies
$\C_i \geq \frac{i \cdot \Area(\sliceLen,\initialLen)}{\kvmem}$.
Summing over all $i \in [\NumberJobs]$ yields
\begin{align*}
\OPT \;=\; \sum\nolimits_{i=1}^{\NumberJobs} \C_i 
\;\geq\; \sum\nolimits_{i=1}^{\NumberJobs} \frac{i \cdot \Area(\sliceLen,\initialLen)}{\kvmem}
\;=\; \frac{\NumberJobs(\NumberJobs + 1)}{2} \cdot \frac{\Area(\sliceLen,\initialLen)}{\kvmem},
\end{align*}
as claimed.
\end{proof}

\begin{proof}[Proof of \Cref{lem:aux:ratio}]
Define axuliary notation $\paral\triangleq \paralSliceLen$.
It suffices to show that
\begin{align*}
{(\paral + 2)\bigl(\initialLen + \tfrac{\sliceLen + 1}{2}\bigr)}{} \;>\; \kvmem.
\end{align*}
By definition of $\paral$ (i.e., maximum memory-feasible degree of parallelism), we have $\MemPeak(\paral + 1, \sliceLen, \initialLen) > \kvmem$.  
Applying \Cref{lem:staggered:peak-memory},
\begin{align*}
\kvmem 
&< \initialLen(\paral + 1) + \frac{\sliceLen(\paral + 1) + \sliceLen + (\paral + 1) - \gcd(\paral + 1, \sliceLen)}{2} \\
&\leq \initialLen(\paral + 1) + \frac{\sliceLen(\paral + 1) + \sliceLen + (\paral + 1) - 1}{2} 
= \initialLen\paral + \initialLen + \frac{\sliceLen\paral}{2} + \frac{\paral}{2} + \sliceLen.
\end{align*}
Now observe that
\begin{align*}
(\paral + 2)\Bigl(\initialLen + \frac{\sliceLen + 1}{2}\Bigr)
= \initialLen\paral + \initialLen + \frac{\sliceLen\paral}{2} + \frac{\paral}{2} + \sliceLen + (\initialLen + 1).
\end{align*}
Since $\initialLen \geq 0$, the numerator exceeds the upper bound on $\kvmem$ derived above. Hence,
\begin{align*}
{(\paral + 2)\bigl(\initialLen + \tfrac{\sliceLen + 1}{2}\bigr)}{} > \kvmem,
\end{align*}
which is equivalent to the desired inequality.
\end{proof}

\begin{proof}[Proof of approximation ratio $1 + o(1)$ in \Cref{thm:staggered:identical job}]
Let $\NumberJobs = |\Jobs|$. Under {\SPipelineSchedule} ({\SPS}) with parameters $(\paral, \sliceLen)$, the completion time of the $i$-th job (indexed from $0$) is $\lfloor i\sliceLen / \paral \rfloor + \sliceLen$. Thus,
\begin{equation}
\label{eq:Fstag:identical}
\FTime{\textsc{SPS}} 
= \sum\nolimits_{i=0}^{\NumberJobs-1} \left( \left\lfloor \frac{i\sliceLen}{\paral} \right\rfloor + \sliceLen \right)
\leq \NumberJobs \sliceLen + \frac{\NumberJobs(\NumberJobs - 1)\sliceLen}{2\paral}.
\end{equation}
Combining this with the area-based lower bound from \Cref{lem:area-lower-bound},
\begin{align*}
\OPT \;\geq\; \frac{\NumberJobs(\NumberJobs + 1)}{2} \cdot \frac{\Area(\sliceLen,\initialLen)}{\kvmem},
\end{align*}
we decompose the approximation ratio as
\begin{align*}
\frac{\FTime{{\SPS}}}{\OPT} \;\leq\; \underbrace{\frac{\NumberJobs\sliceLen}{\OPT}}_{\text{(I)}} \;+\; \underbrace{\frac{\dfrac{\NumberJobs(\NumberJobs-1)\sliceLen}{2\paral}}{\OPT}}_{\text{(II)}}.
\end{align*}

\xhdr{Bounding Term (I).}
Using the trivial lower bound $\OPT \geq \NumberJobs\sliceLen$, we have $\text{(I)} \leq 1$.  
Alternatively, applying \Cref{lem:area-lower-bound} and then \Cref{lem:aux:ratio},
\begin{align*}
\text{(I)} 
&\leq \frac{\NumberJobs\sliceLen}{\frac{\NumberJobs(\NumberJobs+1)}{2} \cdot \frac{\Area(\sliceLen,\initialLen)}{\kvmem}}
= \frac{2\sliceLen}{(\NumberJobs+1) \cdot \frac{\Area(\sliceLen,\initialLen)}{\kvmem}} 
< \frac{2\sliceLen}{(\NumberJobs+1) \cdot \frac{\paral}{\sliceLen(1 + 2/\paral)}}
= \frac{2\paral + 4}{\NumberJobs + 1}.
\end{align*}
Hence,
\begin{align*}
\text{(I)} \;\leq\; \min\!\left\{1,\ \frac{2\paral + 4}{\NumberJobs + 1}\right\}.
\end{align*}

\xhdr{Bounding Term (II).}
Again using \Cref{lem:area-lower-bound} and \Cref{lem:aux:ratio},
\begin{align*}
\text{(II)} 
&\leq \frac{\frac{\NumberJobs(\NumberJobs-1)\sliceLen}{2\paral}}{\frac{\NumberJobs(\NumberJobs+1)}{2} \cdot \frac{\Area(\sliceLen,\initialLen)}{\kvmem}}
= \frac{\NumberJobs - 1}{\NumberJobs + 1} \cdot \frac{\kvmem\sliceLen}{\paral \cdot \Area(\sliceLen,\initialLen)} 
< \frac{\NumberJobs - 1}{\NumberJobs + 1} \cdot \left(1 + \frac{2}{\paral}\right).
\end{align*}

\xhdr{Combining the Bounds.}
Summing the two terms yields
\begin{align}
\label{eq:staggered:identical:final-ratio}
\frac{\FTime{{\SPS}}}{\OPT}
\;<\;
\min\!\left\{1,\ \frac{2\paral + 4}{\NumberJobs + 1}\right\}
\;+\;
\frac{\NumberJobs - 1}{\NumberJobs + 1} \cdot \left(1 + \frac{2}{\paral}\right).
\end{align}
When the prompt length $\initialLen$ and response length $\responseLen$ are fixed, the dominant term in the right-hand side of \eqref{eq:staggered:identical:final-ratio} as $\NumberJobs \to \infty$ is term (II), which approaches $1 + 2/\paral$.
More precisely, consider the regime where both the KV-cache memory budget $\kvmem \to \infty$ and the number of jobs $\NumberJobs \to \infty$ satisfy $\kvmem = o(\NumberJobs)$.
Since $\paral = \Theta(\kvmem)$, we have $\paral \to \infty$ and $\paral/\NumberJobs \to 0$.
Thus, term~(I) is bounded by $O(\paral/\NumberJobs) = o(1)$, while term~(II) approaches $1$. Consequently, the overall approximation ratio converges to $1 + o(1)$, as desired. 
\end{proof}

\begin{remark}
     It can be verified that the right-hand side of \eqref{eq:staggered:identical:final-ratio} is at most $2 + \frac{1}{\paral}$. Thus, this also shows that the approximation ratio is at most $3$ for general identical-job instances under memory-time area analysis framework. To tighten this bound to $2$, we require a specialized analysis that exploits the structural properties of {\OPT} for identical jobs. Its formal proof can be found in \Cref{apx:sps-paral-1}.
\end{remark}

\subsection{Algorithm Overview of {\GeoSALG}}
\label{subsec:main alg}

In this section, we present the {\GeoSALG}. We start by giving an algorithm overview of it and then state it competitive ratio guarantee in \Cref{thm:GSA}.
Its competitive ratio analysis is deferred to \Cref{sec:clairvoyant greedy,sec:algorithms}.

\begin{algorithm2e}[ht]
\caption{{\GeoSALG} ({\GSA})}
\label{alg:gsa}
\KwIn{
    $\NumberJobs$ jobs with known prompt length $\initialLen$ and unknown response lengths $\{\responseLen_i\}_{i\in[\NumberJobs]}$, \\
    KV-cache memory budget $\kvmem$, scaling factor $\responseScalar > 1$
}
\BlankLine

Let $\MaxPhaseIdx = \lfloor \log_{\responseScalar} \big( \kvmem - \initialLen \big) \rfloor, \sliceLowerBound = (\kvmem - \initialLen) / {\responseScalar^\MaxPhaseIdx}$ \tcp*{ensures $1 \leq \sliceLowerBound < \responseScalar$ and $\sliceLowerBound \cdot \responseScalar ^ \MaxPhaseIdx = \kvmem - \initialLen$}
Initialize geometric slice length $\sliceLenHat \gets \sliceLowerBound$ and job set $\remainJobs \gets [\NumberJobs]$ \;

\While {$\remainJobs \neq \emptyset$}
{
    Let $\sliceLen \gets \lfloor \sliceLenHat \rfloor$ \tcp*{integer slice length used in this phase}

    Invoke {\SPipelineSchedule}$(\remainJobs, \paralSliceLen, \sliceLen)$ \tcp*{$\paralSliceLen$ defined in \eqref{eq:max-parallelism}}

    Update set $\remainJobs \gets \{i \in \remainJobs \mid \text{job } i \text{ is not completed}\}$ \;
    Update geometric slice length $\sliceLenHat \gets \responseScalar \cdot \sliceLenHat$ \;
}
\end{algorithm2e}

\xhdr{Algorithm Overview.}
The {\GeoSALG} ({\GSA}) is built on a simple yet powerful principle: \emph{when job response lengths are unknown, the scheduler should iteratively guess their response length using a geometrically increasing time scale, and in each phase, execute a memory-efficient pipeline schedule based on the current guess}.

Specifically, given the set of unfinished jobs, {\GSA} processes all of them in successive phases. Each phase $\PhaseIdx = 0,1,2,\dots$ is defined by a \emph{slice length} $\sliceLenHatP{\PhaseIdx}$, which serves as the current hypothesis for job duration. The algorithm uses the integer slice length $\sliceLenP{\PhaseIdx} = \lfloor \sliceLenHatP{\PhaseIdx} \rfloor$ when invoking {\SPipelineSchedule} ({\SPS}) defined in Procedure~\ref{alg:staggered}.

The algorithm is parameterized by a \emph{scaling factor} $\responseScalar > 1$ (e.g., $\responseScalar = 2$ for doubling). The initial slice length $\sliceLenHatP{0}$ is set to the smallest value of the form $(\kvmem - \initialLen) / \responseScalar^{\MaxPhaseIdx}$ that is at least $1$. This choice ensures that the geometric sequence $\sliceLenHatP{\PhaseIdx} = \responseScalar^{\PhaseIdx} \cdot \sliceLenHatP{0}$ will eventually reach the maximum feasible response length $\kvmem - \initialLen$, i.e., the largest response length that can be accommodated within the memory budget.\footnote{This initialization is designed to optimize the theoretical analysis. Using $\sliceLenHatP{0} = 1$ yields a slightly worse constant factor, and require more discussions on the boundaries, e.g., when $\tau > \kvmem - \initialLen$. In practice (and in our numerical experiments (\Cref{apx:numerical})), this initialization $\sliceLenHatP{0}=\sliceLowerBound$ can be treated as a tunable parameter.}

In each phase $\PhaseIdx=0, 1, 2, \dots$, {\GSA} invokes {\SPS} on the current job set $\remainJobs_{\PhaseIdx}$, using:
\begin{itemize}
    \item slice length $\sliceLen = \sliceLenP{\PhaseIdx}$, and  
    \item the maximum memory-feasible degree of parallelism $\paralSliceLenP$, as defined in \Cref{lem:staggered:max-parallelism} (see Eq.~\eqref{eq:max-parallelism}).
\end{itemize}

During this phase, every job in $\remainJobs_{\PhaseIdx}$ is scheduled by {\SPS} and allowed to run for up to $\sliceLenP{\PhaseIdx}$ rounds. Jobs whose true response length satisfies $\responseLen_i \leq \sliceLenP{\PhaseIdx}$ complete successfully and are removed from the system. The remaining jobs (those with $\responseLen_i > \sliceLenP{\PhaseIdx}$) are killed at the end of the phase---their partial progress is discarded, and they will be restarted in the next phase.
After the phase, the algorithm updates the job set to $\remainJobs_{\PhaseIdx+1} \gets \remainJobs_{\PhaseIdx} \setminus \{\text{completed jobs}\}$ and scales the slice length geometrically:
\begin{align*}
\sliceLenHatP{\PhaseIdx+1} \gets \responseScalar \cdot \sliceLenHatP{\PhaseIdx},
\end{align*}
This geometric growth guarantees that for any job $i$, there exists a phase $\PhaseIdx_i$ such that $\sliceLenP{\PhaseIdx_i} \geq \responseLen_i$, ensuring eventual completion.\footnote{The geometric increase of the slice length is essential for handling arbitrary unknown response lengths while maintaining a bounded number of restarts per job.}
See \Cref{alg:gsa} for a formal description of {\GeoSALG}.

\xhdr{Comparison with ``let-long-jobs-finish'' heuristics in prior work.}
Many non-clairvoyant schedulers---both in prior theoretical work \citep[e.g.,][]{CYZ-25} and in practice \citep[e.g., First-Come-First-Served (FCFS) in ][]{KLZSZYGZS-23}---adopt a ``let-long-jobs-finish'' heuristic: they avoid preempting running jobs unless the memory budget is exceeded, in an effort to minimize restarts and the associated wasted computation. While intuitive, this strategy can lead to significant efficiency loss by blocking memory for extended periods and limiting parallelism when considering the flow time objective.

In contrast, our proposed {\GSA} employs a more aggressive \emph{early-killing} strategy: it systematically terminates all unfinished jobs at the end of each geometric phase, regardless of their progress.  
The following example illustrates that this distinction plays a key role in enabling {\GSA} to potentially outperform the aforementioned heuristics.

\begin{example}[Long-Job Trap for Conservative Heuristics]
\label{example:long-job-trap}
Fix a sufficiently large integer $\MaxPhaseIdx\in\naturals$. Consider an instance with $\NumberJobs \geq 2$ jobs, each having prompt length $\initialLen = 2^{\MaxPhaseIdx}$, and response length $\responseLen_i = 1 + (2^\MaxPhaseIdx - 1)\cdot \indicator{i = 1}$, under a KV-cache memory budget of $\kvmem=2^{\MaxPhaseIdx+1}$. (Namely, the instance contains a single ``long'' job and $\NumberJobs - 1$ ``short'' jobs. The KV-cache memory budget is set such that at most one job can be active at any time.)

Consider a non-clairvoyant ``let-long-jobs-finish'' heuristic $\ALG$ that never kills the long job (i.e., job 1). We further assume that it sequentially processes all jobs in an order chosen uniformly at random, as jobs are not distinguishable before execution. 
Suppose the long job is the $k$-th job to finish ($1 \leq k \leq \NumberJobs$). Then the total flow time consists of three parts: (i) the first $k-1$ short jobs contribute $\sum\nolimits_{\ell=1}^{k-1} \ell$; (ii) the long job finishes at time $(k-1) + 2^\MaxPhaseIdx$, contributing $k-1 + 2^\MaxPhaseIdx$; and (iii) the remaining $\NumberJobs-k$ short jobs wait for the long job, each contributing its completion time $\ell + 2^\MaxPhaseIdx - 1$. Summing these gives a total flow time of
\begin{align*}
\frac{\NumberJobs(\NumberJobs+1)}{2} + (\NumberJobs - k + 1)(2^\MaxPhaseIdx - 1).
\end{align*}
Since the execution order is uniform, $k$ is uniformly distributed in $\{1, \dots, \NumberJobs\}$, so $\mathbb{E}[\NumberJobs - k + 1] = (\NumberJobs+1)/2$. Thus, the expected flow time is
\begin{align*}
    \FTime{\ALG} = \frac{\NumberJobs(\NumberJobs+1)}{2} + \frac{\NumberJobs+1}{2}(2^\MaxPhaseIdx - 1) = \frac{\NumberJobs+1}{2} \cdot \left(\NumberJobs + 2^\MaxPhaseIdx - 1\right).
\end{align*}

In contrast, under {\GeoSALG} ({\GSA}) with scaling factor $\responseScalar = 2$, in the first phase, it invokes {\SPipelineSchedule} ({\SPS}) with slice length $\sliceLen = 1$ and degree of parallelism $\paralSliceLen = 1$ for all $\NumberJobs$ jobs. All $\NumberJobs-1$ short jobs are completed after this phase. In each phase $k \in [2:\MaxPhaseIdx]$, it invokes {\SPS} with slice length $\sliceLen = 2^{k}$ and degree of parallelism $\paralSliceLen = 1$ for the long job, and this long job is completed in the $\MaxPhaseIdx$-th phase. Thus, the total flow time of \textsc{GSA} satisfies
\begin{align*}
    \FTime{\GSA} \leq \underbrace{\frac{1}{2} \cdot (\NumberJobs+2)(\NumberJobs-1)}_{\text{flow time from short jobs}} + \underbrace{\NumberJobs + 2 + 4 + \dots + 2^\MaxPhaseIdx - 1 + 2^\MaxPhaseIdx}_{\text{flow time from long job}} = \frac{(\NumberJobs+2)(\NumberJobs-1)}{2} + \NumberJobs + 2^{\MaxPhaseIdx+1} - 2,
\end{align*}
where the inequality accounts for the worst executing order of $\NumberJobs$ jobs in the first phase when invoking {\SPS}.

Finally, the optimal clairvoyant schedule processes all $\NumberJobs - 1$ short jobs first and then processes the long job, with an optimal flow time of
\begin{align*}
    \FTime{\OPT} = \frac{\NumberJobs(\NumberJobs + 1)}{2} + 2^\MaxPhaseIdx - 1.
\end{align*}
When both $\NumberJobs$ and $\MaxPhaseIdx$ grow to infinity with $\MaxPhaseIdx = \Omega(\NumberJobs)$, the competitive ratios satisfy
\begin{align*}
    \frac{\FTime{\ALG}}{\FTime{\OPT}} \rightarrow \NumberJobs
    \quad \text{and} \quad
    \frac{\FTime{\GSA}}{\FTime{\OPT}} \rightarrow 2.
    \hfill\halmos
\end{align*}
\end{example}

\xhdr{Competitive Ratio Guarantee.} 
While the previous example illustrated that proactively killing and restarting the long job in {\GSA} is sometimes indispensable, even though it might increase restart waste. We now turn to the first main theoretical result of this paper.  
Specifically, {\GSA} achieves a \emph{constant competitive ratio for all problem instances}.  
\Cref{thm:GSA} first states its competitive ratio guarantee for an arbitrary scaling factor $\responseScalar$, and then presents the optimized competitive ratios for general instances and large-memory instances obtained by tuning $\responseScalar$ optimally.

\begin{theorem}[Competitive Ratio Guarantee of {\GSA}]
\label{thm:GSA}
For any $\responseScalar \in (1, \infty)$, the competitive ratio of the non-clairvoyant scheduler {\GeoSALG} (\GSA) with scaling factor $\responseScalar$ is at most 
\begin{align*}
    \approxratio(\GSA)  \le\ \left(2+\frac{2}{\responseScalar-1}\right)\left(\responseScalar ^ 2 \left(1 + \frac{2}{\minparal} \right)\ +\ \responseScalar\ +\ \frac{\responseScalar}{\responseScalar-1}\right)
\end{align*}
where ${\minparal} \geq 1$ is the minimum degree of parallelism for the given input instance.

Moreover, by optimally tuning scaling factor $\responseScalar = (7+\sqrt{13})/6\approx 1.7676$, the competitive ratio guarantee satisfies
\begin{align*}
 \approxratio(\GSA) = \frac{92+26\sqrt{13}}{3} \approx 61.9148
\end{align*}
and in the large memory regime (i.e., $\minparal\rightarrow \infty$),\footnote{This occurs when the prompt length $\initialLen$ and all response lengths $\responseLen_i$ are negligible compared to the KV-cache memory budget $\kvmem$, so that high parallelism is feasible throughout execution.} by setting scaling factor $\responseScalar = 2$, the competitive ratio guarantee converges to
\begin{align*}
\lim_{\minparal\rightarrow \infty}~
 \approxratio(\GSA) = 32~.
\end{align*}
\end{theorem}
The proof of \Cref{thm:GSA} proceeds in two main steps.  
Since the optimal clairvoyant schedule may be complex and difficult to characterize directly, we first consider the clairvoyant setting and design a clairvoyant greedy algorithm. We then establish that this greedy algorithm achieves a constant-factor approximation to the optimal clairvoyant schedule (\Cref{thm:GBA}) in \Cref{sec:clairvoyant greedy}.  
In the second step (\Cref{sec:algorithms}), we bound the competitive ratio between {\GSA} and this clairvoyant greedy algorithm, thereby linking the performance of our non-clairvoyant algorithm to the offline optimum.

\section{Constant-Approximation Algorithm in Clairvoyant Environment}
\label{sec:clairvoyant greedy}

In this section, we temporarily depart from our non-clairvoyant base model and consider the clairvoyant variant in which the scheduler knows all job response lengths in advance. We design a greedy algorithm---{\GeoBALG} (\GBA)---that batches jobs into batches based on their response lengths and then invokes our subroutine {\SPipelineSchedule} ({\SPS}). We show that this algorithm achieves a constant-factor approximation to the optimal clairvoyant schedule. Besides its own practical and theoretical interests, this result serves as a crucial intermediate step in establishing the competitive ratio guarantee of the non-clairvoyant {\GeoSALG} ({\GSA}).

\xhdr{Algorithm Overview.} 
{\GeoBALG} ({\GBA}) is the clairvoyant counterpart of {\GSA}, designed for the setting where all response lengths $\{\responseLen_i\}$ are known. Like {\GSA}, it proceeds in geometric phases with target slice length $\sliceLenHatP{\PhaseIdx}$ and integer slice length $\sliceLenP{\PhaseIdx}=\lfloor \sliceLenHatP{\PhaseIdx} \rfloor$ used by {\SPS}. The algorithm initializes $\sliceLenHatP{0}$ to $\sliceLowerBound = (\kvmem - \initialLen) / \responseScalar^\MaxPhaseIdx$ where $\MaxPhaseIdx = \lfloor \log_{\responseScalar} \big( \kvmem - \initialLen \big) \rfloor$, ensuring that $\sliceLenHatP{0} \in [1, \responseScalar)$, and the sequence eventually hits the memory-limited boundary $\kvmem - \initialLen$ exactly.
In each phase $\PhaseIdx$, {\GBA} schedules only the subset $\JobsClass_{\PhaseIdx} = \{ \jobIdx \in \Jobs : \sliceLenHatP{\PhaseIdx} / \responseScalar < \responseLen_\jobIdx \leq \sliceLenHatP{\PhaseIdx} \}$---those jobs guaranteed to finish within $\sliceLenP{\PhaseIdx}$ rounds. It invokes {\SPipelineSchedule} on this subset with the memory-feasible degree of parallelism $\paralSliceLenP$ (Eqn.~\eqref{eq:max-parallelism}), ensuring no restarts. After completion, these jobs are removed, and $\sliceLenHat$ is scaled by $\responseScalar$. Because {\GBA} avoids unnecessary preemptions and aligns structurally with {\GSA}, it serves as a near-optimal clairvoyant benchmark, enabling a constant-factor competitive analysis for {\GSA} against $\OPT$. See \Cref{alg:gba} for details.

\begin{algorithm2e}[ht]
\caption{{\GeoBALG} ({\GBA})}
\label{alg:gba}
\KwIn{
    $\NumberJobs$ jobs with known prompt length $\initialLen$ and known response lengths $\{\responseLen_i\}_{i\in[\NumberJobs]}$, \\
    KV-cache memory budget $\kvmem$, scaling factor $\responseScalar > 1$
}
\BlankLine

Let $\MaxPhaseIdx = \lfloor \log_{\responseScalar} \big( \kvmem - \initialLen \big) \rfloor, \sliceLowerBound = (\kvmem - \initialLen) / {\responseScalar^\MaxPhaseIdx}$ \tcp*{ensures $1 \leq \sliceLowerBound < \responseScalar$ and $\sliceLowerBound \cdot \responseScalar ^ \MaxPhaseIdx = \kvmem - \initialLen$}
Initialize geometric slice length $\sliceLenHat \gets \sliceLowerBound$ and job set $\remainJobs \gets [\NumberJobs]$ \;

\For{$\PhaseIdx = 0$ \KwTo $\MaxPhaseIdx$}
{
    Construct set $\Jobs_\PhaseIdx \gets \{i \in \remainJobs \mid \sliceLenHat / \responseScalar < \responseLen_i \leq \sliceLenHat\}$ \tcp*{jobs completed in this phase}
    Let $\sliceLen \gets \lfloor \sliceLenHat \rfloor$ \tcp*{integer slice length used in this phase}

    Invoke {\SPipelineSchedule}$(\Jobs_\PhaseIdx, \paralSliceLen, \sliceLen)$ \tcp*{$\paralSliceLen$ defined in \eqref{eq:max-parallelism}}
    
    Update set $\remainJobs \gets \remainJobs \setminus \Jobs_\PhaseIdx$ \;
    Update geometric slice length $\sliceLenHat \gets \responseScalar \cdot \sliceLenHat$ \;
}
\end{algorithm2e}

\subsection{Approximation Guarantee of {\GBA}} 

In this section, we present the theoretical approximation guarantee of {\GBA}.

\begin{theorem}[Approximation Guarantee of {\GBA}]
\label{thm:GBA}
For any $\responseScalar \in (1, \infty)$, the approximation ratio of the clairvoyant scheduler {\GeoBALG} (\GBA) with scaling factor $\responseScalar$ is at most 
\begin{align*}
    \approxratio(\GBA)  \leq
    \responseScalar^2\cdot 
    \left(1 + \frac{2}{\minparal}\right) + \responseScalar + \frac{\responseScalar}{\responseScalar-1}
\end{align*}
where ${\minparal} \triangleq \min_p k_p \geq 1$ is the minimum degree of parallelism for the given input instance.

Moreover, by optimally tuning scaling factor $\responseScalar = 4/3\approx 1.3333$, the approximation guarantee is at most
\begin{align*}
 \approxratio(\GBA) = \frac{32}{3} \approx 10.6667~,
\end{align*}
and in the large memory regime (i.e., $\minparal\rightarrow \infty$), by setting scaling factor $\responseScalar = 1.5$, the approximation guarantee converges to
\begin{align*}
\lim_{\minparal\rightarrow \infty}~
 \approxratio(\GBA) = 6.75~.
\end{align*}
\end{theorem}

\subsection{Analysis of {\GBA}: Key Steps and Technical Lemmas}
\label{sec:GBA:key steps}

Our analysis of {\GBA} (proof of Theorem~\ref{thm:GBA}) proceeds through a sequence of carefully orchestrated steps that connect the structure of {\GBA} to an information-theoretic lower bound on the optimal total flow time $\OPT$.  
At a high level, we extend the memory-time area argument (\Cref{def:area}) from the single-class setting (i.e., identical-job instances) used in \Cref{thm:staggered:identical job} (\Cref{subsec:staggered}) to a multi-class setting, where jobs are partitioned into classes based on their response lengths $\responseLen_i$.

Specifically, consider a fixed job set $\Jobs$ with $|\Jobs| = \NumberJobs$.  
For brevity, we let $\GBA$ denote the total flow time of running {\GBA} on $\Jobs$, and $\OPT$ denote the optimal total flow time for $\Jobs$.  
In {\GBA}, the jobs processed in phase $\PhaseIdx$ are referred to as \emph{class}~$\PhaseIdx$ jobs and denoted by $\JobsClass_{\PhaseIdx}$.  
Let:
\begin{itemize}
    \item $\NumJobsClass_\PhaseIdx \triangleq |\JobsClass_{\PhaseIdx}|$ be the number of jobs in class $\PhaseIdx \in \{0, 1, \dots, \MaxPhaseIdx\}$,
    \item $\sliceLenP{\PhaseIdx} \triangleq \lfloor \sliceLenHatP{\PhaseIdx} \rfloor = \lfloor \sliceLowerBound \responseScalar^\PhaseIdx \rfloor$ be the integer slice length used in phase $\PhaseIdx$,
    \item $\paral_\PhaseIdx \triangleq \paralSliceLenP$ be the degree of parallelism used for class $\PhaseIdx$.
\end{itemize}
For convenience, we also introduce auxiliary notations
\begin{align*}
\NumJobsGt{\PhaseIdx} \triangleq \sum\nolimits_{\PhaseIdxSecond > \PhaseIdx} \NumJobsClass_{\PhaseIdxSecond},
\quad
\NumJobsGe{ \PhaseIdx} \triangleq \sum\nolimits_{\PhaseIdxSecond \geq \PhaseIdx} \NumJobsClass_{\PhaseIdxSecond},
\end{align*}
as the number of jobs in classes strictly greater than $\PhaseIdx$, and in classes greater than or equal to $\PhaseIdx$, respectively.
The key steps of the analysis are as follows:

\xhdr{Step 1: Area-Based Lower Bound on $\OPT$ (\Cref{lem:GBA:phase-area-lower-bound}).}  
We extend the memory-time area argument from \Cref{lem:area-lower-bound} which lower bounds the optimal total flow time $\OPT$ to the multi-class setting. Specifically, we decompose the optimal total flow time $\OPT$ into two components:  
(i) the \emph{within-class cost}, capturing the flow time due to jobs within the same class, and  
(ii) the \emph{between-class cost}, accounting for the delay that smaller jobs impose on larger ones.  
This decomposition yields a phase-aware lower bound on $\OPT$.

Since jobs in phase $\PhaseIdx$ have response length strictly greater than $\sliceLowerBound \responseScalar^{\PhaseIdx-1}$, their memory-time area is at least  
\begin{align*}
\AreaLB_{\PhaseIdx} \;\triangleq\; \Area\left(\sliceLowerBound \responseScalar^{\PhaseIdx-1},\initialLen\right),
\end{align*}
where $\Area(\cdot,\cdot)$ is defined in \Cref{def:area}. This leads to the following lower bound on $\OPT$.

\begin{restatable}[Area-Based Lower Bound for $\OPT$]{lemma}{lemAreaGbaOptLB}
\label{lem:GBA:phase-area-lower-bound}
The optimal total flow time $\OPT$ satisfies
\begin{align*}
\OPT \;\geq\; \WithinClass + \BetweenClasses,
\end{align*}
where
\begin{align*}
\WithinClass \triangleq \sum\nolimits_\PhaseIdx \frac{\AreaLB_{\PhaseIdx}}{\kvmem} \cdot \frac{\NumJobsClass_\PhaseIdx(\NumJobsClass_\PhaseIdx + 1)}{2}, 
\;\;
\mbox{and}
\;\;
\BetweenClasses \triangleq \sum\nolimits_\PhaseIdx \frac{\AreaLB_{\PhaseIdx}}{\kvmem} \cdot \NumJobsGt{\PhaseIdx} \NumJobsClass_\PhaseIdx.
\end{align*}
\end{restatable}

\xhdr{Step 2: Per-Class Packing Efficiency of {\GBA} (\Cref{lem:GBA:per-class-refined-beta}).}  
We next extend the packing efficiency analysis from \Cref{lem:aux:ratio} to the multi-class setting. Specifically, we establish a lemma that characterizes the scheduling efficiency of {\GBA}, which invokes {\SPS} independently on each class of jobs. As in the identical-job case, this lemma relates the average time per job in class $\PhaseIdx$---captured by the ratio $\sliceLenP{\PhaseIdx} / \paral_\PhaseIdx$---to the memory-time area of that class, up to a multiplicative factor that depends on the scaling parameter $\responseScalar$ and the minimum degree of parallelism $\minparal \geq 1$.

\begin{restatable}[Area-Based Packing Efficiency of {\GBA}]{lemma}{lemGbaPackEffciency}
\label{lem:GBA:per-class-refined-beta}
For each class $\PhaseIdx$ induced by {\GBA} with scaling factor $\responseScalar \in (1, \infty)$,
\begin{align*}
\frac{\sliceLenP{\PhaseIdx}}{\paral_\PhaseIdx}\ \le\ 
\responseScalar^2 \cdot \left(1 + \frac{2}{\paral_\PhaseIdx}\right)
\cdot\frac{
\AreaLB_{\PhaseIdx}
}{\kvmem},
\end{align*}
\end{restatable}

\xhdr{Step 3: Upper Bound on {\GBA} via Packing Efficiency (\Cref{lem:GBA:phase-area-upper-bound}).}  
In this step, we leverage the per-class packing efficiency established in Step 2---namely, the bound on $\{\sliceLenP{\PhaseIdx} / \paral_\PhaseIdx\}_\PhaseIdx$---to derive an upper bound on the total flow time of {\GBA}.  
Specifically, this upper bound mirrors the decomposition used in the lower bound for $\OPT$: it is obtained by summing over classes, and for each class $\PhaseIdx$, it accounts for (i) the internal completion times of its jobs (via {\SPS}), and (ii) the waiting time incurred due to all earlier classes. This parallel structure enables a direct comparison between {\GBA} and $\OPT$ in the final analysis.
To formalize this, we introduce two key quantities that capture the two components of the flow time. For each class $\PhaseIdx$, define the \emph{ladder term}
\begin{align}
\label{eq:Qk-stag-def}
\Ladder_\PhaseIdx
\;\triangleq\;
\NumJobsClass_\PhaseIdx \cdot \sliceLenP{\PhaseIdx}
\;+\;
\sum\nolimits_{i=0}^{\NumJobsClass_\PhaseIdx - 1}
\left\lfloor \frac{i \cdot \sliceLenP{\PhaseIdx}}{\paral_\PhaseIdx} \right\rfloor,
\end{align}
which upper bounds the total completion time of the jobs in class $\PhaseIdx$ when scheduled in isolation via {\SPS}, and the \emph{class-prefix term}
\begin{align}
\label{eq:Sk-stag-def}
\ClassPrefix_\PhaseIdx
\;\triangleq\;
\sum\nolimits_{\PhaseIdxSecond < \PhaseIdx}
\left(
    \left\lfloor \frac{\NumJobsClass_\PhaseIdxSecond \cdot \sliceLenP{\PhaseIdxSecond}}{\paral_\PhaseIdxSecond} \right\rfloor
    + \sliceLenP{\PhaseIdxSecond}
\right),
\end{align}
which upper bounds the total waiting time experienced by jobs in class $\PhaseIdx$ due to execution of all earlier classes.
With these definitions, the total flow time of {\GBA} is bounded as follows.
\begin{restatable}[{\GBA} Upper Bound via Packing Efficiency]{lemma}{lemAreaGbaUB}
\label{lem:GBA:phase-area-upper-bound}
The total flow time of {\GBA} satisfies
\begin{align*}
\GBA \;\leq\; \UpperGBA \;\triangleq\; \sum\nolimits_{\PhaseIdx} \left( \NumJobsClass_\PhaseIdx \cdot \ClassPrefix_\PhaseIdx + \Ladder_\PhaseIdx \right).
\end{align*}
where $\Ladder_\PhaseIdx$ and $\ClassPrefix_\PhaseIdx$ are defined in Eqn.~\eqref{eq:Qk-stag-def} and Eqn.~\eqref{eq:Sk-stag-def}, respectively.
\end{restatable}

\xhdr{Step 4: Combining Bounds and Optimizing the Scaling Factor $\responseScalar$.}  
By combining the upper bound on {\GBA} with the lower bound on {\OPT}---using the memory-time area framework developed above---and leveraging the geometric progression of slice lengths, we obtain the approximation guarantee stated in \Cref{thm:GBA}.  
Optimizing the scaling factor $\responseScalar$ yields the constants given in \Cref{thm:GBA}: approximately $10.67$ in the general case and $6.75$ in the large-memory regime.

This structured approach not only establishes a constant-factor approximation guarantee for {\GBA}, but also provides the analytical foundation for the competitive analysis of the non-clairvoyant {\GSA} in Section~\ref{sec:algorithms}.

\subsection{Proofs in the Analysis of {\GBA}}
\label{sec:GBA:missing proofs}

In this section, we include all the proofs of the technical lemmas and then establish \Cref{thm:GBA}.

\lemAreaGbaOptLB*

\begin{proof}[Proof of \Cref{lem:GBA:phase-area-lower-bound}]
Following the argument in \Cref{lem:area-lower-bound}, let $C_1 \leq C_2 \leq \dots \leq C_{\NumberJobs}$ denote the completion times of the $\NumberJobs$ jobs in an optimal schedule. By the memory-time area accounting principle, the total area consumed by the first $i$ jobs is at least $\sum\nolimits_{j=1}^i \Area(\responseLen_j,\initialLen)$, and this cannot exceed $\kvmem \cdot C_i$. Hence,
\begin{align*}
C_i \;\geq\; \frac{1}{\kvmem} \sum\nolimits_{j=1}^i \Area(\responseLen_j,\initialLen).
\end{align*}
Summing over all $i$, we obtain
\begin{align*}
\OPT = \sum\nolimits_{i=1}^{\NumberJobs} C_i 
\;\geq\; \frac{1}{\kvmem} \sum\nolimits_{i=1}^{\NumberJobs} \sum\nolimits_{j=1}^i \Area(\responseLen_j,\initialLen)
\;=\; \frac{1}{\kvmem} \sum\nolimits_{j=1}^{\NumberJobs} (\NumberJobs - j + 1) \cdot \Area(\responseLen_j,\initialLen).
\end{align*}
Since each job $j$ belongs to some class $\PhaseIdx_j$, and by construction $\responseLen_j > \sliceLowerBound \responseScalar^{\PhaseIdx_j - 1}$, we have $\Area(\responseLen_j,\initialLen) \geq \AreaLB_{\PhaseIdx_j}$. Therefore,
\begin{align*}
\OPT \;\geq\; \frac{1}{\kvmem} \sum\nolimits_{j=1}^{\NumberJobs} (\NumberJobs - j + 1) \cdot \AreaLB_{\PhaseIdx_j}.
\end{align*}
To minimize the right-hand side, the jobs should be ordered in non-decreasing order of $\AreaLB_{\PhaseIdx_j}$. Because $\AreaLB_{\PhaseIdx}$ is strictly increasing in $\PhaseIdx$, this is achieved by processing all jobs in class $0$ first, then class $1$, and so on up to class $\MaxPhaseIdx$. Under this ordering, the jobs in class $\PhaseIdx$ occupy positions $j = \NumJobsGt{\PhaseIdx} + 1$ through $j = \NumJobsGe{ \PhaseIdx}$, and for each such job, the coefficient $(\NumberJobs - j + 1)$ equals $\NumJobsGt{\PhaseIdx} + i$, where $i \in \{1, \dots, \NumJobsClass_\PhaseIdx\}$ indexes the job within its class.
Thus,
\begin{align*}
\OPT 
&\geq \sum\nolimits_{\PhaseIdx} \sum\nolimits_{i=1}^{\NumJobsClass_\PhaseIdx} \left( \NumJobsGt{\PhaseIdx} + i \right) \cdot \frac{\AreaLB_{\PhaseIdx}}{\kvmem} 
= \sum\nolimits_{\PhaseIdx} \left( \NumJobsGt{\PhaseIdx} \NumJobsClass_\PhaseIdx \cdot \frac{\AreaLB_{\PhaseIdx}}{\kvmem} + \frac{\AreaLB_{\PhaseIdx}}{\kvmem} \cdot \frac{\NumJobsClass_\PhaseIdx(\NumJobsClass_\PhaseIdx + 1)}{2} \right)
\\
&= \BetweenClasses + \WithinClass,
\end{align*}
which completes the proof.
\end{proof}

\lemGbaPackEffciency*
\begin{proof}[Proof of \Cref{lem:GBA:per-class-refined-beta}]
By the maximality of $\paral_\PhaseIdx$ (see \Cref{lem:staggered:max-parallelism}), we have
\begin{align*}
\kvmem 
&< \initialLen(\paral_\PhaseIdx + 1) + \frac{\sliceLenP{\PhaseIdx}(\paral_\PhaseIdx + 1) + \sliceLenP{\PhaseIdx} + (\paral_\PhaseIdx + 1) - 1}{2} \nonumber 
\leq \initialLen \paral_\PhaseIdx + \initialLen + \frac{\sliceLenP{\PhaseIdx} \paral_\PhaseIdx}{2} + \sliceLenP{\PhaseIdx} + \frac{\paral_\PhaseIdx}{2}.
\end{align*}
Dividing both sides by $\AreaLB_{\PhaseIdx}$ and multiplying by $\sliceLenP{\PhaseIdx}/\paral_\PhaseIdx$ yields
\begin{align*}
\frac{\sliceLenP{\PhaseIdx}}{\paral_\PhaseIdx} \cdot \frac{\kvmem}{\AreaLB_{\PhaseIdx}}
<
\frac{
    \sliceLenP{\PhaseIdx} \initialLen \left(1 + \frac{1}{\paral_\PhaseIdx}\right)
    + \frac{\sliceLenP{\PhaseIdx}^2}{2} \left(1 + \frac{2}{\paral_\PhaseIdx}\right)
    + \frac{\sliceLenP{\PhaseIdx}}{2}
}{
    \AreaLB_{\PhaseIdx}
}.
\end{align*}
Let $y \triangleq \sliceLowerBound \responseScalar^{\PhaseIdx - 1}$, so that $\AreaLB_{\PhaseIdx} = \Area(y,\initialLen) = \initialLen y + \frac{y(y+1)}{2}$.  
Since $\sliceLenP{\PhaseIdx} \leq \responseScalar y$, substituting this upper bound and dividing numerator and denominator by $y$ gives
\begin{align*}
\frac{\sliceLenP{\PhaseIdx}}{\paral_\PhaseIdx} \cdot \frac{\kvmem}{\AreaLB_{\PhaseIdx}}
\;\leq\;
\TmpFunc_{\paral_\PhaseIdx}(y),
\end{align*}
where for any $\paral \geq 1$ and $y > 0$,
\begin{align*}
\TmpFunc_\paral(y) \;\triangleq\;
\frac{
    \responseScalar \initialLen \left(1 + \frac{1}{\paral}\right)
    + \responseScalar^2 y \left(\frac{1}{2} + \frac{1}{\paral}\right)
    + \frac{\responseScalar}{2}
}{
    \initialLen + \frac{y + 1}{2}
}.
\end{align*}
By computing its derivative, it can be verified that $\TmpFunc_\paral(y)$ is non-decreasing in $y$.
Recall that the initialization ensures $\sliceLowerBound \responseScalar^{\MaxPhaseIdx} = \kvmem - \initialLen$. Therefore, for any $\PhaseIdx \leq \MaxPhaseIdx$,
\begin{align*}
y = \sliceLowerBound \responseScalar^{\PhaseIdx - 1}
\;\leq\;
\frac{\kvmem - \initialLen}{\responseScalar}
\;\triangleq\; y_0.
\end{align*}
Since $\TmpFunc_{\paral_\PhaseIdx}$ is non-decreasing, $\TmpFunc_{\paral_\PhaseIdx}(y) \leq \TmpFunc_{\paral_\PhaseIdx}(y_0)$. Substituting and simplifying, we obtain
\begin{align*}
\TmpFunc_{\paral_\PhaseIdx}(y_0)
= \frac{
    \responseScalar \left( \frac{\initialLen}{2} + \kvmem \left( \frac{1}{2} + \frac{1}{\paral_\PhaseIdx} \right) + \frac{1}{2} \right)
}{
    \initialLen + \frac{\kvmem - \initialLen}{2 \responseScalar} + \frac{1}{2}
}.
\end{align*}
Dividing numerator and denominator by $\kvmem$ and letting $\promptRatio \triangleq \initialLen / \kvmem$, this becomes
\begin{align*}
\TmpFunc_{\paral_\PhaseIdx}(y_0)
= \frac{\responseScalar}{\paral_\PhaseIdx} \cdot
\frac{
    (\paral_\PhaseIdx + 2) + \paral_\PhaseIdx \promptRatio + \dfrac{\paral_\PhaseIdx}{\kvmem}
}{
    2 \promptRatio + \dfrac{1 - \promptRatio}{\responseScalar} + \dfrac{1}{\kvmem}
}
\triangleq \Psi(\responseScalar, \promptRatio, \paral_\PhaseIdx, \kvmem)
\end{align*}
Finally, by computing its derivative, it can be verified that auxiliary function $\Psi(\responseScalar, \promptRatio, \paral_\PhaseIdx, \kvmem)$ is
(i) strictly decreasing in $\promptRatio$, and
(ii) strictly increasing in $\kvmem$.
Therefore, for any
$\promptRatio \geq 0$
 and $\kvmem \geq 1$, we have
\begin{align*}
\Psi(\responseScalar, \promptRatio, \paral_\PhaseIdx, \kvmem)
\;\leq\;
\Psi(\responseScalar, 0, \paral_\PhaseIdx, \infty)
\;=\;
\responseScalar^2 \cdot \left(1 + \frac{2}{\paral_\PhaseIdx}\right).
\end{align*}
Substituting this bound into the previous inequality gives
\begin{align*}
\frac{\sliceLenP{\PhaseIdx}}{\paral_\PhaseIdx}
\;\leq\;
\responseScalar^2\cdot  \left(1 + \frac{2}{\paral_\PhaseIdx}\right) \cdot \frac{\AreaLB_{\PhaseIdx}}{\kvmem},
\end{align*}
which completes the proof.
\end{proof}

\lemAreaGbaUB*
\begin{proof}[Proof of \Cref{lem:GBA:phase-area-upper-bound}]
By construction, {\GBA} processes classes sequentially in increasing order of $\PhaseIdx$. All jobs in class $\PhaseIdx$ begin execution only after all jobs in classes $\PhaseIdxSecond < \PhaseIdx$ have completed. 
For any job in class $\PhaseIdx$, its completion time equals:
\begin{align*}
\underbrace{\text{(total time spent on all prior classes)}}_{\leq\, \ClassPrefix_\PhaseIdx}
\;+\;
\underbrace{\text{(its completion time within class $\PhaseIdx$ under \textsc{SPS})}}_{\leq\, \sliceLenP{\PhaseIdx} + \left\lfloor \frac{i \cdot \sliceLenP{\PhaseIdx}}{\paral_\PhaseIdx} \right\rfloor},
\end{align*}
where $i \in \{0, \dots, \NumJobsClass_\PhaseIdx - 1\}$ indexes the job within the class.
Summing over all $\NumJobsClass_\PhaseIdx$ jobs in class $\PhaseIdx$, the total contribution to flow time is at most
$\NumJobsClass_\PhaseIdx \cdot \ClassPrefix_\PhaseIdx + \Ladder_\PhaseIdx$.
Summing over all classes $\PhaseIdx$ yields the desired bound:
\begin{align*}
\GBA \leq \sum\nolimits_{\PhaseIdx} \left( \NumJobsClass_\PhaseIdx \cdot \ClassPrefix_\PhaseIdx + \Ladder_\PhaseIdx \right)
\end{align*}
which completes the proof as desired.
\end{proof}

We now ready to prove \Cref{thm:GBA}.

\begin{proof}[Proof of \Cref{thm:GBA}]
We bound the total flow time of \textsc{GBA} by analyzing its decomposition into within-class and between-class components. Recall from \Cref{lem:GBA:phase-area-upper-bound} that
\begin{align*}
\GBA \leq \UpperGBA = \sum\nolimits_{\PhaseIdx} \left( \NumJobsClass_\PhaseIdx \cdot \ClassPrefix_\PhaseIdx + \Ladder_\PhaseIdx \right).
\end{align*}
We bound the two terms separately.

\xhdr{Within-class ladder term.}
From the definition in Eqn.~\eqref{eq:Qk-stag-def},
\begin{align*}
\Ladder_\PhaseIdx
= \NumJobsClass_\PhaseIdx \sliceLenP{\PhaseIdx} + \sum\nolimits_{i=0}^{\NumJobsClass_\PhaseIdx - 1} \left\lfloor \frac{i \sliceLenP{\PhaseIdx}}{\paral_\PhaseIdx} \right\rfloor
\leq \NumJobsClass_\PhaseIdx \sliceLenP{\PhaseIdx} + \frac{\sliceLenP{\PhaseIdx}}{\paral_\PhaseIdx} \cdot \frac{\NumJobsClass_\PhaseIdx(\NumJobsClass_\PhaseIdx - 1)}{2}.
\end{align*}
By \Cref{lem:GBA:per-class-refined-beta}, we obtain
\begin{align*}
\sum\nolimits_{\PhaseIdx} \Ladder_\PhaseIdx
\leq \sum\nolimits_{\PhaseIdx} \NumJobsClass_\PhaseIdx \sliceLenP{\PhaseIdx}
+ \sum \nolimits_{\PhaseIdx} \responseScalar^2\cdot  \left(1+\frac{2}{\paral_\PhaseIdx}\right) \cdot \frac{\AreaLB_{\PhaseIdx}}{\kvmem} \cdot \frac{\NumJobsClass_\PhaseIdx(\NumJobsClass_\PhaseIdx - 1)}{2}.
\end{align*}
Noting that $\frac{n(n-1)}{2} < \frac{n(n+1)}{2}$ and $\minparal = \min_\PhaseIdx \paral_\PhaseIdx$ (as defined in the theorem), we obtain
\begin{align*}
\sum\nolimits_{\PhaseIdx} \Ladder_\PhaseIdx
\leq \sum\nolimits_{\PhaseIdx} \NumJobsClass_\PhaseIdx \sliceLenP{\PhaseIdx}
+ \responseScalar^2\cdot \left(1+\frac{2}{\minparal}\right) \cdot \WithinClass~.
\end{align*}
Since each job $i$ has response length $\responseLen_i$ and $\sliceLenP{\PhaseIdx} \leq \responseScalar \responseLen_i$ for all $i \in \JobsClass_\PhaseIdx$, we have
\begin{align*}
\sum\nolimits_{\PhaseIdx} \NumJobsClass_\PhaseIdx \sliceLenP{\PhaseIdx} \leq \responseScalar \sum\nolimits_{i=1}^{\NumberJobs} \responseLen_i \leq \responseScalar \cdot \OPT,
\end{align*}
because the optimal schedule must spend at least $\responseLen_i$ rounds on each job. Hence,
\begin{align*}
\sum\nolimits_{\PhaseIdx} \Ladder_\PhaseIdx \leq \responseScalar^2\cdot \left(1+\frac{2}{\minparal}\right) \cdot \WithinClass + \responseScalar \cdot \OPT.
\end{align*}

\xhdr{Between-class prefix term.}
By definition in Eqn.~\eqref{eq:Sk-stag-def} and summation exchange,
\begin{align*}
\sum\nolimits_{\PhaseIdx} \NumJobsClass_\PhaseIdx \ClassPrefix_\PhaseIdx
= \sum\nolimits_{\PhaseIdxSecond} \NumJobsGt{\PhaseIdxSecond} \left( \sliceLenP{\PhaseIdxSecond} + \left\lfloor \frac{\NumJobsClass_\PhaseIdxSecond \sliceLenP{\PhaseIdxSecond}}{\paral_\PhaseIdxSecond} \right\rfloor \right)
\leq (\mathrm{I}) + (\mathrm{II}),
\end{align*}
where
\begin{align*}
(\mathrm{I}) = \sum\nolimits_{\PhaseIdxSecond} \NumJobsGt{\PhaseIdxSecond} \sliceLenP{\PhaseIdxSecond},
\quad
(\mathrm{II}) = \sum\nolimits_{\PhaseIdxSecond} \NumJobsGt{\PhaseIdxSecond} \cdot \frac{\NumJobsClass_\PhaseIdxSecond \sliceLenP{\PhaseIdxSecond}}{\paral_\PhaseIdxSecond}.
\end{align*}
To bound $(\mathrm{I})$, recall that $\sliceLenHatP{\PhaseIdx} = \sliceLowerBound \responseScalar^\PhaseIdx$ and $\sliceLenP{\PhaseIdx} = \lfloor \sliceLenHatP{\PhaseIdx} \rfloor$. Then for any $\PhaseIdx$,
\begin{align}
\label{eq:geometric sum}
\sum\nolimits_{\PhaseIdxSecond < \PhaseIdx} \sliceLenP{\PhaseIdxSecond}
\leq \sum\nolimits_{\PhaseIdxSecond < \PhaseIdx} \sliceLenHatP{\PhaseIdxSecond}
= \frac{\sliceLenHatP{\PhaseIdx} - \sliceLowerBound}{\responseScalar - 1}
\leq \frac{\sliceLenP{\PhaseIdx} + 1 - 1}{\responseScalar - 1}
= \frac{\sliceLenP{\PhaseIdx}}{\responseScalar - 1},
\end{align}
since $\sliceLowerBound \geq 1$ and $\sliceLenP{\PhaseIdx} \geq \sliceLenHatP{\PhaseIdx} - 1$. Therefore,
\begin{align*}
(\mathrm{I}) = \sum\nolimits_{\PhaseIdx} \NumJobsClass_\PhaseIdx \sum\nolimits_{\PhaseIdxSecond < \PhaseIdx} \sliceLenP{\PhaseIdxSecond}
\leq \frac{1}{\responseScalar - 1} \sum\nolimits_{\PhaseIdx} \NumJobsClass_\PhaseIdx \sliceLenP{\PhaseIdx}
\leq \frac{\responseScalar}{\responseScalar - 1} \cdot \OPT.
\end{align*}
For $(\mathrm{II})$, applying \Cref{lem:GBA:per-class-refined-beta} gives
\begin{align*}
(\mathrm{II}) \leq \sum\nolimits_{\PhaseIdxSecond} \responseScalar^2\cdot \left(1+\frac{2}{\paral_\PhaseIdx}\right) \cdot  \NumJobsGt{\PhaseIdxSecond} \cdot \frac{\AreaLB_{\PhaseIdxSecond}}{\kvmem} \cdot \NumJobsClass_\PhaseIdxSecond
\leq \responseScalar^2\cdot \left(1+\frac{2}{\minparal}\right) \cdot \BetweenClasses.
\end{align*}
Thus,
\begin{align*}
\sum\nolimits_{\PhaseIdx} \NumJobsClass_\PhaseIdx \ClassPrefix_\PhaseIdx
\leq \responseScalar^2\cdot \left(1+\frac{2}{\minparal}\right) \cdot \BetweenClasses + \frac{\responseScalar}{\responseScalar - 1} \cdot \OPT.
\end{align*}

\xhdr{Combining the bounds.}
Combining all the pieces above and using $\OPT \geq \WithinClass + \BetweenClasses$ from \Cref{lem:GBA:phase-area-lower-bound}, we obtain
\begin{align*}
\GBA\leq \UpperGBA
&\leq \responseScalar^2\cdot \left(1+\frac{2}{\minparal}\right) \cdot (\WithinClass + \BetweenClasses) + \left( \responseScalar + \frac{\responseScalar}{\responseScalar - 1} \right) \cdot \OPT
\\
&\leq \left( \responseScalar^2\cdot \left(1+\frac{2}{\minparal}\right) + \responseScalar + \frac{\responseScalar}{\responseScalar - 1} \right) \cdot \OPT.
\end{align*}
as claimed.
To obtain the numerical constants, we optimize over $\responseScalar > 1$:
\begin{itemize}
    \item In the worst case ($\minparal = 1$), the bound becomes $3\responseScalar^2 + \responseScalar + \frac{\responseScalar}{\responseScalar - 1}$. Minimizing this expression yields $\responseScalar = 4/3$ and a value of ${32}/{3} \approx 10.6667$.
    \item In the large-memory regime ($\minparal \to \infty$), the bound reduces to $\responseScalar^2 + \responseScalar + \frac{\responseScalar}{\responseScalar - 1}$. Minimization gives $\responseScalar = 3/2$ and a value of $6.75$.
\end{itemize}
This completes the proof of \Cref{thm:GBA}.
\end{proof}

\section{Competitive Ratio Analysis of Geometric Slicing Algorithm}
\label{sec:algorithms}
In this section, we continue to analyze the performance of {\GeoSALG} ({\GSA}) introduced in \Cref{sec:algorithm description}.

The analysis of {\GSA} proceeds by leveraging the structural similarity between the non-clairvoyant {\GSA} and its clairvoyant counterpart {\GBA} (analyzed in \Cref{sec:clairvoyant greedy}). The core insight is that, despite operating without knowledge of job response lengths, {\GSA}’s geometric phase design ensures that its extra cost---due to processing all remaining jobs in every phase---can be bounded by a constant factor times the cost of {\GBA}. This yields a two-step competitive analysis: first bounding {\GSA} against {\GBA}, then combining with {\GBA}'s approximation guarantee against $\OPT$.

\subsection{Key Steps and Technical Lemmas}

\label{sec:GSA:key steps}

To analyze the flow time of {\GSA}, we adopt a decomposition analogous to that used for {\GBA} in \Cref{sec:GBA:key steps}. The key distinction lies in how jobs are scheduled across phases: whereas {\GBA} processes only class-$\PhaseIdx$ jobs in phase $\PhaseIdx$, {\GSA}---operating without knowledge of response lengths---processes \emph{all remaining jobs} in every phase. This non-clairvoyant behavior inflates the flow time, which we decompose into three components:
(i) the \emph{ladder term} $\Ladder_\PhaseIdx$, identical to that in {\GBA} (see Eqn.~\eqref{eq:Qk-stag-def}), capturing internal completion times within a class;
(ii) the \emph{spillover term} $\Spillover_\PhaseIdx$, reflecting the extra runtime incurred by jobs that span multiple phases; and
(iii) the \emph{phase-prefix term} $\PhasePrefix_\PhaseIdx$, representing waiting time due to all prior phases, now amplified by the larger active batches in {\GSA}.

We reuse the notation from the {\GBA} analysis (see \Cref{sec:GBA:key steps}):  
$\JobsClass_\PhaseIdx$ and $\NumJobsClass_\PhaseIdx$ denote the set and number of jobs in class $\PhaseIdx$;  
$\sliceLenP{\PhaseIdx} = \lfloor \sliceLenHatP{\PhaseIdx} \rfloor = \lfloor \sliceLowerBound \responseScalar^\PhaseIdx \rfloor$ is the integer slice length used by {\SPS} in phase $\PhaseIdx$;  
$\paral_\PhaseIdx$ is the corresponding degree of parallelism;
and $\NumJobsGt{\PhaseIdx}$, $\NumJobsGe{ \PhaseIdx}$ are the numbers of jobs in classes strictly greater than $\PhaseIdx$, and in classes greater than or equal to $\PhaseIdx$, respectively.
The key steps are as follows:

\xhdr{Step 1: Upper Bound on {\GSA} via Packing Efficiency (\Cref{lem:GSA:phase-area-upper-bound}).}  
For {\GSA}, every unfinished job participates in all phases up to its completion. To capture this, we define two additional quantities for each phase $\PhaseIdx$:
the \emph{phase spillover cost}
\begin{align}
\label{eq:Dk-stag-def}
\Spillover_\PhaseIdx \;\triangleq\; \Big\lceil \frac{\NumJobsGt{\PhaseIdx} \cdot \sliceLenP{\PhaseIdx}}{\paral_\PhaseIdx} \Big\rceil,
\end{align}
which upper bounds the extra rounds needed to accommodate jobs from later classes, and the \emph{phase-prefix cost}
\begin{align}
\label{eq:Tk-stag-def}
\PhasePrefix_\PhaseIdx \;\triangleq\; \sum\nolimits_{\PhaseIdxSecond < \PhaseIdx} \left( \Big\lfloor \frac{\NumJobsGe{\PhaseIdxSecond} \cdot \sliceLenP{\PhaseIdxSecond}}{\paral_\PhaseIdxSecond} \Big\rfloor + \sliceLenP{\PhaseIdxSecond} \right),
\end{align}
which accounts for the cumulative time spent on all prior phases.
With these definitions, the total flow time of {\GSA} is bounded as follows.

\begin{restatable}[{\GSA} Upper Bound via Packing Efficiency]{lemma}{lemAreaGsaUB}
\label{lem:GSA:phase-area-upper-bound}
The total flow time of {\GSA} satisfies
\begin{align*}
\GSA \;\leq\; \sum\nolimits_{\PhaseIdx} \NumJobsClass_\PhaseIdx \left( \PhasePrefix_\PhaseIdx + \Spillover_\PhaseIdx \right) + \sum\nolimits_{\PhaseIdx} \Ladder_\PhaseIdx,
\end{align*}
where $\Ladder_\PhaseIdx$, $\Spillover_\PhaseIdx$, and $\PhasePrefix_\PhaseIdx$ are as defined in Eqn.~\eqref{eq:Qk-stag-def}, Eqn.~\eqref{eq:Dk-stag-def}, and Eqn.~\eqref{eq:Tk-stag-def}, respectively.
\end{restatable}

\xhdr{Step 2: Upper Bound on Spillover Cost (\Cref{lem:spill-stag}).}  
We show that the total spillover cost incurred by {\GSA} is bounded by the prefix cost of its clairvoyant counterpart {\GBA} plus one ladder term. Specifically, we establish the following bound.
\begin{restatable}[Upper Bound on Spillover Cost of {\GSA}]{lemma}{lemGsaSpillover}
\label{lem:spill-stag}
The total spillover cost of {\GSA} satisfies
\begin{align*}
\sum\nolimits_{\PhaseIdx} \NumJobsClass_\PhaseIdx \cdot \Spillover_\PhaseIdx
\;\leq\;
\sum\nolimits_{\PhaseIdx} \NumJobsClass_\PhaseIdx \cdot \ClassPrefix_\PhaseIdx
\;+\;
\sum\nolimits_{\PhaseIdx} \Ladder_\PhaseIdx.
\end{align*}
\end{restatable}
We emphasize that this inequality holds only in aggregate over all phases; it may be violated for individual phases (i.e., $\NumJobsClass_\PhaseIdx \cdot \Spillover_\PhaseIdx \leq \NumJobsClass_\PhaseIdx \cdot \ClassPrefix_\PhaseIdx + \Ladder_\PhaseIdx$ need not hold for a fixed $\PhaseIdx$).

\xhdr{Step 3: Upper Bound on Inflated Prefix Cost (\Cref{lem:prefix-stag}).}  
We bound the inflated phase-prefix term $\PhasePrefix_\PhaseIdx$ using a novel counting inequality (see \Cref{lem:ceiling}) that relates the ceiling of aggregate job counts to the sum of individual contributions. Combined with the geometric increment of slice lengths, this yields a bound on the total phase-prefix cost of {\GSA} in terms of {\GBA}’s ladder and prefix terms.
\begin{restatable}[Upper Bound on Phase-Prefix Cost of {\GSA}]{lemma}{lemGsaPrefix}
\label{lem:prefix-stag}
For any scaling factor $\responseScalar \in (1, \infty)$, the total phase-prefix cost of {\GSA} satisfies
\begin{align*}
\sum\nolimits_{\PhaseIdx} \NumJobsClass_\PhaseIdx \cdot \PhasePrefix_\PhaseIdx
\;\leq\;
\left(1 + \frac{2}{\responseScalar - 1}\right)
\sum\nolimits_{\PhaseIdx} \NumJobsClass_\PhaseIdx \cdot \ClassPrefix_\PhaseIdx
\;+\;
\frac{2}{\responseScalar - 1}
\sum\nolimits_{\PhaseIdx} \Ladder_\PhaseIdx.
\end{align*}
\end{restatable}

\xhdr{Step 4: Combining Bounds and Optimizing the Scaling Factor $\responseScalar$.}  
By combining the upper bound on {\GSA} (from \Cref{lem:GSA:phase-area-upper-bound}) with the upper bound on {\GBA} (i.e., $\UpperGBA$ in Lemma~\ref{lem:GBA:phase-area-upper-bound}) and the approximation guarantee of $\UpperGBA$ relative to $\OPT$ established in Theorem~\ref{thm:GBA},\footnote{\label{footnote:GBA UB}The analysis in Theorem~\ref{thm:GBA} applies verbatim to $\UpperGBA$. Though $\UpperGBA$ is an upper bound on $\GBA$ and the proof compares this quantity directly to $\OPT$.} we obtain the competitive ratio stated in Theorem~\ref{thm:GSA}.
Optimizing the scaling factor $\responseScalar$ yields the constants given in Theorem~\ref{thm:GSA}: approximately $61.92$ in the general case and $32$ in the large-memory regime.

This approach not only establishes the first constant-competitive algorithm for non-clairvoyant KV-cache scheduling but also demonstrates the power of geometric phase alignment in bridging clairvoyant and non-clairvoyant settings.

\subsection{Proofs in the Analysis of {\GSA}}
\label{sec:GSA:missing proofs}

In this section, we include all the proofs of the technical lemmas and then establish \Cref{thm:GBA}.

\lemAreaGsaUB*

\begin{proof}
By construction, {\GSA} processes all remaining jobs in every phase. Consider any job $i$ that belongs to class $\PhaseIdx$ (i.e., it completes in phase $\PhaseIdx$). Its contribution to the total flow time consists of three parts:

\xhdr{Time spent in prior phases:} In each phase $\PhaseIdxSecond < \PhaseIdx$, job $i$ is active and contributes to the execution of the batch. The total time spent across all prior phases is exactly the duration of those phases, which is upper bounded by $\PhasePrefix_\PhaseIdx$ as defined in Eqn.~\eqref{eq:Tk-stag-def}.

\xhdr{Spillover time in phase $\PhaseIdx$:} At the start of phase $\PhaseIdx$, the active batch includes all unfinished jobs, namely the $\NumJobsGe{\PhaseIdx} = \NumJobsClass_\PhaseIdx + \NumJobsGt{\PhaseIdx}$ jobs from class $\PhaseIdx$ and later classes. The extra load from the $\NumJobsGt{\PhaseIdx}$ jobs from later classes increases the runtime of the phase. The additional rounds needed due to this ``spillover'' are upper bounded by $\Spillover_\PhaseIdx = \lceil \NumJobsGt{\PhaseIdx} \cdot \sliceLenP{\PhaseIdx} / \paral_\PhaseIdx \rceil$, as in Eqn.~\eqref{eq:Dk-stag-def}.

\xhdr{Internal completion time within class $\PhaseIdx$:} Once the spillover effect is accounted for, the relative ordering of the $\NumJobsClass_\PhaseIdx$ jobs in class $\PhaseIdx$ under {\SPS} yields a total internal completion time bounded by the ladder term $\Ladder_\PhaseIdx$ defined in Eqn.~\eqref{eq:Qk-stag-def}.

Summing over all jobs in class $\PhaseIdx$, the total contribution to flow time is at most
\begin{align*}
\NumJobsClass_\PhaseIdx \cdot (\PhasePrefix_\PhaseIdx + \Spillover_\PhaseIdx) + \Ladder_\PhaseIdx.
\end{align*}
Finally, summing over all phases $\PhaseIdx$ gives the desired upper bound:
\begin{align*}
\GSA \leq \sum\nolimits_{\PhaseIdx} \NumJobsClass_\PhaseIdx \left( \PhasePrefix_\PhaseIdx + \Spillover_\PhaseIdx \right) + \sum\nolimits_{\PhaseIdx} \Ladder_\PhaseIdx,
\end{align*}
which completes the proof.
\end{proof}

\lemGsaSpillover*
\begin{proof}[Proof of \Cref{lem:spill-stag}]
Let $\localratio_\PhaseIdx \triangleq \sliceLenP{\PhaseIdx} / \paral_\PhaseIdx$. Since $\lceil x \rceil \leq x + 1$, we have
\begin{align*}
\Spillover_\PhaseIdx
= \left\lceil \NumJobsGt{\PhaseIdx} \cdot \localratio_\PhaseIdx \right\rceil
\leq \NumJobsGt{\PhaseIdx} \cdot \localratio_\PhaseIdx + 1.
\end{align*}
Therefore,
\begin{align*}
\sum\nolimits_{\PhaseIdx} \NumJobsClass_\PhaseIdx \cdot \Spillover_\PhaseIdx
\leq \sum\nolimits_{\PhaseIdx} \NumJobsGt{\PhaseIdx}\NumJobsClass_\PhaseIdx \cdot \localratio_\PhaseIdx
 + \sum\nolimits_{\PhaseIdx} \NumJobsClass_\PhaseIdx.
\end{align*}
Since $\sliceLenP{\PhaseIdx} \geq 1$ and $x < \lfloor x \rfloor + 1$, we obtain
\begin{align*}
\NumJobsClass_\PhaseIdx \cdot \localratio_\PhaseIdx
\leq \left\lfloor \frac{\NumJobsClass_\PhaseIdx \cdot \sliceLenP{\PhaseIdx}}{\paral_\PhaseIdx} \right\rfloor + \sliceLenP{\PhaseIdx}.
\end{align*}
Multiplying by $\NumJobsGt{\PhaseIdx}$ and summing gives
\begin{align*}
\sum\nolimits_{\PhaseIdx} \NumJobsGt{\PhaseIdx}\NumJobsClass_\PhaseIdx \cdot \localratio_\PhaseIdx
&\leq \sum\nolimits_{\PhaseIdx} \NumJobsGt{\PhaseIdx} \left( \left\lfloor \frac{\NumJobsClass_\PhaseIdx \cdot \sliceLenP{\PhaseIdx}}{\paral_\PhaseIdx} \right\rfloor + \sliceLenP{\PhaseIdx} \right) 
= \sum\nolimits_{\PhaseIdx} \NumJobsClass_\PhaseIdx \cdot \ClassPrefix_\PhaseIdx,
\end{align*}
where the equality follows by swapping the order of summation in the definition of $\ClassPrefix_\PhaseIdx$ (Eqn.~\eqref{eq:Sk-stag-def}).
Finally, by definition of $\Ladder_\PhaseIdx$ in Eqn.~\eqref{eq:Qk-stag-def}, we have
\begin{align*}
\sum\nolimits_{\PhaseIdx} \NumJobsClass_\PhaseIdx
&\leq \sum\nolimits_{\PhaseIdx} \NumJobsClass_\PhaseIdx \cdot \sliceLenP{\PhaseIdx}
\leq \sum\nolimits_{\PhaseIdx} \Ladder_\PhaseIdx,
\end{align*}
Combining the bounds completes the proof.
\end {proof}

\lemGsaPrefix*

The proof of \Cref{lem:prefix-stag} relies on the following technical claim. 

\begin{lemma}
\label{lem:ceiling}
For any integers $\NumberJobs \geq 1$ and $\paral \geq 1$, the following inequality holds: 
\begin{align*}
\NumberJobs \cdot \left\lceil \frac{\NumberJobs}{\paral} \right\rceil
\;\leq\;
2 \sum\nolimits_{u=1}^{\NumberJobs} \left\lceil \frac{u}{\paral} \right\rceil.
\end{align*}
\end{lemma}

\begin{proof}
Let $q \triangleq \lfloor \NumberJobs / \paral \rfloor$ and write $\NumberJobs = q \paral + r$, where $0 \leq r < \paral$.  
Then $\lceil \NumberJobs / \paral \rceil = q$ if $r = 0$, and $q + 1$ otherwise.
We first compute the sum on the right-hand side:
\begin{align*}
\sum\nolimits_{u=1}^{\NumberJobs} \left\lceil \frac{u}{\paral} \right\rceil
= \sum\nolimits_{\ell=1}^{q} \ell \cdot \paral + (q + 1) \cdot r
= \frac{q(q + 1)}{2} \cdot \paral + (q + 1) \cdot r.
\end{align*}
We now verify the inequality in both cases.

\xhdr{Case 1: $r = 0$.}  
    In this case, $\NumberJobs = q \paral$ and $\lceil \NumberJobs / \paral \rceil = q$. The left-hand side becomes $q \paral \cdot q = q^2 \paral$, while the right-hand side is
    \begin{align*}
    2 \cdot \frac{q(q + 1)}{2} \cdot \paral = q(q + 1) \paral.
    \end{align*}
    Since $q^2 \leq q(q + 1)$, the inequality holds.

\xhdr{Case 2: $r > 0$.}  
    In this case, $\lceil \NumberJobs / \paral \rceil = q + 1$, and the left-hand side is
    \begin{align*}
    \NumberJobs \cdot (q + 1) = (q \paral + r)(q + 1).
    \end{align*}
    The right-hand side is
    \begin{align*}
    2 \left( \frac{q(q + 1)}{2} \cdot \paral + (q + 1) \cdot r \right)
    = (q + 1)(q \paral + 2r).
    \end{align*}
    Thus, the desired inequality reduces to
    \begin{align*}
    q \paral + r \;\leq\; q \paral + 2r,
    \end{align*}
    which holds since $r \geq 1$. This completes the proof as desired.
\end{proof}

\begin{proof}[Proof of \Cref{lem:prefix-stag}]
We begin by expanding the left-hand side and decomposing $\NumJobsGe{\PhaseIdxSecond} = \NumJobsClass_\PhaseIdxSecond + \NumJobsGt{\PhaseIdxSecond}$:
\begin{align*}
\sum\nolimits_{\PhaseIdx} \NumJobsClass_\PhaseIdx \cdot \PhasePrefix_\PhaseIdx
= \sum\nolimits_{\PhaseIdxSecond} \NumJobsGt{\PhaseIdxSecond} \left( \sliceLenP{\PhaseIdxSecond} + \left\lfloor \frac{\NumJobsGe{\PhaseIdxSecond} \cdot \sliceLenP{\PhaseIdxSecond}}{\paral_\PhaseIdxSecond} \right\rfloor \right).
\end{align*}
Splitting the floor term and using $\lfloor a + b \rfloor \leq \lfloor a \rfloor + \lceil b \rceil$, we obtain
\begin{align*}
\sum\nolimits_{\PhaseIdx} \NumJobsClass_\PhaseIdx \cdot \PhasePrefix_\PhaseIdx
\leq 
\underbrace{
    \sum\nolimits_{\PhaseIdxSecond} \NumJobsGt{\PhaseIdxSecond} \left( \sliceLenP{\PhaseIdxSecond} + \left\lfloor \frac{\NumJobsClass_\PhaseIdxSecond \cdot \sliceLenP{\PhaseIdxSecond}}{\paral_\PhaseIdxSecond} \right\rfloor \right)
}_{(\mathrm{I}) \equiv \sum\nolimits_{\PhaseIdx} \NumJobsClass_\PhaseIdx \cdot \ClassPrefix_\PhaseIdx}
+
\underbrace{
    \sum\nolimits_{\PhaseIdxSecond} \NumJobsGt{\PhaseIdxSecond} \cdot \left\lceil \frac{\NumJobsGt{\PhaseIdxSecond} \cdot \sliceLenP{\PhaseIdxSecond}}{\paral_\PhaseIdxSecond} \right\rceil
}_{(\mathrm{II})}.
\end{align*}
We now bound term $(\mathrm{II})$. Since $\sliceLenP{\PhaseIdxSecond}$ is an integer,
\begin{align*}
(\mathrm{II})
\leq
\sum\nolimits_{\PhaseIdxSecond} \NumJobsGt{\PhaseIdxSecond} \cdot \sliceLenP{\PhaseIdxSecond} \cdot \left\lceil \frac{\NumJobsGt{\PhaseIdxSecond}}{\paral_\PhaseIdxSecond} \right\rceil.
\end{align*}
Applying the counting inequality from \ref{lem:ceiling}, we have
\begin{align*}
\NumJobsGt{\PhaseIdxSecond} \cdot \left\lceil \frac{\NumJobsGt{\PhaseIdxSecond}}{\paral_\PhaseIdxSecond} \right\rceil
\leq
2 \sum\nolimits_{u=1}^{\NumJobsGt{\PhaseIdxSecond}} \left\lceil \frac{u}{\paral_\PhaseIdxSecond} \right\rceil,
\end{align*}
and thus
\begin{align}
\label{eq:II-stag-detail}
(\mathrm{II}) \leq 2 \sum\nolimits_{\PhaseIdxSecond} \sliceLenP{\PhaseIdxSecond} \sum\nolimits_{u=1}^{\NumJobsGt{\PhaseIdxSecond}} \left\lceil \frac{u}{\paral_\PhaseIdxSecond} \right\rceil.
\end{align}
We now enumerate the $\NumJobsGt{\PhaseIdxSecond}$ jobs in classes strictly above $\PhaseIdxSecond$ by their class $\PhaseIdxFourth > \PhaseIdxSecond$ and within-class rank $\RankIdx \in \{1, \dots, \NumJobsClass_\PhaseIdxFourth\}$. By the definition of $\paral$ and the monotonicity of $\MemPeak$ (\Cref{lem:staggered:peak-memory}), for $\PhaseIdxFourth > \PhaseIdxSecond$, we have $\paral_\PhaseIdxSecond \geq \paral_\PhaseIdxFourth$, which implies $1/\paral_\PhaseIdxSecond \leq 1/\paral_\PhaseIdxFourth$. Using this monotonicity,
\begin{align*}
\sum\nolimits_{u=1}^{\NumJobsGt{\PhaseIdxSecond}} \left\lceil \frac{u}{\paral_\PhaseIdxSecond} \right\rceil
&\leq
\sum\nolimits_{\substack{\PhaseIdxFourth > \PhaseIdxSecond :\NumJobsClass_\PhaseIdxFourth \geq 1}}
\sum\nolimits_{\RankIdx=1}^{\NumJobsClass_\PhaseIdxFourth}
\left(
    \left\lceil \frac{\RankIdx}{\paral_\PhaseIdxFourth} \right\rceil
    +
    \sum\nolimits_{\PhaseIdxThird = \PhaseIdxSecond + 1}^{\PhaseIdxFourth - 1}
    \left\lceil \frac{\NumJobsClass_\PhaseIdxThird}{\paral_\PhaseIdxSecond} \right\rceil
\right) \\
&\leq
\sum\nolimits_{\substack{\PhaseIdxFourth > \PhaseIdxSecond : \NumJobsClass_\PhaseIdxFourth \geq 1}}
\sum\nolimits_{\RankIdx=1}^{\NumJobsClass_\PhaseIdxFourth}
\left(
    \left\lceil \frac{\RankIdx}{\paral_\PhaseIdxFourth} \right\rceil
    +
    \sum\nolimits_{\PhaseIdxThird = \PhaseIdxSecond + 1}^{\PhaseIdxFourth - 1}
    \left\lceil \frac{\NumJobsClass_\PhaseIdxThird}{\paral_\PhaseIdxThird} \right\rceil
\right)
\end{align*}
Substituting into \eqref{eq:II-stag-detail} and splitting the sum yields
\begin{align*}
(\mathrm{II}) \leq 2\cdot \left( \mathrm{(II.A)} + \mathrm{(II.B)} \right),
\end{align*}
where
\begin{align*}
\mathrm{(II.A)} \triangleq \sum\nolimits_{\PhaseIdxSecond} \sliceLenP{\PhaseIdxSecond} \sum\nolimits_{\PhaseIdxFourth > \PhaseIdxSecond} \sum\nolimits_{\RankIdx=1}^{\NumJobsClass_\PhaseIdxFourth} \left\lceil \frac{\RankIdx}{\paral_\PhaseIdxFourth} \right\rceil,
\;\;\mbox{and}\;\;
\mathrm{(II.B)} \triangleq \sum\nolimits_{\PhaseIdxSecond} \sliceLenP{\PhaseIdxSecond} \sum\nolimits_{\PhaseIdxFourth > \PhaseIdxSecond} \NumJobsClass_\PhaseIdxFourth \sum\nolimits_{\PhaseIdxThird = \PhaseIdxSecond + 1}^{\PhaseIdxFourth - 1} \left\lceil \frac{\NumJobsClass_\PhaseIdxThird}{\paral_\PhaseIdxThird} \right\rceil.
\end{align*}
We now bound each term using the geometric increment of slice lengths. By Eqn.~\eqref{eq:geometric sum},
\begin{align*}
\sum\nolimits_{\PhaseIdxSecond < \PhaseIdxFourth} \sliceLenP{\PhaseIdxSecond} \leq \frac{\sliceLenP{\PhaseIdxFourth}}{\responseScalar - 1}.
\end{align*}
For (II.A), interchanging the order of summation gives
\begin{align*}
\mathrm{(II.A)}
&= \sum\nolimits_{\PhaseIdxFourth} \left( \sum\nolimits_{\PhaseIdxSecond < \PhaseIdxFourth} \sliceLenP{\PhaseIdxSecond} \right) \sum\nolimits_{\RankIdx=1}^{\NumJobsClass_\PhaseIdxFourth} \left\lceil \frac{\RankIdx}{\paral_\PhaseIdxFourth} \right\rceil 
\leq \frac{1}{\responseScalar - 1} \sum\nolimits_{\PhaseIdxFourth} \sliceLenP{\PhaseIdxFourth} \sum\nolimits_{\RankIdx=1}^{\NumJobsClass_\PhaseIdxFourth} \left\lceil \frac{\RankIdx}{\paral_\PhaseIdxFourth} \right\rceil.
\end{align*}
Since $\sliceLenP{\PhaseIdxFourth}$ is integral, we have
\begin{align*}
\sliceLenP{\PhaseIdxFourth} \sum\nolimits_{\RankIdx=1}^{\NumJobsClass_\PhaseIdxFourth} \left\lceil \frac{\RankIdx}{\paral_\PhaseIdxFourth} \right\rceil
= \NumJobsClass_\PhaseIdxFourth \sliceLenP{\PhaseIdxFourth} + \sum\nolimits_{\RankIdx=0}^{\NumJobsClass_\PhaseIdxFourth - 1} \sliceLenP{\PhaseIdxFourth} \left\lfloor \frac{\RankIdx}{\paral_\PhaseIdxFourth} \right\rfloor
\leq \Ladder_\PhaseIdxFourth,
\end{align*}
where the inequality follows from $\sliceLenP{\PhaseIdxFourth} \lfloor \RankIdx / \paral_\PhaseIdxFourth \rfloor \leq \lfloor \RankIdx \sliceLenP{\PhaseIdxFourth} / \paral_\PhaseIdxFourth \rfloor$. Hence,
\begin{align*}
\mathrm{(II.A)} \leq \frac{1}{\responseScalar - 1} \sum\nolimits_{\PhaseIdxFourth} \Ladder_\PhaseIdxFourth.
\end{align*}
For (II.B), swapping summations yields
\begin{align*}
\mathrm{(II.B)}
&= \sum\nolimits_{\PhaseIdxFourth} \left( \sum\nolimits_{\PhaseIdxSecond < \PhaseIdxFourth} \sliceLenP{\PhaseIdxSecond} \right) \NumJobsGt{\PhaseIdxFourth} \left\lceil \frac{\NumJobsClass_\PhaseIdxFourth}{\paral_\PhaseIdxFourth} \right\rceil \leq \frac{1}{\responseScalar - 1} \sum\nolimits_{\PhaseIdxFourth} \NumJobsGt{\PhaseIdxFourth} \sliceLenP{\PhaseIdxFourth} \left\lceil \frac{\NumJobsClass_\PhaseIdxFourth}{\paral_\PhaseIdxFourth} \right\rceil \\
&\leq \frac{1}{\responseScalar - 1} \sum\nolimits_{\PhaseIdxFourth} \NumJobsGt{\PhaseIdxFourth} \left( \left\lfloor \frac{\NumJobsClass_\PhaseIdxFourth \sliceLenP{\PhaseIdxFourth}}{\paral_\PhaseIdxFourth} \right\rfloor + \sliceLenP{\PhaseIdxFourth} \right) = \frac{1}{\responseScalar - 1} \sum\nolimits_{\PhaseIdx} \NumJobsClass_\PhaseIdx \sum\nolimits_{\PhaseIdxSecond < \PhaseIdx} \left( \left\lfloor \frac{\NumJobsClass_\PhaseIdxSecond \sliceLenP{\PhaseIdxSecond}}{\paral_\PhaseIdxSecond} \right\rfloor + \sliceLenP{\PhaseIdxSecond} \right) = \frac{1}{\responseScalar - 1} \sum\nolimits_{\PhaseIdx} \NumJobsClass_\PhaseIdx \cdot \ClassPrefix_\PhaseIdx.
\end{align*}
Combining the bounds on (I), (II.A), and (II.B), we obtain
\begin{align*}
\sum\nolimits_{\PhaseIdx} \NumJobsClass_\PhaseIdx \cdot \PhasePrefix_\PhaseIdx
\leq
\left(1 + \frac{2}{\responseScalar - 1}\right) \sum\nolimits_{\PhaseIdx} \NumJobsClass_\PhaseIdx \cdot \ClassPrefix_\PhaseIdx
+ \frac{2}{\responseScalar - 1} \sum\nolimits_{\PhaseIdx} \Ladder_\PhaseIdx,
\end{align*}
which completes the proof of \Cref{lem:prefix-stag} as desired.
\end{proof}

We are now ready to combine all technical pieces and establish \Cref{thm:GSA} for {\GSA}.
\begin{proof}[Proof of \Cref{thm:GSA}]
Invoking \Cref{lem:GSA:phase-area-upper-bound,lem:spill-stag,lem:prefix-stag}, we obtain
\begin{align*}
    \GSA &\leq 
    \sum\nolimits_{\PhaseIdx} \NumJobsClass_\PhaseIdx \left( \PhasePrefix_\PhaseIdx + \Spillover_\PhaseIdx \right) + \sum\nolimits_{\PhaseIdx} \Ladder_\PhaseIdx
    \\
    &\leq \sum\nolimits_{\PhaseIdx} \NumJobsClass_\PhaseIdx \cdot \ClassPrefix_\PhaseIdx
    + 2 \sum\nolimits_{\PhaseIdx} \Ladder_\PhaseIdx
    + \left(1 + \frac{2}{\responseScalar - 1}\right)
\sum\nolimits_{\PhaseIdx} \NumJobsClass_\PhaseIdx \cdot \ClassPrefix_\PhaseIdx
\;+\;
\frac{2}{\responseScalar - 1}
\sum\nolimits_{\PhaseIdx} \Ladder_\PhaseIdx
    \leq 
    \left(2+\frac{2}{\responseScalar-1}\right)\,\UpperGBA
\end{align*}
Invoking \Cref{thm:GBA}, we obtain\textsuperscript{\ref{footnote:GBA UB}}  
\begin{align*}
    \GSA \ \le\ \Bigl(2+\frac{2}{\responseScalar-1}\Bigr)\left(\responseScalar^2\cdot \left(1+\frac{2}{\minparal}\right) +\ \responseScalar\ +\ \frac{\responseScalar}{\responseScalar-1}\right)\,\OPT.
\end{align*}
as claimed. To obtain the numerical constants, we optimize over $\responseScalar > 1$:
\begin{itemize}
    \item In the worst case ($\minparal = 1$), the bound becomes $(2+\frac{2}{\responseScalar-1})(3\responseScalar^2 + \responseScalar + \frac{\responseScalar}{\responseScalar - 1})$. Minimizing this expression yields $\responseScalar = (7+\sqrt{13})/6\approx 1.7676$ and a value of $\frac{92+26\sqrt{13}}{3} \approx 61.9147$.
    \item In the large-memory regime ($\minparal \to \infty$), the bound reduces to $(2+\frac{2}{\responseScalar-1})(\responseScalar^2 + \responseScalar + \frac{\responseScalar}{\responseScalar - 1})$. Minimization gives $\responseScalar = 2$ and a value of $32$.
\end{itemize}
This completes the proof of \Cref{thm:GSA}.
\end{proof}

\section{Conclusion and Future Directions}
\label{sec:conclusion}
This work demonstrates that constant-competitive scheduling is achievable for LLM inference under dynamic KV-cache constraints, even without knowledge of response lengths. Our results suggest two broader insights. First, the staggered pipeline mechanism shows that \emph{desynchronizing} job execution can be fundamentally more efficient than batching greedily (and thus simultaneously). Second, geometric slicing illustrates how disciplined preemption can tame non-clairvoyance: by accepting controlled overhead from periodic restarts, we avoid the unbounded delays that arise from prevalent heuristics.

We highlight several promising directions for future work.
First, it would be valuable to explore more general models that incorporate online arrivals or heterogeneous prompt lengths. Our analysis focuses on the offline batch setting with identical prompt lengths. As discussed in \Cref{sec:model justificaiton}, relaxing either assumption leads to strong hardness results. While recent works~\citep[e.g.,][]{ALSW25,WYZ-25,ZLMP-25,LDP-25,JJMMPZ-25} have initiated the study of those extension models, it remains largely open to characterize which intermediate settings admit strong performance guarantees.
Second, tightening the theoretical guarantees of our algorithms and establishing fundamental lower bounds is an important next step. Although we present the first constant competitive ratio for the non-clairvoyant setting, no non-trivial lower bounds are currently known. Proving such bounds---or further improving the upper bounds---would deepen our understanding of the inherent limits of non-clairvoyant scheduling under dynamic memory constraints.
Third, extending our framework to multi-machine settings is a natural and practically relevant direction. Our model assumes a single GPU with a fixed KV-cache budget, whereas real-world LLM inference systems often distribute requests across multiple machines or GPUs with heterogeneous memory capacities. Developing algorithms that effectively route or migrate jobs in such environments poses significant new challenges and opportunities.

\newcommand{\newblock}{}
\setlength{\bibsep}{0.0pt}
\bibliographystyle{plainnat}
\OneAndAHalfSpacedXI
{\footnotesize

\newpage
\bibliography{refs}}

\newpage
\renewcommand{\theHchapter}{A\arabic{chapter}}
\renewcommand{\theHsection}{A\arabic{section}}

\ECSwitch
\ECDisclaimer

\section{Numerical Experiments}
\label{apx:numerical}
We evaluate the performance of our proposed algorithms using a discrete-time event simulator that faithfully captures the memory-constrained scheduling dynamics defined in \Cref{sec:prelim}. Our evaluation is designed to bridge the gap between theoretical analysis and practical deployment, answering two key questions: 
\begin{enumerate}
    \item Do the structural benefits of geometric scheduling observed in our worst-case analysis hold in controlled settings?
    \item Can these insights translate into tangible performance gains under realistic, heavy-tailed LLM workloads?
\end{enumerate}
To this end, we organize our experiments into two categories:

\xhdr{Theoretical-Concept Experiments (\Cref{sec:theoretical experiments}).} 
We first deploy small-scale synthetic workloads to isolate key structural properties of the algorithms. These micro-benchmarks serve to validate our theoretical insights---specifically, the ability of staggered starts to smooth memory usage and the robustness of kill-and-restart policies under adversarial conditions---in a controlled environment free from the noise of random distributions.

\xhdr{Trace-Driven Experiments (\Cref{sec:trace-driven experiments}).} 
We then evaluate practical performance using realistic job traces derived from the LMSYS-Chat-1M dataset \citep{ZCSLTLZWXGZSZHZ24}. These macro-benchmarks demonstrate the effectiveness of our algorithms and their heuristic variants in handling the right-skewed generation lengths characteristic of modern LLM inference.

\subsection{Baselines}
\label{sec:baselines}

We first introduce the baseline policies considered in our numerical experiments.

\xhdr{Baselines from Prior Work.}
We use two baselines from prior work for comparison.
\begin{itemize}
    \item {\MCSF} (Memory-Constrained Shortest First) from \citep{JJMMPZ-25}: a clairvoyant greedy algorithm that employs shortest job first scheduling under memory constraints, admitting as many jobs as possible in order of increasing response length.
    \item {\AMin} from \citep{CYZ-25}, a non-clairvoyant algorithm that maintains an estimate $\responseLenEst_i$ of response lengths $\responseLen_i$ based on observed progress, prioritizes jobs with the smallest estimated sizes, and kills the job with minimum estimated remaining length when memory is not sufficient. It also admits as many jobs as possible in order of increasing estimated response length, breaking tie uniformly at random.
\end{itemize}

\xhdr{Practical Baseline.}
We implement a practical baseline {\vLLM} similar to the default vLLM scheduler \citep{KLZSZYGZS-23}, which describes a family of First-Come-First-Served (FCFS) policy: jobs are admitted in order of arrival\footnote{
In our model, all jobs arrive at time $0$. Thus, jobs are ordered by their indices by default.}, and when memory is insufficient, the job with the latest arrival time is evicted to free memory. This ensures fairness among jobs and achieves good average latency in practice.

In our simulator we remove production safeguards (e.g., watermark mechanisms for memory overflow), making it slightly more aggressive; this can only improve its performance in our setting.

\subsection{Theoretical-Concept Experiments}
\label{sec:theoretical experiments}

In this section, we present our two theoretical-concept experiments.

\xhdr{Uniform-Size Distribution for Clairvoyant Algorithms.}
We simulate a workload of $\NumberJobs = 200$ jobs, each having identical response length $\responseLen_i = 16$ and prompt length $\initialLen=0$, under a memory budget of $\kvmem=256$.

As illustrative in \Cref{fig:exp-uniform}, {\MCSF} exhibits synchronized execution waves that repeatedly reach the memory capacity, leading to inefficient resource utilization. In contrast, {\GBA} staggers job start times, maintaining a smoother and more uniform memory profile over time.

\begin{figure}
  \centering
  \begin{minipage}[c]{0.48\textwidth}
    \centering
    \subfloat[{\MCSF}: 21632]{
      \includegraphics[width=\linewidth]{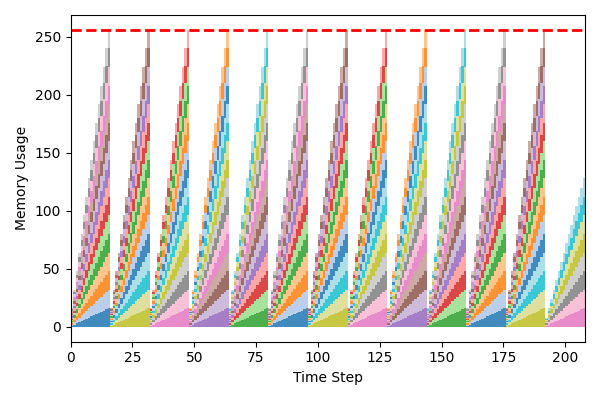}
    }
\end{minipage}
  \begin{minipage}[c]{0.48\textwidth}
    \centering
    \subfloat[{\GBA} (with $\responseScalar = 2$): 14083]{
      \includegraphics[width=\linewidth]{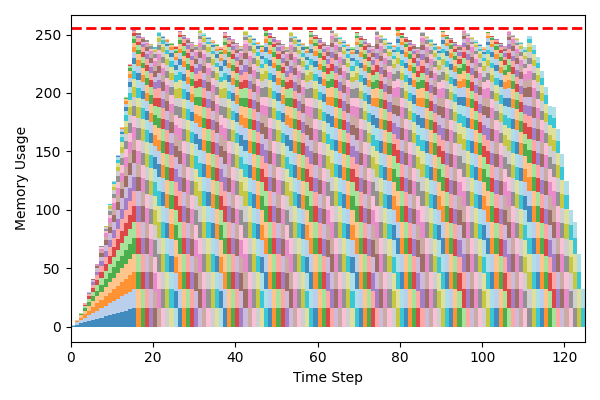}
    }
\end{minipage}
  \caption{Uniform-size distribution results 
  for clairvoyant algorithms.
  The total flow times are indicated in the titles. The $x$-axis represents time (rounds), and the $y$-axis represents the usage of memory in each round, where
  each job is represented by a colored block.
  }
  \label{fig:exp-uniform}
\end{figure}

\xhdr{Two-Point Distribution for Non-Clairvoyant Algorithms.}
We consider an instance where all jobs share a prompt length of $\initialLen = 96$. The response lengths follow a bimodal distribution: six ``long'' jobs have $\responseLen_i = 160$, while the remaining $194$ ``short'' jobs have $\responseLen_i = 1$. The memory budget is fixed at $\kvmem=256$.

The results in \Cref{fig:exp-twopoint} reveal severe head-of-line blocking in both {\vLLM} and {\AMin}. {\GSA}, however, mitigates this by strategically killing long jobs to admit short ones, thereby preventing prolonged blocking phases. Notably, the randomized strategy of {\AMin} yields no significant improvement in this adversarial setting, as it still fall into the ``let-long-jobs-finish'' trap.

\begin{figure}
  \centering
  \begin{minipage}[c]{0.48\textwidth}
    \centering
\subfloat[{\vLLM}: 78319]{
      \includegraphics[width=\linewidth]{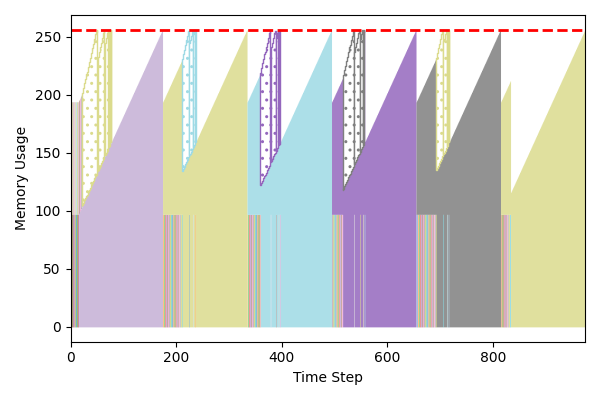}
    }\\
    \subfloat[{\AMin}: 53508]{
      \includegraphics[width=\linewidth]{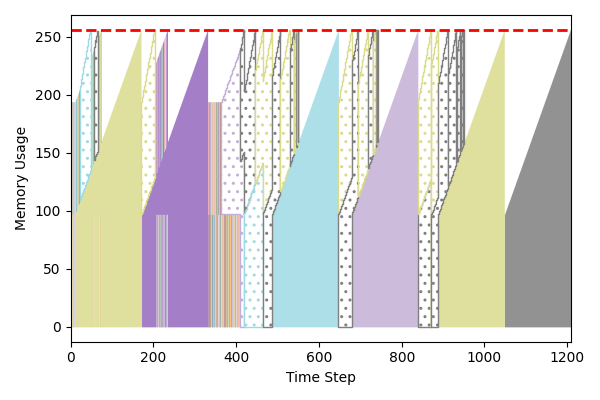}
    }
  \end{minipage}
  \hfill
  \begin{minipage}[c]{0.48\textwidth}
    \centering
\subfloat[{\GSA}($\responseScalar = 2$): 18268]{
      \includegraphics[width=\linewidth]{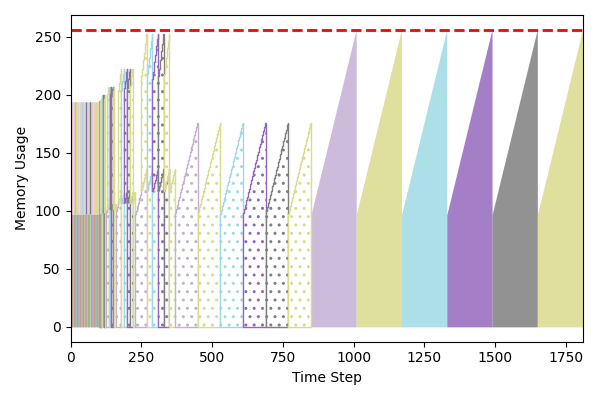}
    }
\end{minipage}
  \caption{Two-point distribution results for non-clairvoyant algorithms. The total flow times are indicated in the titles. The $x$-axis represents time (rounds), and the $y$-axis represents the usage of memory in each round, where
  each finished job is represented by a colored block, and each killed job is represented by a dotted block.
  }
  \label{fig:exp-twopoint}
\end{figure}

\subsection{Trace-Driven Experiments}
\label{sec:trace-driven experiments}

To capture real-world job size distributions, we construct workloads using token counts from the LMSYS-Chat-1M dataset \citep{ZCSLTLZWXGZSZHZ24}. Specifically, the job response length $\responseLen_i$ is defined as the number of tokens in the first assistant response of each conversation. The prompt length is set to $\initialLen = 79$, reflecting the average token count of user messages preceding the response. \Cref{fig:job size distribution} illustrates the resulting response length distribution, which exhibits the right-skewed characteristic typical of LLM inference workloads. For each trial with $\NumberJobs$ jobs used in the simulation, the first $\NumberJobs$ valid rows in the dataset are selected.

\begin{figure}[h]
  \centering
  \includegraphics[width=0.48\linewidth]{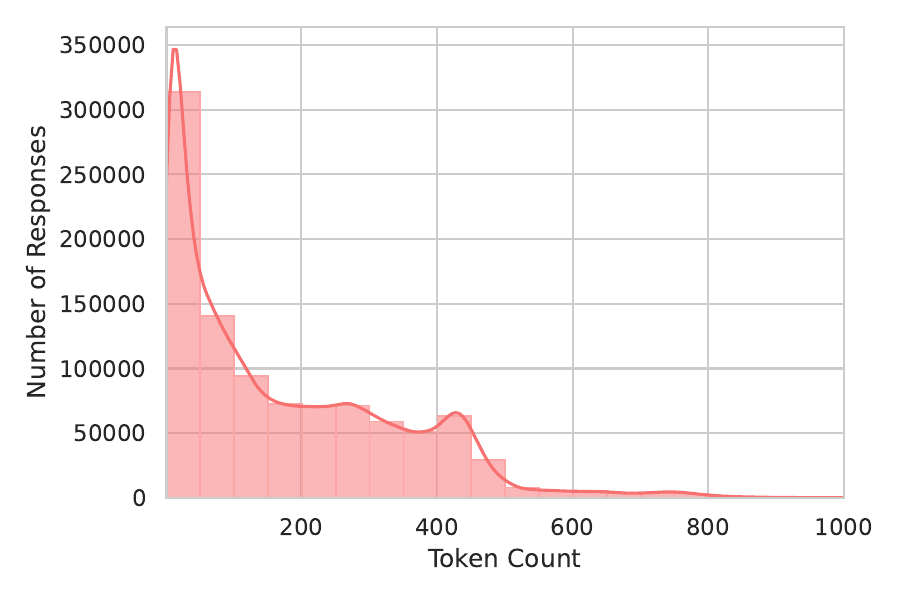}
  \caption{Response length distribution from the LMSYS-Chat-1M dataset, counted using \texttt{cl100k\_base}, the tokenizer of GPT-4 \citep{OpenAI-GPT4-23}.
  }
  \label{fig:job size distribution}
\end{figure}

Due to the right-skewed job size distributions observed in real-world LLM inference workloads, directly implementing {\GBA} and {\GSA} may lead to suboptimal memory utilization and increased idle times:
while fixed, geometrically rounded slice lengths ensure clean theoretical guarantees, in practice most jobs exit before the slice completes and the pipeline remains idle until the slice ends.
To address this, we implement two practical variants of our algorithms for trace-driven experiments while retaining their worst-case approximation / competitiveness guarantees.

\xhdr{Dynamic Refill (\GBAD).}
Our theoretical algorithm {\GBA} utilizes geometric slicing to simplify analysis. While this bridges the analysis from clairvoyant to non-clairvoyant settings, the synchronous rounds and fixed slice lengths are practically inefficient and can be improved by taking 
simple heuristics to eliminate the deliberate synchronization.
{\GBAD}, our heuristic variant of {\GBA}, employs a \textit{dynamic refill} strategy:
\begin{itemize}
  \item 
Whenever it is possible to continue the pipeline with a slice length equal to the smallest unfinished job response length ($\sliceLen = \responseLen_i$) without violating memory feasibility, it admits the job immediately.
\end{itemize}
{\GBAD} never delays any job relative to the baseline {\GBA} schedule; it only fills idle time, and therefore preserves the worst-case guarantee. In our experiments, {\GBAD} achieves the best performance on nearly all instances.

\xhdr{Speculative Runs (\GSAH).}  In the non-clairvoyant setting, we must pessimistically choose the slice length to ensure memory feasibility. However, this may lead to under-utilization of memory due to not fully utilize the available budget.
To mitigate this, {\GSAH}, our heuristic variant of {\GSA}, further enhances dynamic refill with \textit{speculative runs}.
\begin{itemize}
  \item 
When the memory is available after admitting all jobs in the current pipeline, it opportunistically starts new speculative runs for unfinished jobs using the freed memory.
\item When the memory is insufficient, it kills the speculative jobs to ensure the original geometric slicing schedule is maintained.
\item The speculative runs are started and killed in a First-Come-First-Serve (FCFS) manner like {\vLLM}.
\end{itemize}
This heuristic further improves memory utilization by utilizing freed memory for speculative runs, while maintaining the theoretical guarantees of {\GSA}.

\xhdr{Trace-Driven Power-of-Two Rounding Results.}
We begin by evaluating performance using a \emph{power-of-two quantized} workload where response lengths are rounded up to the nearest power of two. This setting serves two purposes. First, from a systems perspective, it mirrors the behavior of memory allocators commonly used in OS and GPU runtimes, where resource reservations effectively snap to power-of-two boundaries to minimize external fragmentation. Second, from an algorithmic perspective, it aligns the workload granularity with the geometric phase structure of {\GSA}. This allows us to isolate the core efficiency of our scheduling logic by reducing the noise from fine-grained size variations, effectively testing the algorithm under a ``structured uncertainty'' regime while retaining practical relevance.
In \Cref{fig:exp-power-two}, we present the mean flow times---a metric that is linearly related to total flow time, but more interpretable with growing $\NumberJobs$---under memory budgets $\kvmem \in \{4096, 8192\}$, and vary $\NumberJobs$ from $100$ to $1000$. {\GBAD} consistently outperforms all baselines, and {\GSAH} also demonstrates competitive performance, steadily surpassing {\vLLM} and {\AMin}. Notably, as $\NumberJobs$ increases, the performance gap between our proposed algorithms and the baselines widens, highlighting their scalability and effectiveness in handling larger workloads.

\begin{figure}[h]
  \centering
  \begin{minipage}[c]{0.48\textwidth}
    \centering
    \subfloat[$\kvmem=4096$]{
      \includegraphics[width=\linewidth]{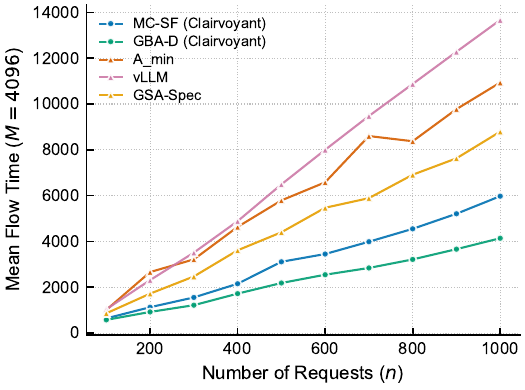}
    }
  \end{minipage}
  \hfill
  \begin{minipage}[c]{0.48\textwidth}
    \centering
    \subfloat[$\kvmem=8192$]{
      \includegraphics[width=\linewidth]{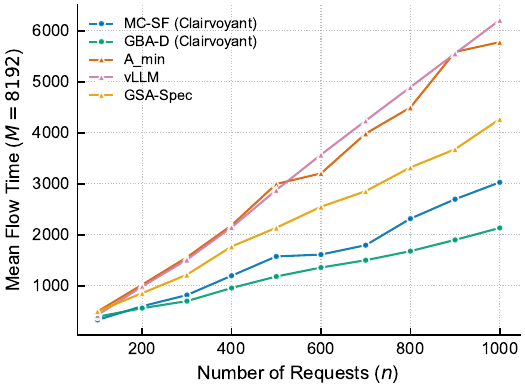}
    }
  \end{minipage}
  \caption{Results of trace-driven power-of-two rounding experiments. Each panel shows mean flow time---total flow time divided by $\NumberJobs$---as we sweep $\NumberJobs$ under a fixed memory budget. {\GSAH} uses $\responseScalar = 2$ and $\sliceLowerBound = 64$.}
  \label{fig:exp-power-two}
\end{figure}

\xhdr{Trace-driven Experiments Results.}
In \Cref{fig:exp-trace}, we present mean flow times under memory budgets $\kvmem \in \{4096, 8192\}$, varying $\NumberJobs$ from $100$ to $1000$, without any rounding of job sizes.

\begin{figure}[h]
  \centering
  \begin{minipage}[c]{0.48\textwidth}
    \centering
    \subfloat[$\kvmem=4096$]{
      \includegraphics[width=\linewidth]{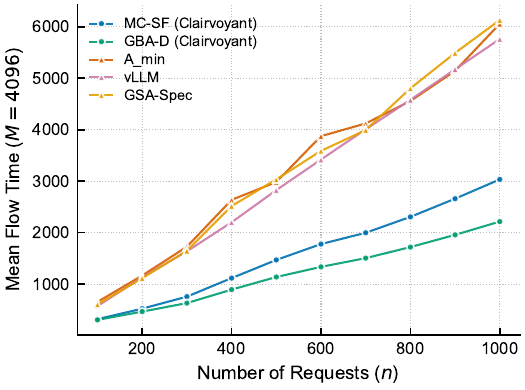}
    }
  \end{minipage}
  \hfill
  \begin{minipage}[c]{0.48\textwidth}
    \centering
    \subfloat[$\kvmem=8192$]{
      \includegraphics[width=\linewidth]{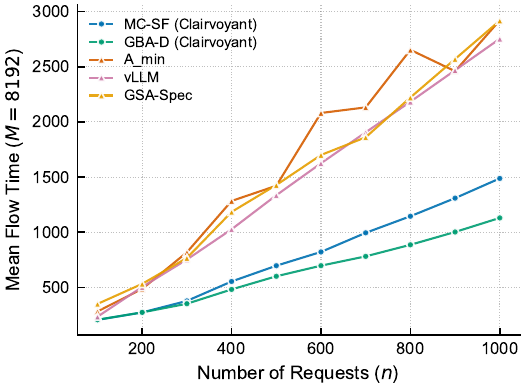}
    }
  \end{minipage}
  \caption{Results of trace-driven experiments. Each panel shows mean flow time---total flow time divided by $\NumberJobs$---as we sweep $\NumberJobs$ under a fixed memory budget. {\GSAH} uses $\responseScalar = 2$ and $\sliceLowerBound = 256$.}
  \label{fig:exp-trace}
\end{figure}

In these trace-driven settings, {\GBAD} continues to demonstrate superior performance, and {\GSAH} remains competitive with the baselines. Comparing these results with the power-of-two experiments suggests that the increased challenge for {\GSAH} stems from memory under-utilization, a consequence of using conservative slice lengths for diverse job sizes. However, by allowing speculative runs and tuning parameters like $\sliceLowerBound$, {\GSAH} can better adapt to the workload characteristics, improving memory utilization and achieving a favorable balance between theoretical guarantees and practical performance.

\section{Missing Proofs in Section~\ref{subsec:staggered}}
\label{apx:missing proof algorithm description}
\subsection{Proof of the Approximation Ratio of 2 in Theorem~\ref{thm:staggered:identical job}}
\label{apx:sps-paral-1}

In this section, we present the formal proof of the approximation ratio of 2 stated in \Cref{thm:staggered:identical job}.

\begin{proof}[Proof of approximation ratio of 2 in \Cref{thm:staggered:identical job}]
Let $L \triangleq \lceil \sliceLen/2 \rceil$. We first establish a necessary condition for feasibility: in any memory-valid schedule where $S_i$ and $C_i$ denote the start and completion times of job $i$ (indexed $0, 1, \dots$), the start times $S_0 \le S_1 \le \dots$ must satisfy
\begin{align}
    S_{i+\paral} - S_i \geq L \qquad \text{for all } i. \label{eq:spacing-condition}
\end{align}
To prove this, assume for the sake of contradiction that $S_{i+\paral} - S_i \leq L - 1$ for some $i$. Consider the system state at time $T \triangleq S_i + \sliceLen - 1$. At this time, job $i$ is in its final round of execution ($\curProgress_{i,T} = \sliceLen - 1$) and consumes $\initialLen + \sliceLen$ memory.
Since $S_{i+r} \le S_{i+\paral} \le S_i + L - 1$ for all $r \in \{1, \dots, \paral\}$, the progress of job $i+r$ at time $T$ satisfies:
$$
    \curProgress_{i+r,T} = T - S_{i+r} \ge (S_i + \sliceLen - 1) - (S_i + L - 1) = \sliceLen - L.
$$
Consequently, each of these $\paral$ jobs consumes at least $\initialLen + (\sliceLen - L) + 1 = \initialLen + \lfloor \sliceLen/2 \rfloor + 1$ memory. Summing over the window of $\paral+1$ jobs ($i, i+1, \dots, i+\paral$), the total memory usage at time $T$ is:
\begin{align*}
    \textsf{Mem}(T) 
    &\geq (\initialLen + \sliceLen) + \paral\left(\initialLen + \left\lfloor \frac{\sliceLen}{2} \right\rfloor + 1\right) 
    = (\paral+1)\initialLen + \sliceLen + \paral + \paral\left\lfloor \frac{\sliceLen}{2} \right\rfloor.
\end{align*}
Using the inequality $\lfloor x/2 \rfloor \geq (x-1)/2$, we have $\lfloor \sliceLen/2 \rfloor \geq (\sliceLen - 1)/2$. Substituting this bound:
\begin{align*}
    \textsf{Mem}(T) 
    &\geq (\paral+1)\initialLen + \sliceLen + \paral + \frac{\paral(\sliceLen - 1)}{2} 
    = (\paral+1)\initialLen + \frac{\sliceLen(\paral+1) + \sliceLen + \paral}{2}.
\end{align*}
Comparing this with the definition of peak memory in \Cref{lem:staggered:peak-memory}:
\begin{align*}
    \MemPeak(\paral+1, \sliceLen, \initialLen) = (\paral+1)\initialLen + \frac{\sliceLen(\paral+1) + \sliceLen + (\paral+1) - \gcd(\paral+1, \sliceLen)}{2}.
\end{align*}
Since $\gcd(\paral+1, \sliceLen) \geq 1$, we have $(\paral+1) - \gcd(\paral+1, \sliceLen) \leq \paral$. Therefore:
$$
    \textsf{Mem}(T) \ge \MemPeak(\paral+1, \sliceLen, \initialLen).
$$
By definition of $\paral$, $\MemPeak(\paral+1, \sliceLen, \initialLen) > \kvmem$. This implies $\textsf{Mem}(T) > \kvmem$, yielding the contradiction to memory feasibility. Thus \eqref{eq:spacing-condition} holds.

We now continue to the approximation ratio. By \eqref{eq:spacing-condition}, any group of $\paral$ jobs in the optimal schedule takes at least $L$ rounds to admit, i.e., $S_i^\OPT \geq \lfloor i/\paral \rfloor L$, and
$$
    C_i^\OPT \geq \left\lfloor \frac{i}{\paral} \right\rfloor L + \sliceLen.
$$
In our algorithm ({\SPS}), the completion time of job $i$ is:
$$
    C_i^{\text{SPS}} = \left\lfloor \frac{i\sliceLen}{\paral} \right\rfloor + \sliceLen.
$$
Using the inequality $\lfloor i\sliceLen/\paral \rfloor < (\lfloor i/\paral \rfloor + 1)\sliceLen$ and $2L \geq \sliceLen$:
\begin{align*}
    C_i^{\text{SPS}} 
    &< \left(\left\lfloor \frac{i}{\paral} \right\rfloor + 1\right)\sliceLen + \sliceLen = \left(\left\lfloor \frac{i}{\paral} \right\rfloor + 2\right)\sliceLen 
    \le 2\left(\left\lfloor \frac{i}{\paral} \right\rfloor L + \sliceLen\right) \le 2C_i^\OPT.
\end{align*}
Thus, $C_i^{\text{SPS}} \le 2C_i^\OPT$ for all $i$, completing the proof.
\end{proof}

\begin{remark} Our analysis of approximation ratio of 2 in \Cref{thm:staggered:identical job} above relies on a rigid structural property, the spacing lower bound in \eqref{eq:spacing-condition}, of the optimal schedule. This is unique to the identical-job setting.
For general instances, {\OPT} can exhibit more complex behavior, ruling out such simple structural comparisons.
\end{remark}

\subsection{Lower Bounds for {\SimPipelineSchedule}}
\label{apx:simps-lower-bounds}

We first restate the {\SimPipelineSchedule} ({\SimPS}) used in \Cref{subsec:staggered}, with more details.
For identical jobs with prompt length $\initialLen$ and response length $\responseLen$, {\SimPS} admits
jobs in batches of size
$$
B \;\triangleq\; \left\lfloor \frac{\kvmem}{\initialLen + \responseLen} \right\rfloor,
$$
runs each batch for $\responseLen$ rounds, and starts a new batch only after the current batch finishes.
All jobs in the same batch complete simultaneously. This procedure describes the behavior of a family of greedy algorithms when $\kvmem$ is a multiple of $\initialLen + \responseLen$.

In the following, we use $\varepsilon$ to denote $o(1)$, a small positive constant that can be arbitrarily close to $0$ in the asymptotic regime, to avoid confusion with response length $\responseLen$. 

\begin{proposition}[Asymptotic lower bound of $2$ for {\SimPS}]
\label{prop:simps-lb-2}
There exists a family of identical-job instances such that the approximation ratio of {\SimPS} is at least
$2 - \varepsilon$ even in the asymptotic regime where $\kvmem \to \infty$, $\NumberJobs \to \infty$, and $\kvmem = o(\NumberJobs)$.
\end{proposition}
\begin{proof}
Fix $\initialLen = 0$ and let the response length $\responseLen$ grow with $\kvmem$, so that both
$\responseLen \to \infty$ and $B \triangleq \kvmem/\responseLen \to \infty$ are integers. Consider
$\NumberJobs$ jobs with $\NumberJobs$ a multiple of $B$.
{\SimPS} completes batches of size $B$ every $\responseLen$ rounds, so the total flow time is
$$
\FTime{{\SimPS}}
\;=\; B\responseLen \sum_{j=1}^{\NumberJobs/B} j
\;=\; \frac{\responseLen \NumberJobs^2}{2B} + O(\NumberJobs \responseLen).
$$
On the other hand, by \Cref{lem:staggered:max-parallelism} with $\sliceLen=\responseLen$ and
$\initialLen=0$, the maximum feasible parallelism for a staggered schedule satisfies
$$
\paralSliceLen \;\ge\; \left\lfloor \frac{2\kvmem - \responseLen + 1}{\responseLen + 1} \right\rfloor
\;=\; \left\lfloor 2B - \frac{2B-1}{\responseLen + 1} \right\rfloor.
$$
Let $k \triangleq \paralSliceLen$ and consider {\SPS} with $(\paral,\sliceLen)=(k,\responseLen)$.
Using the standard bound on its total flow time,
$$
\FTime{\SPS}
\;\le\; \NumberJobs \responseLen + \frac{\responseLen}{2k}\NumberJobs(\NumberJobs-1)
\;=\; \frac{\responseLen \NumberJobs^2}{2k} + O(\NumberJobs \responseLen).
$$
Since $\OPT \le \FTime{\SPS}$, we obtain
$$
\frac{\FTime{{\SimPS}}}{\OPT}
\;\ge\; \frac{\FTime{{\SimPS}}}{\FTime{\SPS}}
\;\ge\; \frac{k}{B}\cdot (1 - \varepsilon)
\;\ge\; \left(2 - \frac{2}{\responseLen + 1} - \varepsilon\right).
$$
Letting $\responseLen \to \infty$ yields the desired ratio.
\end{proof}

\begin{proposition}[Lower bound approaching $3$ for {\SimPS}]
\label{prop:simps-lb-3}
There exists a family of identical-job instances for which the approximation ratio of {\SimPS} approaches
$3$ from below.
\end{proposition}
\begin{proof}
Fix $\initialLen = 0$ and let $\responseLen$ be a large multiple of $3$.
Choose an integer $\delta \ge 2$ and define a start-time spacing
$$
d \;\triangleq\; \frac{\responseLen}{3} + \delta.
$$
Set the memory budget to
$$
\kvmem \;\triangleq\; 3\responseLen - 3d + 3 \;=\; 2\responseLen - 3\delta + 3,
$$
so $\kvmem < 2\responseLen$ and {\SimPS} admits only one job at a time, i.e., $B=1$.
Thus
$$
\FTime{{\SimPS}} \;=\; \frac{\responseLen \NumberJobs(\NumberJobs+1)}{2}.
$$

Now consider the staggered schedule, $\mathcal{S}$, that starts job $i$ at time $i d$ and runs it for
$\responseLen$ rounds.
Because $d > \responseLen/3$, at most three jobs overlap at any time.
The peak memory usage occurs when the oldest active job is about to finish, at which time the
progress values of the three jobs are at most $\responseLen$, $\responseLen - d$, and $\responseLen - 2d$.
The total memory usage is therefore at most
$$
\responseLen + (\responseLen - d) + (\responseLen - 2d) + 3
\;=\; 3\responseLen - 3d + 3
\;=\; \kvmem,
$$
so the schedule is feasible.
Its total flow time is
$$
\FTime{\mathcal{S}}
\;=\; \NumberJobs \responseLen + d \cdot \frac{\NumberJobs(\NumberJobs-1)}{2}.
$$
Consequently,
$$
\frac{\FTime{{\SimPS}}}{\OPT}
\;\ge\; \frac{\FTime{{\SimPS}}}{\FTime{\mathcal{S}}}
\;\ge\; \frac{\responseLen}{d}\cdot (1 - \varepsilon)
\;=\; \frac{3}{1 + 3\delta/\responseLen}\cdot (1 - \varepsilon).
$$
Letting $\responseLen \to \infty$ with fixed $\delta$ shows that the ratio approaches $3$ from below.
\end{proof}

\begin{remark}The instance construction in \Cref{prop:simps-lb-3} corresponds to the worst-case scenario in the $4$-approximation analysis of {\SimPS} for identical jobs by \citet{JJMMPZ-25}. Regarding the memory-area bound, the average per-round memory usage of {\SimPS} approaches $\kvmem/4$ when $\kvmem$ is slightly less than $2\responseLen$, implying an approximation guarantee of nearly $4$. However, in this specific scenario, {\OPT} is also unable to fully utilize the available memory (e.g., because the spacing constraint \eqref{eq:spacing-condition} prevents perfect packing of triangular memory profiles); consequently, the actual approximation ratio approaches $3$ rather than $4$, leaving a gap between the upper and lower bounds for {\SimPS} on identical-job instances.
\end{remark}

\end{document}